\documentclass[12pt]{book} 
\usepackage{fancyhdr}
\usepackage{emptypage}
\usepackage{times}
\usepackage{setspace}
\usepackage[pdftex]{graphicx,color}
\usepackage{pst-node}
\usepackage{epsfig}
\usepackage{amssymb}
\usepackage{amsthm}
\usepackage{stmaryrd}
\usepackage{alltt}
\usepackage{mathpartir}

\let\RescueRightarrow=\Rightarrow
\usepackage{marvosym}
\renewcommand{\Rightarrow}{\RescueRightarrow}

\usepackage{mflogo}
\usepackage[hang,flushmargin]{footmisc} 
\usepackage[
  pagebackref=true,
  colorlinks=true,
  plainpages=false,
  backref=page,
  linkcolor=black,
  citecolor=black,
  urlcolor=black,
  pdfpagemode=None,
  pdftitle={The Complexity of Flow Analysis in Higher-Order Languages},
  pdfauthor={David Van Horn},
  baseurl={http://svn.lambda-calcul.us/dissertation/},
  pdfsubject={Computer science, PhD dissertation, Brandeis University},
  pdfkeywords={flow analysis, cfa, kcfa, complexity, lambda calculus, functional programming,
               linear logic, program analysis}]{hyperref}

\definecolor{darkred}{rgb}{.8,0,0}

\newcommand\papa{\em in memory of William Gordon Mercer\\ July 22, 1927--October 8, 2007}

\renewcommand*{\backref}[1]{}
\renewcommand*{\backrefalt}[4]{% 
  \ifcase #1 % 
   No citations.% 
  \or 
   Cited on page #2.% 
  \else 
   Cited on pages #2.% 
  \fi}

\newcommand\TB{\ \ \ \ }

\newcommand\tensor{\otimes}
\newcommand{\parr}{\bindnasrepma}

\newcommand\kcfa{\texorpdfstring{$k$}{k}CFA}

\newcommand\ptime{\texorpdfstring{{\sc ptime}}{PTIME}}
\newcommand\twonpda{\texorpdfstring{{\sc 2npda}}{2NPDA}}
\newcommand\pspace{\texorpdfstring{{\sc pspace}}{PSPACE}}
\newcommand\dtime{{\sc dtime}}
\newcommand\ntime{{\sc ntime}}
\newcommand\dspace{{\sc dspace}}
\newcommand\np{\texorpdfstring{{\sc nptime}}{NPTIME}}
\newcommand\nc{{\sc nc}}
\newcommand\exptime{\texorpdfstring{{\sc exptime}}{EXPTIME}}
\newcommand\logspace{\texorpdfstring{{\sc logspace}}{LOGSPACE}}

\newcommand\Lab{\mathbf{Lab}}
\newcommand\Env{\mathbf{Env}}

\newcommand\CEnv{\mathbf{CEnv}}
\newcommand\Cache{\mathbf{Cache}}
\newcommand\ACache{\mathbf{\widehat{Cache}}}
\newcommand\Var{\mathbf{Var}}
\newcommand\Val{\mathbf{Val}}
\newcommand\Exp{\mathbf{Exp}}
\newcommand\Unit{\mathbf{Unit}}
\newcommand\Term{\mathbf{Term}}
\newcommand\cache{\widehat{\mathsf{C}}}
\newcommand\ecache{\mathsf{C}}          %% Exact cache wears no hat.
\newcommand\eenv{\mathsf{r}}
\newcommand\aenv{\hat{\mathsf{r}}}

\newcommand\ce{\rho}
\newcommand\emptyenv{\bullet}
\newcommand\kmodels[2]{\models^{#1}_{#2}}
\newcommand\fv[1]{\ensuremath{\mathbf{fv}(#1)}}
\newcommand\bv[1]{\ensuremath{\mathbf{bv}(#1)}}

\newcommand\evalf[1][]{\ensuremath{\mathcal{E}{#1}}}
\newcommand\ivf{\ensuremath{\mathcal{I}}}
\newcommand\avf{\ensuremath{\mathcal{A}}}
\newcommand\eval[2][]{\ensuremath{\mathcal{E}{#1}\sem{#2}}}
\newcommand\iv[1]{\ensuremath{\mathcal{I}\sem{#1}}}
\newcommand\ev[3]{\ensuremath{\evalf\sem{#1}^{#2}_{#3}}}
\newcommand\av[4][]{\ensuremath{\avf{#1}\sem{#2}^{#3}_{#4}}}
\newcommand\avz[1]{\av[_0]{#1}\empty\empty}
\newcommand\avk[3]{\av[_k]{#1}{#2}{#3}}

\newcommand\dom[1]{\ensuremath{\mathbf{dom}(#1)}}
\newcommand\sem[1]{\ensuremath{\llbracket #1\rrbracket}}
\newcommand\lab[1]{\ensuremath{\mathbf{lab}(#1)}}
\newcommand\restrict{\ensuremath{\!\upharpoonright\!}}
\newcommand{\loves}{\ensuremath{\vdash}}
\newcommand\unknown{\ensuremath{\mathit{unknown}}}

\newcommand\expnd{\ensuremath{\Rightarrow}}
\newcommand\cfared{\ensuremath{\Rightarrow_{\mathsf{cfa}}}}
\newcommand\vp{\leadsto}
\newcommand\lolli{\multimap}

\newcommand\True{\ensuremath{\mbox{\tt True}}}
\newcommand\False{\ensuremath{\mbox{\tt False}}}
\newcommand\TT{\ensuremath{\mbox{\tt tt}}}
\newcommand\FF{\ensuremath{\mbox{\tt ff}}}
\newcommand\Not{\ensuremath{\mbox{\tt Not}}}
\newcommand\Copy{\ensuremath{\mbox{\tt Copy}}}
\newcommand\Implies{\ensuremath{\mbox{\tt Implies}}}
\newcommand\AND{\ensuremath{\mbox{\tt And}}}
\newcommand\Andgate{\ensuremath{\mbox{\tt Andgate}}}
\newcommand\Or{\ensuremath{\mbox{\tt Or}}}
\newcommand\Orgate{\ensuremath{\mbox{\tt Orgate}}}
\newcommand\Notgate{\ensuremath{\mbox{\tt Notgate}}}
\newcommand\Copygate{\ensuremath{\mbox{\tt Copygate}}}
\newcommand\Implgate{\ensuremath{\mbox{\tt Impliesgate}}}

\newcommand\Widget{\ensuremath{\mbox{\tt Widget}}}
\newcommand\Extract{\ensuremath{\mbox{\tt Extract}}}
\newcommand\Null{\ensuremath{\mbox{\tt Null}}}
\newcommand\Zero{\ensuremath{\mathbf 0}}  %% bits
\newcommand\One{\ensuremath{\mathbf 1}}

\newtheorem{theorem}{Theorem}
\newtheorem{lemma}{Lemma}
\newtheorem{corollary}{Corollary}
\newtheorem{proposition}{Proposition}
\newtheorem{definition}{Definition}

\newtheorem{conjecture}{Conjecture}

\usepackage{natbib}
\bibpunct();A{},
\let\cite=\citep
\title{The Complexity of Flow Analysis\\ 
       in Higher-Order Languages}
\author{David {Van Horn}}

\let\OrigTableofcontents=\tableofcontents
\renewcommand{\tableofcontents}{%
  \setlength\parskip{0in} 
  \OrigTableofcontents}
\let\OrigListoffigures=\listoffigures
\renewcommand{\listoffigures}{%
  \setlength\parskip{0in} 
  \OrigListoffigures
  \setlength\parskip{0.1in}}
\let\bdquotation\quotation
\let\bdendquotation\endquotation
\renewenvironment{quotation}%
  {\bdquotation%
     \noindent%
     \setlength\parindent{0pt}%
     \setlength\parskip{0.1in}}%
  {\bdendquotation}

\begin{document}
  \pagenumbering{roman}
% \maketitle

\thispagestyle{empty}%
\begin{center}
  \vspace*{.25in}%
  {\Huge \bf\baselineskip=.8\baselineskip 
    The Complexity of Flow Analysis in Higher-Order Languages \\}
  \vspace*{.5in}%
  \includegraphics[height=2in]{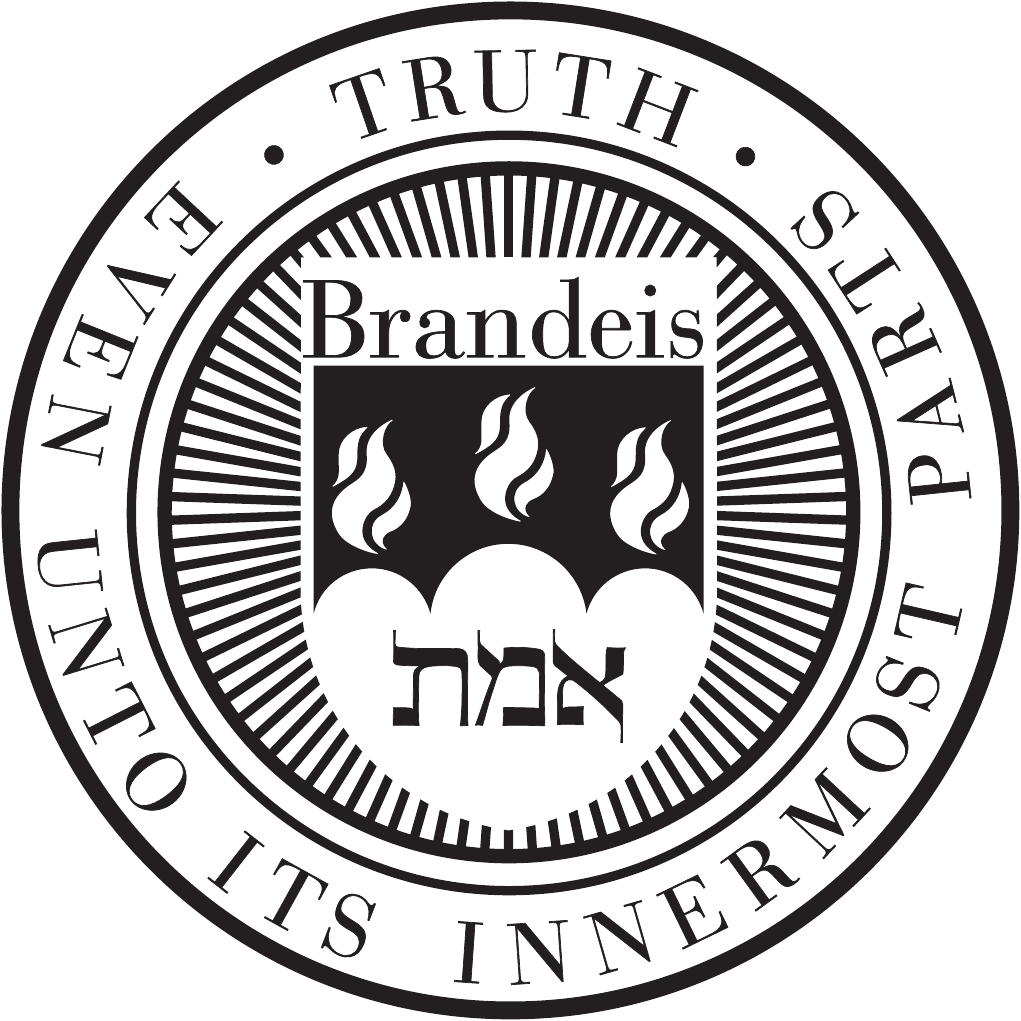}\\
  \vspace*{.5in}%
  {\large David Van Horn}
\end{center}
\cleardoublepage

\thispagestyle{empty}%
\begin{center}
  \vspace*{.25in}%
  {\Huge \bf\baselineskip=.8\baselineskip 
    The Complexity of Flow Analysis in Higher-Order Languages \\}
  \vspace*{.5in}%
  A Dissertation\\
  \vspace*{.25in}%
  Presented to\\
  The Faculty of the Graduate School of Arts and Sciences\\
  Brandeis University\\
  Mitchom School of Computer Science\\
%  Harry G. Mairson, Advisor\\
  \vspace*{.5in}
  In Partial Fulfillment\\
  of the Requirements for the Degree\\
  Doctor of Philosophy\\
  \vspace*{.5in}
  by\\
  David Van Horn\\
  August, 2009\\
\end{center}%
\cleardoublepage

\thispagestyle{empty}%
%\begin{flushleft}%
  \vspace*{.15in}%
  This dissertation, directed and approved by David Van Horn's
  committee, has been accepted and approved by the Graduate
  Faculty of Brandeis University in partial fulfillment of the
  requirements for the degree of:
  \vspace*{0.25in}%
  \begin{flushright}%
    \centerline{\bfseries DOCTOR OF PHILOSOPHY}%
    \par
    \parbox{3.6in}{%
      \vspace{.8in}%
      Adam B. Jaffe, Dean of Arts and Sciences}%
  \end{flushright}%
  \vspace*{0.4in}%
  Dissertation Committee:
  \par
  Harry G. Mairson, Brandeis University, Chair\\
  Olivier Danvy, University of Aarhus\\
  Timothy J. Hickey, Brandeis University\\
  Olin Shivers, Northeastern University  
%\end{flushleft}
\cleardoublepage

\thispagestyle{empty}%
\vspace*{2.5in}%
\begin{center}\copyright\ David Van Horn, 2009\\
Licensed under the Academic Free License version 3.0.
\end{center}
\cleardoublepage

\thispagestyle{empty}%
\vspace*{2.5in}%
\begin{center}\papa\end{center}
\cleardoublepage

\chapter*{Acknowledgments}
\addcontentsline{toc}{chapter}{Acknowledgments}
Harry taught me so much, not the least of which was a compelling kind
of science.

It is fairly obvious that I am not uninfluenced by Olivier Danvy and
Olin Shivers and that I do not regret their influence upon me.

My family provided their own weird kind of emotional support and
humor.

I gratefully acknowledge the support of the following people, groups,
and institutions, in no particular order: Matthew Goldfield.  Jan
Midtgaard.  Fritz Henglein.  Matthew Might.  Ugo Dal Lago.
Chung-chieh Shan.  Kazushige Terui.  Christian Skalka.  Shriram
Krishnamurthi.  Michael Sperber.  David McAllester.  Mitchell Wand.
Damien Sereni.  Jean-Jacques L\'evy.  Julia Lawall.  Matthias
Felleisen.  Dimitris Vardoulakis.  David Herman.  Ryan Culpepper.
Richard Cobbe.  Felix S Klock II.  Sam Tobin-Hochstadt.  Patrick
Cousot.  Alex Plotnick.  Peter M{\o}ller Neergaard.  Noel Welsh.
Yannis Smaragdakis.  Thomas Reps.  Assaf Kfoury.  Jeffrey Mark
Siskind.  David Foster Wallace.  Timothy J.~Hickey.  Myrna Fox.  Jeanne
DeBaie.  Helene Greenberg.  Richard Cunnane.  Larry Finkelstein.  Neil
D.~Jones and the lecturers of the Program Analysis and Transformation
Summer School at DIKU. New England Programming Languages and Systems
Symposium.  IBM Programming Language Day.  Northeastern University
Semantics Seminar and PL Jr.~Seminar series.  The reviewers of
ICFP'07, SAS'08, and ICFP'08.  Northeastern University.  Portland
State University. The National {\it Science} Foundation, grant
CCF-0811297.

And of course, Jessie.  All the women in the world aren't whores,
just mine.

\chapter*{Abstract}
\addcontentsline{toc}{chapter}{Abstract}
%% - 350 words or less : OK

This dissertation proves lower bounds on the inherent difficulty of
deciding flow analysis problems in higher-order programming languages.
We give exact characterizations of the computational complexity of
0CFA, the $k$CFA hierarchy, and related analyses. 
In each case, we precisely capture both the {\em expressiveness} and
{\em feasibility} of the analysis, identifying the elements
responsible for the trade-off.

0CFA is complete for polynomial time. 
This result relies on the insight that when a program is linear (each
bound variable occurs exactly once), the analysis makes no
approximation; abstract and concrete interpretation coincide, and
therefore program analysis becomes evaluation under another guise.
Moreover, this is true not only for 0CFA, but for a number of {\em
  further approximations} to 0CFA.  In each case, we derive polynomial
time completeness results.

For any $k > 0$, $k$CFA is complete for exponential time.  
Even when $k = 1$, the distinction in binding contexts results in a
limited form of {\em closures}, which do not occur in 0CFA.
This theorem validates empirical observations that $k$CFA is
intractably slow for any $k > 0$. 
There is, in the worst case---and plausibly, in practice---no way to
tame the cost of the analysis. 
Exponential time is required.  
The empirically observed intractability of this analysis can be
understood as being {\em inherent in the approximation problem being
  solved}, rather than reflecting unfortunate gaps in our programming
abilities.

\chapter*{Preface}
\addcontentsline{toc}{chapter}{Preface}
\markboth{\slshape \MakeUppercase{Preface}}{}

\subsection*{What to expect, What not to expect}

This dissertation investigates lower bounds on the computational
complexity of flow analysis for higher-order languages, uncovering its
inherent computational costs and the fundamental limits of efficiency
for {\em any} flow analysis algorithm.  As such, I have taken
existing, representative, flow analysis specifications ``off the
shelf'' without modification.  This is {\em not} a dissertation on the
design and implementation of novel flow analyses (although it should
inform such work).  The reader is advised to expect no benchmarks or
prototype implementations, but rather insightful proofs and theorems.

This dissertation relates existing research in order to situate the
novelty and significance of this work.  It does not attempt to
comprehensively survey the nearly thirty years of research on flow
analysis, nor the wealth of frameworks, formulations, and variants.  A
thorough survey on flow analysis has been undertaken by
\citet{midtgaard-07}.

%% Although this is a standard approach to giving one's dissertation the
%% girth that marks a real swinging dick academic, I have avoided it.  In
%% part, this is done out of laziness, but more importantly this is done
%% because a thorough survey on the topic has been undertaken by
%% \citet{midtgaard-07}.

\subsection*{Assumptions on the reader}

%% First and foremost, I am a programmer.  The extent to which I
%% understand logic or math, is---more or less---the extent to which I
%% can make sense of these things {\em via} a programming perspective.

For the reader expecting to understand the intuitions, proofs, and
consequences of the results of this dissertation, I assume familiarity
with the following, in roughly descending order of importance:
\begin{itemize}

\item functional programming.  

%   Proficiency in fundamental programming skills and idioms such as
%   found in {\em How to Design Programs} by \citet{htdp} are assumed
%   and 
  The reader should be at ease programming with higher-order
  procedures in languages such as Scheme or ML.  For an introduction
  to programming in Scheme specifically, {\em The Scheme Programming
    Language} by \citet{dybvig-tspl} and {\em Teach Yourself Scheme in
    Fixnum Days} by \citet{sitaram-tyscheme} are recommended; for ML,
  {\em Programming in Standard ML} by \citet{harper-sml} and {\em ML
    for Working Programmer} by \citet{paulson} are recommended.

%% \footnote{To the
%% frugal scholar: full texts of these books are available online.
%% \begin{itemize}
%% \item \url{http://www.htdp.org/}
%% \item \url{http://www.scheme.com/tspl3/}
%% \item \url{http://www.ccs.neu.edu/home/dorai/t-y-scheme/t-y-scheme.html}
%% \item \url{http://www.cs.cmu.edu/~rwh/smlbook/}
%% \end{itemize}}
This dissertation relies only on the simplest applicative subsets of
these languages.

\item interpreters (evaluators).

  The reader should understand the fundamentals of writing an
  interpreter, in particular an environment-based interpreter
  \cite{landin-64} for the functional core of a programming
  language.\footnote{Note that almost every modern programming
    language includes a higher-order, functional core: Scheme, ML,
    JavaScript, Java, Haskell, Smalltalk, Ruby, C\#, etc., etc.}  The
  definitive reference is ``Definitional interpreters for higher-order
  programming languages'' by \citet{reynolds-acm72,reynolds-hosc98}.
  Understanding sections 2--6 are an absolute must (and joy).  For a
  more in-depth textbook treatment, see the gospel according to
  \citet{sicp}: {\em Structure and Interpretation of Computer
    Programs}, Chapter 3, Section 2, ``The Environment Model of
  Evaluation,'' and Chapter 4, Section 1, ``The Metacircular
  Evaluator.''
%% \footnote{Again, to the frugal scholar: the full
%% text of this book is also freely available online.
%% \begin{itemize}
%% \item \url{http://mitpress.mit.edu/sicp/}
%% \end{itemize}}
Finally, {\em Essentials of Programming Languages} by \citet{eopl3} is
highly recommended.\footnote{As an undergraduate, I cut my teeth on
  the first edition \citeyearpar{eopl1}.  
  % It was my first glimpse of beauty in this field and I have been in
  % love with the topic ever since.
}

\item the $\lambda$-calculus.  

  The encyclopedic reference is {\em The Lambda Calculus: Its Syntax
    and Semantics} by \citet{barendregt}, which is an overkill for the
  purpose of understanding this dissertation.  Chapters 1 and 3 of
  {\em Lectures on the Curry-Howard Isomorphism} by
  \citet{sorensen-urzyczyn} offers a concise and sufficient
  introduction to untyped and typed $\lambda$-calculus, respectively.
  There are numerous others, such as {\em An Introduction to Lambda
    Calculi for Computer Scientists} by \citet{hankin-lambda}, {\em
    Functional programming and lambda calculus} by
  \citet{barendregt-volb}, and so on.  Almost any will
  do.\footnote{See \citet[Footnote 1]{cardone-hindley-06} for
    references to French, Japanese, and Russian overviews of the
    $\lambda$-calculus.}

\item basic computational complexity theory.   

  The reader should be familiar with basic notions such as complexity
  classes, Turing machines, undecidability, hardness, and complete
  problems.  \citet{Papadimitriou94} is a standard introduction (See
  chapters 2--4, 7--9, 15, and 16).  \citet{jones97} is a good
  introduction with a stronger emphasis on programming and programming
  languages (See part IV and V). Almost any decent undergraduate text
  on complexity would suffice.

  In particular, the classes \logspace, \ptime, \np, and \exptime\ are
  used.  Reductions are given from canonical complete problems for
  these classes to various flow analysis problems.  These canonical
  complete problems include, respectively: the permutation problem,
  circuit value problem (CVP), Boolean satisfiability (SAT), and a
  generic reduction for simulating deterministic, exponential time
  Turing machines.

  Such a reduction from a particular complete computational problem to
  a corresponding flow analysis problem establishes a {\em lower
    bound} on the complexity of flow analysis: solving the flow
  problem is {\em at least as hard} as solving the corresponding
  computational problem (SAT, CVP, etc.), since any instance of these
  problems can be transformed (reduced), using very limited resources,
  to an instance of the flow analysis problem.  In other words, an
  algorithm to solve one problem can be used as an algorithm to solve
  the other.

%% answer for one may be transformed
%% (reduced) to an answer for the other.

%% Ch 2 (TM), 3 (Computability), 4 (Boolean logic), 7 (Relations between
%% complexity classes), 8 (Reductions and completeness), 9 (NP-complete
%% problems), 15 (Parallel computation), 16 (Logarithmic space).

\item fundamentals of program analysis.

  A basic understanding of program analysis would be beneficial,
  although I have tried to make plain the connection between analysis
  and evaluation, so a thorough understanding of program {\em
    interpretation} could be sufficient.  Perhaps the standard text on
  the subject is {\em Principles of Program Analysis} by
  \citet{nielson-nielson-hankin}, which I have followed closely
  because it offers an authoritative and precise definition of flow
  analysis.  It is thorough and rigorous, at the risk of slight {\em
    rigor mortis}.  Shivers' dissertation, {\em Control-Flow Analysis
    of Higher-Order Languages}, contains the original development of
  $k$CFA and is replete with intuitions and explanations.

\item logic and proof theory. 

  The reader should be comfortable with the basics of propositional
  logic, such as De Morgan duality, modus ponens, etc.  The reader is
  also assumed to be comfortable with sequent calculi, and in
  particular sequents for linear logic.
  \citet{girard-proofs-and-types} provides a solid and accessible
  foundation.

All of the theorems will be accessible, but without this background,
only a small number of the more supplemental proofs will be
inaccessible.  Fear not if this is not your cup of meat.

\end{itemize}

%% For the reader expecting to further this line of research using the
%% approach and toolkit employed here, there are significantly more
%% assumptions.

%% \begin{itemize}
%% \item Sharing graphs, optimal reduction, and geometry of interaction.
%% \item Linear logic and proof-nets.
%% \item First-class control.
%% \item Constructive classical logic.
%% \item Game semantics.
%% \end{itemize}

\subsection*{Previously published material}

Portions of this dissertation are derived from material previously
published in the following papers, written jointly with Harry Mairson:
\begin{enumerate}
\item Relating Complexity and Precision in Control Flow Analysis.  In
{\em Proceedings of the 12th International Conference on Functional
Programming}, Frieburg, Germany. \citet{vanhorn-mairson-icfp07}.

\item Flow Analysis, Linearity, and PTIME. In {\em The 15th
International Static Analysis
Symposium}, Valencia, Spain. \citet{vanhorn-mairson-sas08}.

\item Deciding $k$CFA is complete for EXPTIME. In {\em Proceedings of
the 13th International Conference on Functional Programming},
Victoria, BC, Canada. \citet{vanhorn-mairson-icfp08}.

\end{enumerate}

\cleardoublepage
\phantomsection
\addcontentsline{toc}{chapter}{Contents}
\tableofcontents
\cleardoublepage
\phantomsection
\addcontentsline{toc}{chapter}{List of Figures}
\listoffigures

\pagenumbering{arabic}
\chapter{Introduction}
\label{chap:introduction}

We analyze the computational complexity of flow analysis for
higher-order languages, yielding a number of novel insights: $k$CFA is
provably intractable; 0CFA and its approximations are inherently
sequential; and analysis and evaluation of linear programs are
equivalent.

\section{Overview}
\label{sec-overview}
%% Danvy, ``On Scientific Writing'': The introduction, It should start
%% with a bang.
Predicting the future is hard.

Nevertheless, such is the business of an optimizing compiler: it reads
in an input program, predicts what will happen when that program is
run, and then---based on that prediction---outputs an optimized
program.

Often these predictions are founded on semantics-based program
analysis \cite{cousot-cousot-popl77, cousot-cousot-jlc92,
  muchnick-jones-81, nielson-nielson-hankin}, which aims to discover
the run-time behavior of a program {\em without actually running it}
\cite[page xv]{muchnick-jones-81}.  But as a natural consequence of
Rice's theorem \citeyearpar{rice}, a perfect prediction is almost
always impossible.  So {\em tractable} program analysis must
necessarily trade exact evaluation for a safe, computable
approximation to it.  This trade-off induces a fundamental dichotomy
at play in the design of program analyzers and optimizing compilers.
On the one hand, the more an analyzer can discover about what will
happen when the program is run, the more optimizations the compiler
can perform.  On the other, compilers are generally valued not only
for producing fast code, but doing so quickly and efficiently; some
optimizations may be forfeited because the requisite analysis is too
difficult to do in a timely or space-efficient manner.

% We like programming in very expressive languages, where powerful
% computations can be invoked.  We not only allow, but often expect that
% programs can run for a long time.  In contrast, compile-time static
% analysis is valued for being fast, and resource-efficient. Yet, in the
% face of these constraints, it must remain {\em expressive}, able to
% predict rich and useful properties of programs.  This disparity begs
% for a comprehensive, precise answer to the following question:
% \begin{quotation}\it\noindent
% What can we possibly learn about the
% runtime behavior of a computation-intensive program with a
% resource-bounded static analyzer?
% \end{quotation}

As an example in the extreme, if we place {\em no limit} on the
resources consumed by the compiler, it can perfectly predict the
future---the compiler can simply simulate the running of the program,
watching as it goes.  When (and if) the simulation completes, the
compiler can optimize with perfect information about what will happen
when the program is run.  With good reason, this seems a bit like
cheating.

So at a minimum, we typically expect a compiler will eventually finish
working and produce an optimized program.  (In other words, we expect
the compiler to compute within bounded resources of time and space).
%\footnote{This expectation is not typical for Scheme programmers.} 
After all, what good is an optimizing compiler that never finishes?

But by requiring an analyzer to compute within bounded resources, we
have necessarily and implicitly limited its ability to predict the
future.

As the analyzer works, it {\em must} use some form of approximation;
knowledge must be given up for the sake of computing within bounded
resources.  Further resource-efficiency requirements may entail
further curtailing of knowledge that a program analyzer can discover.
%
% For example, what kind of properties can we discover if the analyzer
% is allowed only a fixed number of pointers, or just a polynomial
% amount of time?
%
But the relationship between approximation and efficiency is far from
straightforward.  Perhaps surprisingly, as has been observed
empirically by researchers
\cite{wright-jagannathan-toplas98,jagannathan-etal-popl98,might-shivers-icfp06},
added precision may avoid needless computation induced by
approximation in the analysis, resulting in computational {\em
  savings}---that is, better information can often be produced faster
than poorer information.  So what exactly is the analytic relationship
between forfeited information and resource usage for any given design
decision?

In trading exact evaluation for a safe, computable approximation to
it, analysis negotiates a compromise between complexity and precision.
But what exactly are the trade-offs involved in this negotiation?  For
any given design decision, what is given up and what is gained?  What
makes an analysis rich and expressive?  What makes an analysis fast
and resource-efficient?

We examine these questions in the setting of {\em flow analysis}
\cite{jones81,sestoft-ms,sestoft-fpca89,shivers-pldi88,shivers-phd,midtgaard-07}, a
fundamental and ubiquitous static analysis for higher-order
programming languages.  It forms the basis of almost all other
analyses and is a much-studied component of compiler technology.

Flow analysis answers basic questions such as ``what functions can be
applied?,'' and ``to what arguments?''  These questions specify
well-defined, significant {\em decision problems}, quite apart from
any algorithm proposed to solve them.  This dissertation examines the
inherent computational difficulty of deciding these problems.

% \dvhnote{Discuss the essence of complexity classes and hardness
% results. (Maybe put in prelims).}

% Control-flow analysis may be feasible and useful for higher-order languages

% Shivers' seminal dissertation, {\em Control-Flow Analysis of
% Higher-Order Languages} \citeyearpar{shivers-phd}, opens with the
% thesis:
% \begin{quote}
% Control-flow analysis is feasible and useful for higher-order languages. 
% \end{quote}
% The subsequent pages explain how flow analysis is computable
% (feasible) and enabling (useful).  But what is the {\em relationship}
% between these two properties?

% In answering the flow analysis question ``does this value flow into
% this point?,'' the most approximate analysis will always answer {\em
% yes}, which takes no resources to compute---and is of little use.  On
% the other hand, the most precise analysis will answer {\em yes} if and
% only if the given value flows into the program point during
% evaluation, which is useful, albeit uncomputable.  In mediating
% between these extremes, every static analysis necessarily gives up
% valuable information for the sake of computing an answer within
% bounded resources.  

% Designing a static analyzer, therefore, requires making trade-offs
% between {\em precision} and {\em complexity}.

% In other words:
% \begin{center}
% {\em What are the computationally potent ingredients in a
% static analysis?}
% \end{center}

If we consider the most useful analysis the one which yields complete
and perfectly accurate information about the running of a program, then
clearly this analysis is intractable---it consumes the same
computational resources as running the program.  At the other end of
the spectrum, if the least useful analysis yields no information about
the running of a program, then this analysis is surely feasible, but
useless.

% In the more than 25 years of research on the subject of flow analysis,
% there has remained a poverty of analytic knowledge concerned with the
% relationship between complexity and precision.  

If the design of software is really a science, we have to understand
the trade-offs between the running time of static analyzers, and the
accuracy of their computations. 

There is substantial empirical experience, which gives a partial
answer to these questions.  However, despite being the fundamental
analysis of higher-order programs, despite being the subject of
investigation for over twenty-five years, and the great deal of
expended effort deriving clever ways to tame the cost, there has
remained a poverty of analytic knowledge on the complexity of flow
analysis, the essence of how it is computed, and where the sources of
approximation occur that make the analysis work.  

This dissertation is intended to repair such lacunae in our
understanding.

%% That knowledge can be used to build better analyzers.

%% They express a goal of analytically relating such techniques,
%% appealing to the relations emphasized by computational complexity,
%% where problems are related via the idea of reduction (which is no more
%% than a carefully resource-bounded compiler).

\section{Summary of Results}

\begin{itemize}
\item Normalization and analysis are equivalent for linear programs.

\item 0CFA and other monovariant flow analyses are complete for \ptime.

\item 0CFA of typed, $\eta$-expanded programs is complete for \logspace.

\item $k$CFA is complete for \exptime\ for all $k>0$.
\end{itemize}

\section{Details}

\subsection{Linearity, Analysis and Normalization}
\begin{itemize}
\item Normalization and analysis are equivalent for linear programs.
\end{itemize}

Although variants of flow analysis abound, we identify a core
language, the linear $\lambda$-calculus, for which all of these
variations coincide.  In other words, for {\em linear}
programs---those written to use each bound variable exactly once---all
known flow analyses will produce equivalent information.  

It is straightforward to observe that in a {\em linear}
$\lambda$-term, each abstraction $\lambda x.e$ can be applied to at
most one argument, and hence the abstracted value can be bound to at
most one argument.
% \footnote{Note that this observation is clearly
%   untrue for the {\em nonlinear} $\lambda$-term $(\lambda f.f (a (f
%   b))) (\lambda x.x)$, as $x$ is bound to $b$, and also to $ab$.}
Generalizing this observation, analysis of a linear $\lambda$-term
coincides exactly with its evaluation.  So not only are the varying
analyses equivalent to each other on linear terms, they are all
equivalent to evaluation.

Linearity is an equalizer among variations of static analysis, and a
powerful tool in proving lower bounds.

\subsection{Monovariance and \ptime}
\begin{itemize}
\item 0CFA and other monovariant flow analyses are complete for \ptime.
\end{itemize}

By definition, a {\em monovariant} analysis (e.g.~0CFA), does not
distinguish between occurrences of the same variable bound in
different calling contexts.  But the distinction is needless for
linear programs and analysis becomes evaluation under another name.
This opens the door to proving lower bounds on the complexity of the
analysis by writing---to the degree possible---computationally
intensive, linear programs, which will be faithfully executed by the
analyzer rather than the interpreter.

We rely on a symmetric coding of Boolean logic in the linear
$\lambda$-calculus to simulate circuits and reduce the 0CFA decision
problem to the canonical \ptime\ problem, the circuit value problem.
This shows, since the inclusion is well-known, that 0CFA is complete
for \ptime.  Consequently, 0CFA is inherently sequential and there is
no fast parallel algorithm for 0CFA (unless \ptime\ $=$ \nc).
Moreover, this remains true for a number of {\em further
  approximations} to 0CFA.

The best known algorithms for computing 0CFA are often not practical
for large programs. Nonetheless, information can be given up in the
service of quickly computing a necessarily less precise analysis.  For
example, by forfeiting 0CFA's notion of directionality, algorithms for
Henglein's simple closure analysis run in near linear time
\citeyearpar{henglein92d}.  Similarly, by explicitly bounding the
number of passes the analyzer is allowed over the program, as in
Ashley and Dybvig's sub-0CFA \citeyearpar{ashley-dybvig-toplas98}, we
can recover running times that are linear in the size of the program.
But the question remains: Can we do better?  For example, is it
possible to compute these less precise analyses in logarithmic space?
We show that without profound combinatorial breakthroughs (\ptime\ $=$
\logspace), the answer is no.  Simple closure analysis, sub-0CFA, and
other analyses that approximate or restrict 0CFA, {\em require}---and
are therefore, like 0CFA, complete for---polynomial time.

\subsection{0CFA with \texorpdfstring{$\eta$}{Eta}-Expansion and \logspace}
\begin{itemize}
\item 0CFA of typed, $\eta$-expanded programs is complete for \logspace.
\end{itemize}

We identify a restricted class of functional programs whose 0CFA
decision problem may be simpler---namely, complete for
\logspace.  Consider programs that are simply typed, and where a
variable in the function position or the argument position of an
application is fully $\eta$-expanded.  This case---especially, but not
only when the programs are linear---strongly resembles multiplicative
linear logic with {\em atomic} axioms. 

We rely on the resemblance to bring recent results on the complexity
of normalization in linear logic to bear on the analysis of
$\eta$-expanded programs resulting in a \logspace-complete variant of
0CFA.

\subsection{\kcfa\ and \exptime}
\begin{itemize}
\item $k$CFA is complete for \exptime\ for all $k>0$.
\end{itemize}

We give an exact characterization of the computational complexity of
the $k$CFA hierarchy.  For any $k > 0$, we prove that the control flow
decision problem is complete for deterministic exponential time.  This
theorem validates empirical observations that such control flow
analysis is intractable.  It also provides more general insight into
the complexity of abstract interpretation.

A fairly straightforward calculation shows that $k$CFA can be computed
in exponential time.  We show that the naive algorithm is essentially
the best one.  There is, in the worst case---and plausibly, in
practice---no way to tame the cost of the analysis.  Exponential time
is required.

\chapter{Foundations}
\label{chap:foundations}

The aim of flow analysis is to safely approximate an answer the
question:\footnote{See for example the transparencies accompanying
  Chapter 3 ``Control Flow Analysis'' of
  \citet{nielson-nielson-hankin}:
  \url{http://www2.imm.dtu.dk/~riis/PPA/ppa.html}}
\begin{center}\it
For each function application, which functions may be applied?
\end{center}

% \begin{itemize}
% \item To get the most precise answer, simply run the program and find out.

% \item To get the least precise answer, don't do anything: the answer
% is {\em all of them}.  

% \item To get an answer somewhere in between, do something in between.
% \end{itemize}

% This is an intentionally oversimplified view of program analysis, but
% it demonstrates that 
Analysis can easily be understood as the safe approximation to program
evaluation.  It makes sense, then, to first consider evaluation in
detail.  In the following sections, an evaluation function
($\mathcal{E}$) is defined, from which an instrumented variation
($\mathcal{I}$) is derived and abstracted to obtain the abstract
evaluator ($\mathcal{A}$).  Finally, we review basic concepts from
complexity theory and sketch our approach to proving lower bounds.

\section{Structure and Interpretation}

The meaningful phrases of a program are called {\em expressions}, the
process of executing or interpreting these expressions is called {\em
evaluation}, and the result of evaluating an expression is called a
{\em value} \cite{reynolds-acm72}.

We will consider a higher-order applicative programming language based
on the $\lambda$-calculus, in which evaluation is environment based
and functional values are represented using closures.
The syntax of the language is given by the following grammar:
\begin{displaymath}
\begin{array}{l@{\quad}l@{\quad}l}
\Exp & e ::= x\ |\ e\;e\ |\ \lambda x.e & \mbox{expressions}
\end{array}
\end{displaymath}

\begin{quotation}%
% \begin{singlespace}
{\bf Note to the reader:} This may seem a rather minimal programming
language and you may wonder what the broader applicability of these
results are in the face of other language features.  But as noted
earlier, this dissertation is concerned with {\em lower bounds} on
static analysis.  By examining a minimal, core language, all results
will immediately apply to any language which includes this core.  In
other words, the more restricted the subject language, the broader the
applicability.

It may be that the lower bound can be improved in the presence of some
language features, that is, the given set of features may make
analysis provably harder, but it certainly can not make it any
easier.\footnote{This is to be contrasted with, for example, a type
  soundness theorem, where it is just the opposite: adding new language
  feature may revoke soundness.

  For similar remarks, see the discussion in section 2 of
  \citet{reps-ai96} concerning the benefits of formulating an analysis
  problem in ``trimmed-down form,'' which not only leads to a wider
  applicability of the lower bounds, but also ``allows one to gain
  greater insight into exactly what aspects of an
  interprocedural-analysis problem introduce what computational
  limitations on algorithms for these problems.''

  In other words, lower bounds should be derived not by piling feature
  on top of feature, but by removing the weaknesses and restrictions
  that make additional features appear necessary.}
% \end{singlespace}
\end{quotation}

Following \citet{landin-64}, substitution is modelled using {\em
environments}.  Procedures will be represented as {\em closures}, a
$\lambda$-term together with its lexical environment, which closes
over the free variables in the term, binding variables to values.

We use $\rho$ to range over {\em environments} (mappings from
variables to closures), and $v$ to range over {\em closures} (pairs
consisting of a term and an environment that closes the term).  The
empty environment (undefined on all variables) is denoted $\emptyenv$,
and we occasionally abuse syntax and write the closed term $e$ in
place of the closure $\langle e,\emptyenv\rangle$.  Environment extension
is written $\rho[x\mapsto v]$ and we write $\emptyenv[x\mapsto v]$ as
$[x\mapsto v]$ and $[x_1\mapsto v_1]\dots[x_n\mapsto v_n]$ as
$[x_1\mapsto v_1,\dots,x_n\mapsto v_n]$.

\begin{displaymath}
\begin{array}{l@{\quad}l@{\quad}l}
\Env & \rho \in \Var \rightharpoonup \Val & \mbox{environments}\\
\Val & v \in \langle\Exp,\Env\rangle & \mbox{closures}
\end{array}
\end{displaymath}

The meaning of an expression is given in terms of an evaluation
function, or interpreter.  Following \citet[Chapter 4,
``Metalinguistic Abstraction'']{sicp}, an interpreter is defined as
follows:
\begin{quotation}%
An {\em evaluator} (or {\em interpreter}) for a programming language
is a procedure that, when applied to an expression of the language,
performs the actions required to evaluate that expression.
\end{quotation}

\begin{figure}
\begin{displaymath}
\begin{array}{lcl}
\evalf & : & \Exp \times \Env \rightharpoonup \Val\\
\\
\eval{x}\rho & = & \rho(x)\\
\eval{\lambda x.e}\rho & = & \langle\lambda x.e,\rho'\rangle\\
\ &\ & \quad\mbox{where } \rho' = \rho \restrict \fv{\lambda x.e}\\
\eval{e_1 e_2}\rho & = &
       \mbox{\bf let }\langle\lambda x.e_0,\rho'\rangle = 
                      \eval{e_1}\rho \mbox{\bf\ in }\\   %% \rho\restrict\fv{e_1}
\ &\ & \mbox{\bf let }v = 
                      \eval{e_2}\rho \mbox{\bf\ in }\\   %% \rho\restrict\fv{e_2}
\ &\ &   \quad   \eval{e_0}{\rho'[x \mapsto v]}
\end{array}
\end{displaymath}
\caption{Evaluator $\evalf$.}
\label{figure-eval}
\end{figure}

The evaluation function for the language is given in
\autoref{figure-eval}.  We say $e$ evaluates to $v$ under environment
$\rho$ iff $\eval{e}\rho = v$ and computing the evaluation function
defines our notion of the ``running of a program.''  Some examples of
evaluation:
\begin{eqnarray*}
\eval{\lambda x.x}{\emptyenv} & = & \langle \lambda x.x, \emptyenv\rangle\\
\eval{(\lambda x.\lambda z. x) (\lambda y.y)}\emptyenv & = & \langle \lambda z.x,[x\mapsto \langle \lambda y.y,\emptyenv\rangle]\rangle\\
\eval{(\lambda f.ff(\lambda y.y))(\lambda x.x)}\emptyenv & = & \langle\lambda y.y,\emptyenv\rangle
\end{eqnarray*}

This gives us a mostly {\em extensional} view of program
behaviour---evaluation maps programs to values, but offers little
information regarding {\em how} this value was computed.  For the sake
of program optimization it is much more useful to know about
operational (or {\em intensional}) properties of programs.  These
properties are formulated by appealing to an ``instrumented
interpreter,'' which is the subject of the next section.  Intuitively,
the instrumented evaluator works just like the un-instrumented
evaluator, but additionally maintains a complete history of the
operations carried out during the running of the program.

%% We choose an evaluation strategy of call-by-value.\footnote{These
%% choices may seem arbitrary, but the evaluator is in fact they an
%% artifact derived from the program analyses we chose to investigate.}

%% \dvhnote{Discuss {\em extensional} aspect and the extensional property
%% of evaluation.  See chapter 6 of \cite{sampath-phd} for a good
%% discussion.  Compare this later to intensional aspects: the
%% intrumented interpreter, exact cfa, and finally, its approximation.}

\section{Instrumented Interpretation}
\label{section-instrumented}

Instrumented (or concrete) interpretation carries out the running of
program while maintaining a trace of the operations performed, thus
providing an operational history of evaluation.  A suitable
reformulation of the original definition of an evaluator to
incorporate instrumentation is then:

\begin{quotation}%
An {\em instrumented evaluator} (or {\em instrumented interpreter})
for a programming language is a procedure that, when applied to an
expression of the language, performs {\em and records} the actions
required to evaluate that expression.
\end{quotation}

Exactly which actions should be record will vary the domain of any
given static analysis and there is no universal notion of a program
trace, but for flow analysis, the interesting actions are:

\begin{itemize}
\item Every time the value of a subexpression is computed, record its
value and the context in which it was evaluated.

\item Every time a variable is bound, record the value and context in
which it was bound.
\end{itemize}

These actions are recorded in a {\em cache}, and there is one for each
kind of action:
\begin{eqnarray*}
\ecache & : & \Lab \times \Delta \rightharpoonup \Val \\
\eenv   & : & \Var \times \Delta \rightharpoonup \Val\\
\Cache & = & (\Lab \times \Delta \rightharpoonup \Val) \times
             (\Var \times \Delta \rightharpoonup \Val)\\
\end{eqnarray*}
The $\ecache$ cache records the result, or returned value, of each
subcomputation, and the $\eenv$ cache records each binding of a
variable to a value.  Given the label of a subexpression ($\Lab$) and
a description of the context ($\Delta$), $\ecache$ returns the value
produced by that subexpression evaluated in that context.  Given the
name of a variable ($\Var$) and a description of the context
($\Delta$), $\eenv$ returns the value bound to that variable in that
context.  The $\ecache$ cache is a partial function since a
subexpression 1) may not produce a value, it may diverge, or 2) may
not be evaluated in the given context.  The $\eenv$ cache is partial
for analogous reasons.

The set $\Lab$ is used to index subexpressions.  It can easily be
thought of as the implicit source location of the expression, but our
formalism will use an explicit labelling scheme.  We use $\ell$ to
range over labels.  The syntax of the source language is given by the
following grammar, and programs are assumed to be uniquely labelled:
\begin{displaymath}
\begin{array}{l@{\quad}l@{\quad}l}
\Exp & e ::= t^\ell & \mbox{expressions (or labeled terms)}\\
\Term & t ::= x\ |\ (e\;e)\ |\ (\lambda x.e) & \mbox{terms (or unlabeled expressions)}
\end{array}
\end{displaymath}

Irrelevant labels are frequently omitted for presentation purposes.

The set $\Delta$ consists of {\em contours}, which are strings of
labels from application nodes in the abstract syntax of the program.
A string of application labels describes the context under which the
term evaluated.

A variable may be bound to any number of values during the course of
evaluation.  Likewise, a subexpression that occurs once syntactically
may evaluate to any number of values during evaluation.  So asking
about the flows of a subexpression is ambiguous without further
information.  Consider the following example, where \True\ and \False\
are closed, syntactic values:
\begin{displaymath}
(\lambda f.f(f\; \True)) (\lambda y. \False)
\end{displaymath}
During evaluation, $y$ gets bound to both \True\ and \False---asking
``what was $y$ bound to?'' is ambiguous.  But let us label the
applications in our term:
\begin{displaymath}
((\lambda f.(f(f\; \True)^1)^2) (\lambda y. \False))^3
\end{displaymath}
Notice that $y$ is bound to different values within different
contexts.  That is, $y$ is bound to \True\ when evaluating the
application labeled 1, and $y$ is bound to \False\ when evaluating the
application labeled 2.  Both of these occur while evaluating the
outermost application, labeled 3.  A string of these application
labels, called a {\em contour}, uniquely describes the {\em context}
under which a subexpression evaluates.  Adopting the common convention
of writing a context as an expression with a hole ``$[\;]$'' in it
\cite{pllc}, the following contours describe the given contexts:
\begin{eqnarray*}
3 2 1 & \mbox{describes} & ((\lambda f.(f [\;]^1)^2) (\lambda y.\False))^3\\
3 2 & \mbox{describes} & ((\lambda f.[\;]^2) (\lambda y.\False))^3
\end{eqnarray*}
So a question about what a subexpression evaluates to {\em within a
  given context} has an unambiguous answer.  The interpreter,
therefore, maintains an environment that maps each variable to a
description of the context in which it was bound.  Similarly, flow
questions about a particular subexpression or variable binding must be
accompanied by a description of a context.  Returning to the example,
the binding cache would give $\eenv(y,321) = \True$ and $\eenv(y,32) =
\False$.

The set $\Val$ consists of closures, however the environment component
of the closures are non-standard.  Rather than mapping variables to
values, these environments map variables to contours; the contour
describes the context in which the variable was bound, so the value
may be retrieved from the $\eenv$ cache.  In other words, these
environments include an added layer of indirection through the cache:
variables are mapped not to their values but the location of their
definition site, where the value can be found.

So we have the following data definitions:
\begin{displaymath}
\begin{array}{cclcll}
\delta & \in & \Delta & = & \Lab^{\star}  & \mbox{contours}\\
     v & \in & \Val & = & \Term \times \Env & \mbox{(contour) values}\\
  \rho & \in & \Env & = & \Var \rightharpoonup \Delta & \mbox{(contour) environments}
\end{array}
\end{displaymath}
Note that this notation overloads the meaning of $\Val$, $\Exp$, and
$\Env$ with that given in the previous section.  It should be clear
from setting which is meant, and when both meanings need to be used in
the same context, the latter will be refered to as {\em contour
values} and {\em contour environments}.

The cache is computed by the instrumented interpreter, $\mathcal{I}$,
the instrumented, intentional analog of $\mathcal{E}$.  It can be
concisely and intuitively written as an imperative program that
mutates an initially empty cache, as given in
\autoref{fig-i-imperative}.

\begin{figure}[h]
\begin{displaymath}
\begin{array}{lcl}
\ivf & : & \Exp \times \Env \times \Delta \rightharpoonup \Unit\\
\\
\iv{x^\ell}^\rho_\delta & = & \ecache(\ell,\delta) \leftarrow \eenv(x,\rho(x))\\
\iv{(\lambda x.e)^\ell}^\rho_\delta & = & \ecache(\ell,\delta) \leftarrow \langle\lambda x.e,\rho'\rangle\\
\ &\ & \quad\mbox{where } \rho' = \rho \restrict \fv{\lambda x.e}\\
\iv{(t^{\ell_1}t^{\ell_2})^\ell}^\rho_\delta & = & \iv{t^{\ell_1}}^{\rho}_\delta;\ \iv{t^{\ell_2}}^{\rho}_\delta;\\
\ &\ & \mbox{let }\langle\lambda x.t^{\ell_0},\ce'\rangle = \ecache(\ell_1,\delta)\mbox{ in}\\
\ & \ & \quad \eenv(x,\delta\ell) \leftarrow \ecache(\ell_2,\delta);\\
\ & \ & \quad \iv{t^{\ell_0}}_{\delta\ell}^{\rho'[x\mapsto \delta\ell]};\\
\ & \ & \quad \ecache(\ell,\delta) \leftarrow \ecache(\ell_0,\delta\ell)
\end{array}
\end{displaymath}
\caption{Instrumented evaluator $\mathcal{I}$, imperative style.}
\label{fig-i-imperative}
\end{figure}

$\iv{t^\ell}^\ce_\delta$ evaluates $t$ and writes the result of
evaluation into the $\ecache$ cache at location $(\ell,\delta)$.  The
notation $\ecache(\ell,\delta) \leftarrow v$ means that the cache is
mutated so that $\ecache(\ell,\delta) = v$, and similarly for
$\eenv(x,\delta) \leftarrow v$.  The type $\Unit$ is used here to
emphasize the imperative nature of the instrumented evaluator; no
meaningful return value is produced, the evaluator is run only for
effect on the caches.  The notation $\delta\ell$ denotes the
concatenation of contour $\delta$ and label $\ell$.  The symbol
$\epsilon$ denotes the empty contour.

We interpret $\ecache(\ell,\delta) = v$ as saying, ``the expression
labeled $\ell$ {\em evaluates to} $v$ in the context described by
$\delta$,'' and $\eenv(x,\delta) = v$ as ``the variable $x$ is {\em
  bound to} $v$ in the context described by $\delta$.''  
Conversely, we say ``$v$ {\em flows to} the expression labelled $\ell$
into the context described by $\delta$,'' and ``$v$ {\em flows to} the
binding of $x$ in the context described by $\delta$,'' respectively.
We refer to a fact such as, $\ecache(\ell,\delta) = v$ or
$\eenv(x,\delta) = v $, as a {\em flow}.
The instrumented interpreter works by accumulating a set of flows as
they occur during evaluation.

Notice that this evaluator does not return values---it writes them
into the cache: if the expression $t^\ell$ evaluates in the contour
$\delta$ to $v$, then $\ecache(\ell,\delta)$ is assigned $v$.  When
the value of a subexpression is needed, as in the application case,
the subexpression is first interpreted (causing its value to be
written in the cache) and subsequently retrieved from the $\ecache$
cache.  When the value of a variable is needed, it is retrieved from
the $\eenv$ cache, using the contour environment to get the
appropriate binding.  

In other words, $\ecache$ is playing the role of a global return
mechanism, while $\eenv$ is playing the role of a global environment.

Although the evaluator is mutating the cache, each location is written
into just once.  A straight-forward induction proof shows that the
current label together with the current contour---which constitute the
cache address that will be written into---forms a unique string.

Returning to the earlier example, the cache constructed by
\begin{displaymath}
\iv{((\lambda f.(f(f\; \True)^1)^2) (\lambda y. \False))^3}_\epsilon^\emptyenv
\end{displaymath}
includes the following entries:
\begin{eqnarray*}
\eenv(f,3) & = & \lambda y.\False\\
\eenv(y,321) & = & \True\\
\eenv(y,32) & = & \False\\
\ecache(1,32) & = & \lambda y.\False\\
%% \ecache(2,3) & = & \False\\
\ecache(3,\epsilon) & = & \False
\end{eqnarray*}

The evaluator can be written in a functional style by threading the
cache through the computation as seen in \autoref{fig-i-functional}.
\begin{figure}[h]
\begin{displaymath}
\begin{array}{lcl}
\ivf & : & \Exp \times \Env \times \Delta \times \Cache \rightharpoonup \Cache\\
\\
%%%%%%%%%%%%%%%%%%
%% Var
\iv{x^\ell}^\rho_\delta\;\ecache,\eenv & = & 
\ecache[(\ell,\delta) \mapsto \eenv(x,\ce(x))], \eenv\\
%%%%%%%%%%%%%%%%%%
%% Lam
\iv{(\lambda x.e)^\ell}^\rho_\delta\;\ecache,\eenv & = & 
\ecache[(\ell,\delta) \mapsto \langle\lambda x.e,\ce'\rangle], \eenv\\
\ &\ & \quad\mbox{where } \rho' = \rho \restrict \fv{\lambda x.e}\\
%%%%%%%%%%%%%%%%%%
%% App
\iv{(t^{\ell_1}t^{\ell_2})^\ell}^\rho_\delta\;\ecache,\eenv & = & 
\mbox{let }\ecache_1,\eenv_1 = \iv{t^{\ell_1}}_{\ce}^\delta\;\ecache,\eenv\mbox{ in}\\
\ & \ & \mbox{let }\ecache_2,\eenv_2 = \iv{t^{\ell_2}}_{\ce}^\delta\;\ecache_1,\eenv_1\mbox{ in}\\
\ &\ & \mbox{let }\langle\lambda x.t^{\ell_0},\ce'\rangle = \ecache_2(\ell_1,\delta)\mbox{ in}\\
\ &\ & \mbox{let }\eenv_3 = \eenv_2[(x,\delta\ell) \mapsto \ecache_3(\ell_2,\delta)]\mbox{ in}\\
\ &\ & \mbox{let }\ecache_3,\eenv_4 = \iv{t^{\ell_0}}_{\ce'[x\mapsto \delta\ell]}^{\delta\ell}\;\ecache_2,\eenv_3\mbox{ in}\\
\ &\ & \mbox{let }\ecache_4 = \ecache_3[(\ell,\delta) \mapsto \ecache_3(\ell_0,\delta\ell)]\mbox{ in}\\
\ &\ & \quad \ecache_4,\eenv_4
\end{array}
\end{displaymath}
\caption{Instrumented evaluator $\mathcal{I}$, functional style.}
\label{fig-i-functional}
\end{figure}

In a more declarative style, we can write a specification of {\em
  acceptable caches}; a cache is acceptable iff it records at least
all of the flows which occur during instrumented evaluation.  The
smallest cache satisfying this acceptability relation is the one that
is computed by the above interpreter, clearly.  The acceptability
relation is given in \autoref{fig-i-declarative}.  It is same cache
acceptability relation can be obtained from that given by \citet[Table
3.10, page 192]{nielson-nielson-hankin} for $k$CFA by letting $k$ be
arbitrarily large.  (Looking ahead, the idea of $k$CFA is that the
evaluator will begin to lose information and approximate evaluation
after a contour has reached a length of $k$.  If $k$ is sufficiently
large, approximation never occurs.  So the acceptability relation of
\autoref{fig-i-declarative} can also be seen as the specification of
``$\infty$CFA''.  For any program that terminates, there is a $k$ such
that performing $k$CFA results in a cache meeting the specification of
\autoref{fig-i-declarative}.  In other words, for any program that
halts, there is a $k$ such that $k$CFA runs it.)
\begin{figure}[h]
\begin{displaymath}
\begin{array}{lcl}
%%%%%%%%%%%%%%%%%%
%% Var
\ecache,\eenv \kmodels{\ce}{\delta} x^\ell & 
\mbox{ iff } & 
\ecache(\ell,\delta) = \eenv(x,\ce(x)) \\
%%%%%%%%%%%%%%%%%%
%% Lam
\ecache,\eenv \kmodels{\ce}{\delta}  (\lambda x .e)^\ell & 
\mbox{ iff } & 
\ecache(\ell,\delta) = \langle\lambda x .e,\ce'\rangle\\
\ & \ & \quad  \mbox{where } \ce' = ce \restrict \fv{\lambda x.e}\\
%%%%%%%%%%%%%%%%%%
%% App
\ecache,\eenv \kmodels{\ce}{\delta} (t^{\ell_1}\ t^{\ell_2})^\ell & 
\mbox{ iff } & 
\ecache\kmodels{\ce}{\delta} t^{\ell_1} \wedge
\ecache\kmodels{\ce}{\delta} t^{\ell_2} \;\wedge \\
\ & \ & \mbox{let }\langle\lambda x.t^{\ell_0},\ce'\rangle = \ecache(\ell_1,\delta) \mbox{ in}\\
\ & \ & \quad\eenv(x,\delta\ell) = \ecache(\ell_2,\delta) \;\wedge \\
\ & \ & \quad\ecache,\eenv\kmodels{\ce'[x\mapsto\delta\ell]}{\delta\ell} t^{\ell_0}\; \wedge  \\
\ & \ & \quad\ecache(\ell,\delta) = \ecache(\ell_0,\delta\ell)
\end{array}
\end{displaymath}
\caption{Exact cache acceptability, or Instrumented evaluator $\mathcal{I}$, declarative style.}
\label{fig-i-declarative}
\end{figure}

There may be a multitude of acceptable analyses for a given program,
so caches are partially ordered by:

\begin{displaymath}
\begin{array}{lllcl}
\ecache & \sqsubseteq & \ecache' & 
\mbox{ iff } & 
\forall \ell,\delta : \ecache(\ell,\delta) = v\Rightarrow \ecache'(\ell,\delta) = v\\
\eenv & \sqsubseteq & \eenv' & 
\mbox{ iff } & 
\forall x,\delta : \eenv(x,\delta) = v\Rightarrow \eenv'(x,\delta) = v
\end{array}
\end{displaymath}

Generally, we are concerned only with the {\em least} such caches with
respect to the domain of variables and labels found in the given
program of interest.

Clearly, because constructing such a cache is equivalent to evaluating
a program, it is not effectively computable.

All of the flow analyses considered in this dissertation can be
thought of as an abstraction (in the sense of being a {\em computable
  approximation}) to this instrumented interpreter, which not only
evaluates a program, but records a history of {\em flows}.

\section{Abstracted Interpretation}
\label{sec:ai}

Computing a complete program trace can produce an arbitrarily large
cache.  One way to regain decidability is to bound the size of the
cache.  This is achieved in $k$CFA by bounding the length of contours
to $k$ labels.

If during the course of evaluation, or more appropriately {\em
  analysis}, the contour is extended to exceed $k$ labels, the
analyzer will truncate the string, keeping the $k$ most recent labels.

But now that the cache size is bounded, a sufficiently large
computation will exhaust the cache.  Due to truncation, the uniqueness
of cache locations no longer holds and there will come a point when a
result needs to be written into a location that is already occupied
with a different value.  If the analyzer were to simply overwrite the
value already in that location, the analysis would be unsound.
Instead the analyzer must consider {\em both} values as flowing out of
this point.

This in turn can lead to further approximation.  Suppose a function
application has two values given for flow analysis of the operator
subterm and another two values given for the operand.  The analyzer
must consider the application of each function to each argument.

$k$CFA is a safe, computable approximation to this instrumented
interpreter; it is a kind of abstract interpretation
\cite{cousot-cousot-popl77,jones-nielson,nielson-nielson-hankin}.
Rather than constructing an {\em exact} cache $\ecache,\eenv$, it
constructs an {\em abstract} cache $\cache,\aenv$:
\begin{displaymath}
\begin{array}{cclcll}
\cache & : & \Lab \times \Delta \rightarrow \widehat{\Val} \\
\aenv & : & \Var \times \Delta \rightarrow \widehat{\Val}\\
\\
\ACache & = & (\Lab \times \Delta \rightarrow \widehat{\Val}) \times
              (\Var \times \Delta \rightarrow \widehat{\Val})
\end{array}
\end{displaymath}
which maps labels and variables, not to values, but to sets of
values---{\em abstract values}:
\begin{displaymath}
\begin{array}{cclcll}
\hat{v}      & \in & \widehat{\Val}   & = & \mathcal{P}(\Term \times \Env) & \mbox{abstract values}.
\end{array}
\end{displaymath}

Approximation arises from contours being bounded at length $k$.  If
during the course of instrumented evaluation, the length of the
contour would exceed length $k$, then the $k$CFA abstract interpreter
will truncate it to length $k$.  In other words, only a partial
description of the context can be given, which results in ambiguity.
A subexpression may evaluate to two distinct values, but within
contexts which are only distinguished by $k+1$ labels.  Questions
about which value the subexpression evaluates to can only supply $k$
labels, so the answer must be {\em both}, according to a sound
approximation.

When applying a function, there is now a set of possible closures that
flow into the operator position.  Likewise, there can be a
multiplicity of arguments.  What is the interpreter to do?  The
abstract interpreter must apply all possible closures to all possible
arguments.

The abstract interpreter $\mathcal{A}$, the imprecise analog of
$\mathcal{I}$, is given in \autoref{fig-a-imperative} using the
concise imperative style.
\begin{figure}[h]
\begin{displaymath}
\begin{array}{lcl}
\avf_k & : & \Exp \times \Env \times \Delta \rightarrow \Unit\\ 
\\
\avk{x^\ell}{\ce}{\delta} & = & \cache(\ell,\delta) \leftarrow \aenv(x,\ce(x))\\
\avk{(\lambda x.e)^\ell}{\ce}{\delta} & = & \cache(\ell,\delta) \leftarrow \{\langle\lambda x.e,\ce'\rangle\}\\
\ &\ & \quad\mbox{where } \ce' = \ce \restrict \fv{\lambda x.e}\\
\avk{(t^{\ell_1} t^{\ell_2})^\ell}{\ce}{\delta} & = & \avk{t^{\ell_1}}{\ce}{\delta}; \avk{t^{\ell_2}}{\ce}{\delta};\\
\ &\ &\mbox{\bf for each }\langle\lambda x.t^{\ell_0},\ce'\rangle \mbox{\bf\ in\ } \cache(\ell_1,\delta) \mbox{\bf\ do}\\
\ &\ & \quad\aenv(x,\lceil\delta\ell\rceil_k) \leftarrow \cache(\ell_2,\delta);\\
\ &\ & \quad\avk{t^{\ell_0}}{\ce'[x\mapsto\lceil\delta\ell\rceil_k]}{\lceil\delta\ell\rceil_k};\\
\ &\ & \quad\cache(\ell,\delta)\leftarrow\cache(\ell_0,\lceil\delta\ell\rceil_k)
\end{array}
\end{displaymath}
\caption{Abstract evaluator $\avf$, imperative style.}
\label{fig-a-imperative}
\end{figure}
We write $\cache(\ell,\delta) \leftarrow \hat{v}$ (or $\aenv(x,\delta)
\leftarrow \hat{v}$) to indicate an updated cache where $\ell,\delta$
(resp., $x,\delta$) maps to $\cache(\ell,\delta) \cup \hat{v}$ (resp.,
$\aenv(\ell,\delta) \cup \hat{v}$).  The notation
$\lceil\delta\rceil_k$ denotes $\delta$ truncated to the rightmost
(i.e., most recent) $k$ labels.

There are many ways the concise imperative abstract evaluator can be
written in a more verbose functional style, and this style will be
useful for proofs in the following sections. 
% Here we write \dots to mean \dots

\begin{figure}[h]
\begin{displaymath}
\begin{array}{lcl}
\avf_k & : & \Exp \times \Env \times \Delta \times \ACache\rightarrow \ACache\\ 
\\
%%%%%%%%%%%%%%%%%%
%% Var
\avk{x^\ell}{\ce}{\delta}\;\cache,\aenv & = & 
\cache[(\ell,\delta) \mapsto \aenv(x,\ce(x))], \aenv\\
%%%%%%%%%%%%%%%%%%
%% Lam
\avk{(\lambda x.e)^\ell}{\ce}{\delta}\;\cache,\aenv & = & 
\cache[(\ell,\delta) \mapsto \{\langle\lambda x.e,\ce'\rangle\}], \aenv\\
\ &\ & \quad\mbox{where } \ce' = \ce \restrict \fv{\lambda x.e}\\
%%%%%%%%%%%%%%%%%%
%% App
\avk{(t^{\ell_1}\ t^{\ell_2})^\ell}{\ce}{\delta}\;\cache,\aenv & = &
\cache_3[(\ell,\delta)\mapsto \cache_3(\ell_0,\delta')],\aenv_3,\mbox{ where}\\
\ &\ &\begin{array}{lcl}
  \delta' & = & \lceil\delta\ell\rceil_k\\
  \cache_1,\aenv_1 & = & \avk{t^{\ell_1}}\ce\delta\;\cache,\aenv\\
  \cache_2,\aenv_2 & = & \avk{t^{\ell_2}}\ce\delta\;\cache_1,\aenv_1\\
  \cache_3,\aenv_3 & = & \\
\multicolumn{3}{c}{
\bigsqcup\nolimits_{\langle\lambda x.t^{\ell_0},\ce'\rangle}^{\cache_2(\ell, \delta)}
\left(\avk{t^{\ell_0}}{\ce'[x\mapsto\delta']}{\delta'}\;\cache_2,\aenv_2[(x,\delta')\mapsto \cache_2(\ell_2,\delta)]\right)}
\end{array}
\end{array}
\end{displaymath}
\caption{Abstract evaluator $\avf$, functional style.}
\label{fig-a-functional}
\end{figure}

Compared to the exact evaluator, contours similarly distinguish
evaluation within contexts described by as many as $k$ application
sites: beyond this, the distinction is blurred.  The imprecision of
the analysis requires that $\avf$ be iterated until the cache reaches
a fixed point, but care must taken to avoid looping in an iteration
since a single iteration of $\avk{e}{\ce}{\delta}$ may in turn make a
recursive call to $\avk{e}{\ce}{\delta}$ under the same contour and
environment.  This care is the algorithmic analog of appealing to the
co-inductive hypothesis in judging an analysis acceptable (described
below).

We interpret $\cache(\ell,\delta) = \{v,\dots\}$ as saying, ``the
expression labeled $\ell$ {\em may evaluate to} $v$ in the context
described by $\delta$,'' and $\eenv(x,\delta) = v$ as ``the variable
$x$ is {\em may be bound to} $v$ in the context described by
$\delta$.''
Conversely, we say ``$v$ {\em flows to} the expression labelled $\ell$
into the context described by $\delta$'' and each $\{v,\dots\}$ ``{\em
  flows out} of the expression labelled $\ell$ in the context
described by $\delta$,'' and ``$v$ {\em flows to} the binding of $x$
in the context described by $\delta$,'' respectively.
We refer to a fact such as, $\cache(\ell,\delta) \ni v$ or
$\aenv(x,\delta) \ni v $, as a {\em flow}.
The abstract interpreter works by accumulating a set of flows as they
occur during abstract interpretation until reaching a fixed point.
Although this overloads the terminology used in describing the
instrumented interpreter, the notions are compatible and the setting
should make it clear which sense is intended.

An acceptable $k$-level control flow analysis for an expression $e$ is
written $\cache,\aenv\kmodels{\ce}{\delta} e$, which states that
$\cache,\aenv$ is an acceptable analysis of $e$ in the context of the
current environment $\ce$ and current contour $\delta$ (for the top
level analysis of a program, these will both be empty).

Just as done in the previous section, we can write a specification of
acceptable caches rather than an algorithm that computes.  The
resulting specification given in \autoref{fig-kcfa-declarative} is
what is found, for example, in \citet{nielson-nielson-hankin}.

\begin{figure}[h]
\begin{displaymath}
\begin{array}{lcl}
\cache,\aenv \kmodels{\ce}{\delta} x^\ell 
& \mbox{ iff } 
& \aenv(x,\ce(x)) \subseteq \cache(\ell,\delta)\\
\cache,\aenv \kmodels{\ce}{\delta} (\lambda x .e)^\ell 
& \mbox{ iff } 
& \langle\lambda x .e,\ce'\rangle \in \cache(\ell,\delta)\\
\ & \ & \quad  \mbox{where } \ce' = ce \restrict \fv{\lambda x.e}\\
\cache,\aenv \kmodels{\ce}{\delta} (t^{\ell_1}\ t^{\ell_2})^\ell &
\mbox{ iff } &
\cache,\aenv\kmodels{\ce}{\delta} t^{\ell_1} \wedge
\cache,\aenv\kmodels{\ce}{\delta} t^{\ell_2} \wedge \\
\ &\ & \quad\forall \langle\lambda x.t^{\ell_0},\ce'\rangle \in \cache(\ell_1,\delta) : \\
\ &\ & \qquad\cache(\ell_2,\delta) \subseteq \aenv(x,\lceil\delta\ell\rceil_k) \wedge\\
\ &\ & \qquad\cache,\aenv\kmodels{ce'[x \mapsto \lceil\delta\ell\rceil_k]}{\lceil\delta\ell\rceil_k} t^{\ell_0} \wedge  \\
\ &\ & \qquad\cache(\ell_0,\lceil\delta\ell\rceil_k) \subseteq \cache(\ell,\delta)
\end{array}
\end{displaymath}
\caption{Abstract cache acceptability, or Abstract evaluator $\mathcal{A}$, declarative style.}
\label{fig-kcfa-declarative}
\end{figure}

There may be a multitude of acceptable analyses for a given program,
so caches are partially ordered by:

\begin{displaymath}
\begin{array}{lllcl}
\cache & \sqsubseteq & \cache' & 
\mbox{ iff } & 
\forall \ell,\delta : \cache(\ell,\delta) = \hat{v}\Rightarrow \hat{v} \subseteq \cache'(\ell,\delta) \\
\aenv & \sqsubseteq & \aenv' & 
\mbox{ iff } & 
\forall x,\delta : \aenv(x,\delta) = \hat{v}\Rightarrow \hat{v}\subseteq \aenv'(x,\delta)
\end{array}
\end{displaymath}

Generally, we are concerned only with the {\em least} such caches with
respect to the domain of variables and labels found in the given
program of interest.

By bounding the contour length, the inductive proof that
$(\ell,\delta)$ was unique for any write into the cache was
invalidated.  Similarly, induction can no longer be relied upon for
verification of acceptability.  It may be the case that proving
$\cache,\aenv\models_\delta^\rho t^\ell$ obligates proofs of other
propositions, which in turn rely upon verification of
$\cache,\aenv\models_\delta^\rho t^\ell$.  Thus the acceptability
relation is defined {\em co-inductively}, given by the greatest fixed
point of the functional defined according to the following clauses of
\autoref{fig-kcfa-declarative}.  Proofs by co-induction would allow
this later obligation to be dismissed by the {\em co-inductive
  hypothesis}.

\paragraph{Fine print:} 
To be precise, we take as our starting point {\em uniform} $k$CFA
rather than a $k$CFA in which,
\begin{eqnarray*}
\ACache & = & (\Lab \times \Env \rightarrow \widehat{\Val}) \times
              (\Var \times \Env \rightarrow \widehat{\Val})
\end{eqnarray*}
The differences are immaterial for our purposes.  See
\citet{nielson-nielson-hankin} for details and a discussion on the use
of coinduction in specifying static analyses.

Having established the foundations of evaluation and analysis, we now
turn to the foundations of our tools and techniques employed in the
investigation of program analysis.

\section{Computational Complexity}

\subsection{Why a Complexity Investigation?}

Static analysis can be understood as a ``technique for computing
conservative approximations of solution for undecidable
problems.''\footnote{Quoted from Michael Schwartzbach's 2009 Static
  Analysis course description from University of Aarhus
  (\url{http://www.brics.dk/~mis/static.html/}, accessed June 3,
  2009).  The same sentiment is expressed in Patrick Cousot's 2005 MIT
  Abstract Interpretation lecture notes on undecidability, complexity,
  automatic abstract termination proofs by semidefinite programming
  (\url{http://web.mit.edu/16.399/www/}).}  Complexity
characterizations therefore offer a natural refinement of that
understanding.

A fundamental question we need to be able to answer is this: what can
be deduced about a long-running program with a time-bounded analyzer?
When we statically analyze exponential-time programs with a
polynomial-time method, there should be a analytic bound on what we
can learn at compile-time: a theorem delineating how exponential time
is being viewed through the compressed, myopic lens of polynomial time
computation.

We are motivated as well by yardsticks such as Shannon's theorem from
information theory \cite{shannon}: specify a bandwidth for
communication and an error rate, and Shannon's results give bounds on
the channel capacity.  We too have essential measures: the time
complexity of our analysis, the asymptotic differential between that
bound and the time bound of the program we are analyzing.  There ought
to be a fundamental result about what information can be yielded as a
function of that differential.  At one end, if the program and
analyzer take the same time, the analyzer can just run the program to
find out everything.  At the other end, if the analyzer does no work
(or a constant amount of work), nothing can be learned.  Analytically
speaking, what is in between?

In the closely related area of pointer analysis, computational
complexity has played a prominent role in the development and
evaluation of analyses.\footnote{See \autoref{sec-pointer-analysis}
  for details.}  It was the starting point for the widely influential
\citet{landi-ryder-pldi92}, according to the authors' retrospective
\citeyearpar{landi-ryder-pldi92retro}.

The theory of computational complexity has proved to be a fruitful
tool in relating seemingly disparate computational problems.  Through
notions of logspace reductions\footnote{Logspace reductions are
  essentially memory efficient translations between instances of
  problems.  The PL-theorist may be most comfortable thinking of these
  reductions as space-efficient compilers.} between problems, what may
have appeared to be two totally unrelated problems can be shown to be,
in fact, so closely related that a method for efficiently computing
solutions to one can be used as a method for efficiently computing the
other, and {\em vice versa}.  For example, at first glance, 0CFA and
circuit evaluation have little to do with each other.  Yet, as shown
in this dissertation, the two problems are intimately related; they
are both complete for \ptime.

There are two fundamental relations a problem can have to a complexity
class.  The problem can be {\em included} in the complexity class,
meaning the problem is no harder than the hardest problems in the
class.  Or, the problem can be a {\em hard} problem within the class,
meaning that no other problem in the class is harder than this one.
When a problem is both included and hard for a given class, it said to
be {\em complete} for that class; it is as hard as the hardest
problems in the class, but no harder.

Inclusion results establish feasibility of analysis---it tells us
analysis can be performed within some upper-bound on resources.  These
results are proven by constructing efficient and correct program
analyzers, either by solving the analysis problem directly or reducing
it to another problem with a known inclusion.

Hardness results, on the other hand, establish the minimum resource
requirements to perform analysis in general.  They can be viewed as
lower-bounds on the difficulty of analysis, telling us, for example,
when no amount of cleverness in algorithm design will improve on the
efficiency of performing analysis.  So while inclusion results have an
existential form: ``there exists an algorithm such that it operates
within these bounded resources,'' hardness results have a universal
form: ``for all algorithms computing the analysis, at least these
resources can be consumed.''

Whereas inclusion results require clever algorithmic insights applied
to a program analyzer, hardness results require clever exploitation of
the analysis to perform computational work.  

Lower bounds are proved by giving reductions---efficient
compilers---for transforming instances of some problem that is known
to be hard for a class, (e.g. circuit evaluation and \ptime) into
instances of a flow analysis problem.
Such a reduction to a flow analysis problem establishes a {\em lower
  bound} on the complexity of flow analysis: solving the flow problem
must be {\em at least as hard} as solving the original problem, which
is known to be of the hardest in the class.

The aim, then, is to solve hard problems by make principled use of the
analysis.  From a programming language perspective, the analyzer can
be regarded as an evaluator of a language, albeit a language with
implicit resource bounds and a (sometimes) non-standard computational
model.  Lower bounds can be proved by writing computationally
intensive programs in this language.  That is, proving lower bounds is
an exercise in clever hacking in the language of analysis.

For flow analysis, inclusion results are largely known.  Simple
algorithmic analysis applied to the program analyzer is often
sufficient to establish an upper bound.

Much work in the literature on flow analysis has been concerned with
finding more and more efficient ways to compute various program
analyses.  But while great effort has been expended in this direction,
there is little work addressing the fundamental limits of this
endeavour.  Lower-bounds tell us to what extent this is possible.

This investigation also provides insight into a more general subject:
the complexity of computing via abstract interpretation.  It stands to
reason that as the computational domain becomes more refined, so too
should computational complexity.  In this instance, the domain is the
size of the abstract cache $\cache$ and the values (namely, {\em
  closures}) that can be stored in the cache.  As the cache size and
number of closures increase\footnote{Observe that since closure
  environments map free variables to contours, the number of closures
  increases when the contour length $k$ is increased.}, so too should
the complexity of computation.  From a theoretical perspective, we
would like to understand better the trade-offs between these various
parameters.

Viewed from another perspective, hardness results can be seen as a
characterization of the {\em expressiveness} of an analysis; it is a
measure of the work an analysis is capable of doing.  The complexity
and expressivity of an analysis are two sides of the same coin and
analyses can be compared and characterized by the class of
computations each captures.  In their definitive study,
\citet{nielson-nielson-hankin} remark, ``Program analysis is still far
from being able to precisely relate ingredients of different
approaches to one another,'' but we can use computational complexity
theory as an effective tool in relating analyses.  Moreover, insights
gleaned from this understanding of analysis can inform future analysis
design.  To develop rich analyses, we should expect larger and larger
classes to be captured.  In short: computational complexity is a means
to both organize and extend the universe of static analyses.

Other research has shown a correspondence between 0CFA and certain
type systems \cite{palsberg-okeefe-toplas95,heintze-sas95} and a
further connection has been made between intersection typing and
$k$CFA \cite{Mossin:97:FlowAnalysis, palsberg-pavlopoulou-jfp01}.
Work has also been done on relating the various flavors of control
flow analysis, such as 0CFA, $k$CFA, polymorphic splitting, and
uniform $k$CFA \cite{nielson-nielson-popl97}.
%% We can enhance these connections by relating of these systems through
%% minimal complete developments.
Moreover, control flow analysis can be computed under a number of
different guises such as set-based analysis \cite{heintze-lfp94},
closure analysis \cite{sestoft-ms,sestoft-fpca89}, abstract
interpretation
\cite{shivers-phd,tang-jouvelot-tacs94,might-shivers-popl06,might-phd,midtgaard-jensen-sas-08,midtgaard-jensen-icfp09},
and type and effect systems
\cite{faxen-sas95,heintze-sas95,faxen-lomaps97,banerjee-icfp97}.

We believe a useful taxonomy of these and other static analyses can be
derived by investigating their computational complexity.  Results on
the complexity of static analyses are way of understanding when two
seemingly different program analyses are in fact computing the same
thing.

% Leading researchers in this field have made evident that this
% knowledge has been both needed and lacking. For example, in a

% definitive survey on principles of program analysis,
% \citet{nielson-nielson-hankin} underline the similarities among
% different program analysis techniques.  In the preface, they make an
% analogy between their approach and that of computational complexity,
% where seemingly unrelated problems can be related through notions of
% {\em reduction} (which is just a resource-bounded compiler).

% While what is known about inclusions between complexity classes ``does
% not amount to much'' \cite[p. 139]{Papadimitriou94}, some\dots some of
% these relations between classes are known.

% \begin{quotation}%
% There is an important analogy to complexity theory here. [\dots] The
% notions of logspace reduction between problems and of \np-complete
% problems lead to realising that problems may appear unrelated at first
% sight and nonetheless be so closely related that a good algorithm or
% heuristics for one will give rise to a good algorithm or heuristics of
% the other.
% \end{quotation}

\subsection{Complexity Classes}

In this section, we review some basic definitions about complexity
classes and define the flow analysis problem.

A complexity class is specified by a model of computation, a mode of
computation (e.g.~deterministic, non-deterministic), the designation
of a unit of {\em work}---a resource used up by computation
(e.g.~time, space), and finally, a function $f$ that bounds the use of
this resource.  The complexity class is defined as the set of all
languages decided by some Turing machine $M$ that for any input $x$
operates in the given mode within the bounds on available resources,
at most $f(|x|)$ units of work \cite[page 139]{Papadimitriou94}.

Turing machines are used as the model of computation, in both
deterministic and non-deterministic mode.  Recall the formal
definition of a Turing machine: a 7-tuple
\[
\langle Q,\Sigma,\Gamma,\delta,q_0,q_a,q_r\rangle
\]
where $Q$, $\Sigma$, and $\Gamma$ are finite sets, $Q$ is the set of
machine states (and $\{q_0,q_a,q_r\}\subseteq Q$), $\Sigma$ is the
input alphabet, and $\Gamma$ is the tape alphabet, where
$\Sigma\subseteq\Gamma$.  For a deterministic machine, the transition
function, $\delta$, is a partial function that maps the current
machine state and tape contents under the head position to the next
machine state, a symbol to write on the tape under the head, and left
or right shift operation for moving the head.  For a non-deterministic
machine, the transition function is actually a relation and may map
each machine configuration to multiple successor configurations.  The
states $q_0$, $q_a$, and $q_r$ are the machine's initial, accept, and
reject states, respectively.  Transitions consume one unit of time and
space consumption is measured by the amount of tape used.

\begin{definition}
Let $f:\mathcal{N} \rightarrow \mathcal{N}$.  We say that machine {\em
  $M$ operates within time $f(n)$} if, for any input string $x$, the
time required by $M$ on $x$ is at most $f(|x|)$ where $|x|$ denotes
the length of string $x$.  Function $f(n)$ is a {\em time bound} for
$M$.  

Let $g:\mathcal{N} \rightarrow \mathcal{N}$. We say that machine {\em
  $M$ operates within space $g(n)$} if, for any input string $x$, the
space required for the work tape of $M$ on $x$ is at most $g(|x|)$.
Function $g(n)$ is a {\em space bound} for $M$.
\end{definition}

Space bounds do not take into consideration the amount of space needed
to represent the input or output of a machine, but only its working
memory.  To this end, one can consider Turing machines with three
tapes: an input, output and work tape.  (Such machines can be
simulated by single tape Turing machines with an inconsequential loss
of efficiency). The input tape is read-only and contains the input
string.  The work tape is initially empty and can be read from and
written to.  The output tape, where the result of computation is
written, is initially empty and is write only.  A space bound
characterizes the size of the work only.  See \citet[Sections
2.3--5]{Papadimitriou94} or \citet[Section 7.5]{garey-johnson} for
further details.

A complexity class is a set of languages representing decision
problems that can be decided within some specified bound on the
resources used, typically time and space.  Suppose the decision
problem can be decided by a deterministic Turing machine operating in
time (space) $f(n)$, we say the problem is in \dtime$(f(n))$
(\dspace$(f(n))$); likewise, if a problem can be decided by a
non-deterministic Turing machine operating in time $f(n)$, we say the
problem is in \ntime$(f(n))$.

We make use of the following standard complexity classes:
\begin{displaymath}
\begin{array}{lrclcl}
 & \mbox{\logspace} & = & \bigcup_{j>0} \mbox{\dspace}(j \log n) \\
\subseteq & \mbox{\ptime} &  = & \bigcup_{j>0} \mbox{\dtime}(n^j)\\
\subseteq & \mbox{\np} & = & \bigcup_{j>0} \mbox{\ntime}(n^j)\\
\subseteq & \mbox{\exptime} & = & \bigcup_{j>0} \mbox{\dtime}(2^{n^j})
\end{array}
\end{displaymath}
In addition to the above inequalities, it is known that \ptime\
$\subset$ \exptime.

% \begin{definition}
% Let the {\em Kalmar elmentary functions} be 
% %\begin{displaymath}
% \begin{eqnarray*}
% \mathbf{K}(0,n) & = & n\\
% \mathbf{K}(t+1,n) & = & 2^{\mathbf{K}(t,n)}
% \end{eqnarray*}
% %\end{displaymath}
% \end{definition}

% \dvhnote{This should be related to the ICC community where there is a
% strong interest in capturing class of computations with implicit
% syntactic discipline.}

What is the difficulty of computing within this hierarchy? What are
the sources of approximation that render such analysis tractable?  We
examine these questions in relation to {\em flow analysis problems},
which are {\em decision problems}, computational problems that require
either a ``yes'' or ``no'' answer, and therefore are insensitive to
output size (it is just one bit).

The flow analysis decision problem we examine is the following:
\begin{description}
\item[Flow analysis decision problem:] Given an expression $e$, an
abstract value $v$, and a pair $(\ell,\delta)$, does $v$ flow into
$(\ell,\delta)$ by this flow analysis?
\end{description}

\section{Proving Lower Bounds: The Art of Exploitation}
\label{sec:approach}

Program exploitation---a staple of hacking---is simply a clever way of
making a program do what you want it to do, even if it was designed to
prevent that action \cite[page 115]{erickson-hacking}. 
This is precisely the idea in proving lower bounds, too.

The program to be exploited, in this setting, is the static analyzer.
What we would like to do varies, but generally we want to exploit
analysis to solve various computationally difficult classes of
problems.
In some cases, what we want to do is {\em evaluate}, rather than
approximate, the program being analyzed.
In this sense, we truly subvert the analyzer by using it to carry out
that which it was designed to avoid (to discover {\em without actually
  running} \cite[page xv]{muchnick-jones-81}).

The approach of exploiting analysis for computation manifests itself
in two ways in this dissertation:

\newpage
\begin{enumerate}
\item Subverting abstraction

The first stems from a observation that perhaps the languages of
abstract and concrete interpretation intersect.  That is, abstract
interpretation makes approximations compared to concrete
interpretation, but perhaps there is a subset of programs on which no
approximation is made by analysis.  For this class of programs,
abstract and concrete interpretation are synonymous.  Such a language
certainly exists for all of the flow analyses examined in this
dissertation.  We conjecture that in general, for every useful
abstract interpretation, there is a subset of the analyzed language
for which abstract and concrete interpretation coincide.  By
identifying this class of programs, lower bounds can be proved by
programming within the subset.

One of the fundamental ideas of computer science is that ``we can
regard almost any program as the evaluator for some language''
\cite[page 360]{sicp}.  So it is natural to ask of any algorithm, {\em
  what is the language being evaluated?}  The question is particularly
relevant when asked of an abstract evaluator.  We can gain insight
into an analysis by comparing the language of the abstract evaluator
to the language of the concrete evaluator.

So a program analysis itself can be viewed as a kind of programming
language, and an analyzer as a kind of evaluator.
Because of the requisite decidability of analysis, these languages
will come with implicit bounds on computational resources---if the
analysis is decidable, these languages cannot be Turing-complete.
But lower bounds can be proved by clever hacking within the
unconventional language of analysis.

This approach is the subject of \autoref{chapter-0cfa}.

\item Harnessing re-evaluation

  The second way analysis can be exploited is to identify the sources
  of approximation in the analysis and instead of avoiding them,
  (turning the abstract into the concrete as above), harness them for
  combinatorial power.  In this approach, lower bounds can be proved
  by programming the language of the analysis in a way that has little
  to do with programming in the language of concrete interpretation.

  Researchers have made empirical observations that computing a more
  precise analysis is often cheaper than performing a less precise
  one.  The less precise analysis ``yields coarser approximations, and
  thus induces more merging. More merging leads to more propagation,
  which in turn leads to more reevaluation''
  \cite{wright-jagannathan-toplas98}.  \citet{might-shivers-icfp06}
  make a similar observation: ``imprecision reinforces itself during a
  flow analysis through an ever-worsening feedback loop.''
  For the purposes of proving lower bounds, we are able to harness
  this re-evaluation as a computation engine.

This approach is the subject of \autoref{chapter-kcfa}.

\end{enumerate}
\chapter{Monovariant Analysis and \ptime}
\chaptermark{Monovariant Analysis and PTIME}
\label{chapter-0cfa}

The monovariant form of flow analysis defined over the pure
$\lambda$-calculus has emerged as a fundamental notion of flow
analysis for higher-order languages, and some form of flow analysis is
used in most analyses for higher-order languages
\cite{heintze-mcallester-pldi97}.

In this chapter, we examine several of the most well-known variations
of monovariant flow analysis: Shivers' 0CFA
\citeyearpar{shivers-pldi88}, Henglein's simple closure analysis
\citeyearpar{henglein92d}, Heintze and McAllester's subtransitive flow
analysis \citeyearpar{heintze-mcallester-pldi97}, Ashley and Dybvig's
sub-0CFA \citeyearpar{ashley-dybvig-toplas98}, Mossin's single
source/use analysis \citeyearpar{mossin-njc98}, and others.

In each case, evaluation and analysis are proved equivalent for the
class of linear programs and a precise characterization of the
computational complexity of the analysis, namely \ptime-completeness,
is given.

\section{The Approximation of Monovariance}

To ensure tractability of any static analysis, there has to be an {\em
  approximation} of something, where information is deliberately {\em
  lost} in the service of providing what is left in a reasonable
amount of time.  A good example of what is lost during {\em
  monovariant} static analysis is that the information gathered for
each occurrence of a bound variable is merged.  When variable $f$
occurs twice in function position with two different arguments,
\begin{displaymath}
%% f(\lambda x.e') \cdots f (\lambda y.e'')
(f\;v_1) \cdots (f\;v_2)
\end{displaymath}
a monovariant flow analysis will blur which copy of the function is
applied to which argument.  If a function $\lambda z.e$ flows into $f$
in this example, the analysis treats occurrences of $z$ in $e$ as
bound to {\em both} $v_1$ and
$v_2$. %% $\lambda x.e'$ and $\lambda y.e''$.

Shivers' 0CFA is among the most well-known forms of monovariant flow
analysis; however, the best known algorithm for 0CFA requires nearly
cubic time in proportion to the size of the analyzed program.

It is natural to wonder whether it is possible to do better, avoiding
this bottleneck, either by improving the 0CFA algorithm in some clever
way or by {\em further} approximation for the sake of faster
computation.

Consequently, several analyses have been designed to approximate 0CFA
by trading precision for faster computation.  Henglein's simple
closure analysis, for example, forfeits the notion of directionality
in flows.  Returning to the earlier example,
\begin{displaymath}
f(\lambda x.e') \cdots f (\lambda y.e'')
\end{displaymath}
simple closure analysis, like 0CFA, will blur $\lambda x.e'$ and
$\lambda y.e''$ as arguments to $f$, causing $z$ to be bound to both.
But unlike 0CFA, a {\em bidirectional} analysis such as simple closure
analysis will identify two $\lambda$-terms with each other.  That is,
because both are arguments to the same function, by the
bi-directionality of the flows, $\lambda x.e'$ may flow out of $\lambda
y.e''$ and {\em vice versa}.

Because of this further curtailing of information, simple closure
analysis enjoys an ``almost linear'' time algorithm.  But in making
trade-offs between precision and complexity, what has been given up
and what has been gained?  Where do these analyses differ and where do
they coincide?

We identify a core language---the linear $\lambda$-calculus---where
0CFA, simple closure analysis, and many other known approximations or
restrictions to 0CFA are rendered identical.  Moreover, for this core
language, analysis corresponds with (instrumented) evaluation.
Because analysis faithfully captures evaluation, and because the
linear $\lambda$-calculus is complete for \ptime, we derive
\ptime-completeness results for all of these analyses.

Proof of this lower bound relies on the insight that linearity of
programs subverts the approximation of analysis and renders it
equivalent to evaluation.  We establish a correspondence between
Henglein's simple closure analysis and evaluation for linear terms.
In doing so, we derive sufficient conditions effectively
characterizing not only simple closure analysis, but many known flow
analyses computable in less than cubic time, such as Ashley and
Dybvig's sub-0CFA, Heintze and McAllester's subtransitive flow
analysis, and Mossin's single source/use analysis.

By using a nonstandard, symmetric implementation of Boolean logic
within the linear lambda calculus, it is possible to simulate circuits
at analysis-time, and as a consequence, we prove that all of the above
analyses are complete for \ptime.  Any sub-polynomial algorithm for
these problems would require (unlikely) breakthrough results in
complexity, such as \ptime\ $=$ \logspace.

We may continue to wonder whether it is possible to do better, either
by improving the 0CFA algorithm in some clever way or by further
approximation for faster computation.  However these theorems
demonstrate the limits of both avenues.  0CFA is inherently
sequential, and so is {\em any} algorithm for it, no matter how
clever.  Designing a provably efficient parallel algorithm for 0CFA is
as hard as parallelizing all polynomial time computations.  On the
other hand, further approximations, such as simple closure analysis
and most other variants of monovariant flow analysis, make no
approximation on a linear program.  This means they too are inherently
sequential and no easier to parallelize.

\section{0CFA}
\label{sec:0cfa}

Something interesting happens when $k=0$.  Notice in the application
rule of the $k$CFA abstract evaluator of \autoref{fig-a-imperative}
that environments are extended as $\rho[x\mapsto
\lceil\delta\ell\rceil_k]$.  When $k=0$,
$\lceil\delta\ell\rceil_0=\epsilon$.  All contour environments map to
the empty contour, and therefore carry no contextual information.  As
such, 0CFA is a ``monovariant'' analysis, analogous to simple-type
inference, which is a monovariant type analysis.

Since there is only one constant environment (the ``everywhere
$\epsilon$'' environment), environments of \autoref{sec:ai} can be
eliminated from the analysis altogether and the cache no longer needs
a contour argument.  Likewise, the set of abstract values collapses
from $\mathcal{P}(\Term\times\Env)$ into $\mathcal{P}(\Term)$.

The result of 0CFA is an {\em abstract cache} that maps each program
point (i.e., label) to a set of lambda abstractions which potentially
flow into this program point at run-time:
\begin{displaymath}
\begin{array}{ccl}
\cache & : & \Lab \rightarrow \mathcal{P}(\Term) \\
%& \qquad\mbox{abstract 0CFA caches}\\
\aenv & : & \Var \rightarrow \mathcal{P}(\Term)\\
\\
\ACache & = & (\Lab \rightarrow \mathcal{P}(\Term)) \times
              (\Var \rightarrow \mathcal{P}(\Term))  
\end{array}
\end{displaymath}
Caches are extended using the notation $\cache[\ell \mapsto s]$, and
we write $\cache[\ell \mapsto^+ s]$ to mean $\cache[\ell \mapsto
(s\cup \cache(\ell))]$.  It is convenient to sometimes think of caches
as mutable tables (as we do in the algorithm below), so we abuse
syntax, letting this notation mean both functional extension and
destructive update.  It should be clear from context which is implied.

\paragraph{The Analysis:} We present the specification of the analysis
here in the style of \citet{nielson-nielson-hankin}.  Each
subexpression is identified with a unique superscript {\em label}
$\ell$, which marks that program point; $\cache(\ell)$ stores all
possible values flowing to point $\ell$, $\aenv(x)$ stores all
possible values flowing to the definition site of $x$.  An {\em
  acceptable} 0CFA analysis for an expression $e$ is written
$\cache,\aenv\models e$ and derived according to the scheme given in
\autoref{fig-0cfa-acceptability}.

\begin{figure}[h]
\begin{displaymath}
\begin{array}{lcl}
\cache,\aenv \models x^\ell &
\mbox{ iff } &
\aenv(x) \subseteq \cache(\ell)\\
\cache,\aenv \models (\lambda x .e)^\ell &
\mbox{ iff } &
\lambda x.e \in \cache(\ell)\\
\cache,\aenv \models (t^{\ell_1}\ t^{\ell_2})^\ell &
\mbox{ iff } &
\cache,\aenv\models t^{\ell_1}\ \wedge
\cache,\aenv\models t^{\ell_2}\ \wedge \\
\ & \ & \quad \forall \lambda x.t^{\ell_0} \in \cache(\ell_1) : \\
\ & \ & \qquad \cache(\ell_2) \subseteq \aenv(x)\ \wedge\ \\
\ & \ & \qquad \cache,\aenv \models t^{\ell_0}\ \wedge \\
\ & \ & \qquad \cache(\ell_0) \subseteq \cache(\ell)
\end{array}
\end{displaymath}
\caption{0CFA abstract cache acceptability.}
\label{fig-0cfa-acceptability}
\end{figure}

The $\models$ relation needs to be coinductively defined since
verifying a judgment $\cache,\aenv \models e$ may obligate
verification of $\cache,\aenv \models e'$ which in turn may require
verification of $\cache,\aenv \models e$.  The above specification of
acceptability, when read as a table, defines a functional, which is
monotonic, has a fixed point, and $\models$ is defined coinductively
as the greatest fixed point of this functional.\footnote{See
  \citet{nielson-nielson-hankin} for details and a thorough discussion
  of coinduction in specifying static analyses.}

Writing $\cache,\aenv \models t^\ell$ means ``the abstract cache
contains all the flow information for program fragment $t$ at program
point $\ell$.''  The goal is to determine the {\em least} cache
solving these constraints to obtain the most precise analysis.  Caches
are partially ordered with respect to the program of interest:
\begin{displaymath}
\begin{array}{lllcl}
\cache & \sqsubseteq & \cache' & 
\mbox{ iff } & 
\forall \ell : \cache(\ell) \subseteq \cache'(\ell)\\
\aenv & \sqsubseteq & \aenv' & 
\mbox{ iff } & 
\forall x : \aenv(x) \subseteq \aenv'(x)
\end{array}
\end{displaymath}

These constraints can be thought of as an abstract evaluator---
$\cache,\aenv \models t^\ell$ simply means {\em evaluate} $t^\ell$, which
serves {\em only} to update an (initially empty) cache.

\begin{figure}[h]
\begin{displaymath}
\begin{array}{lcl}
\avz{x^\ell} & = & 
\cache(\ell) \leftarrow \aenv(x)\\
\avz{(\lambda x.e)^\ell} & = & 
\cache(\ell) \leftarrow \{ \lambda x.e \}\\
\avz{(t^{\ell_1}\ t^{\ell_2})^\ell} & = & 
\avz{t^{\ell_1}};\ \ \avz{t^{\ell_2}};\\
\ &\ & \mbox{\bf for each }\lambda x.t^{\ell_0} \mbox{\bf\ in\ } \cache(\ell_1)\mbox{\bf\ do }\\
\ &\ &   \quad   \aenv(x) \leftarrow \cache(\ell_2);\\ 
\ &\ &   \quad   \avz{t^{\ell_0}};\\
\ &\ &   \quad   \cache(\ell) \leftarrow \cache(\ell_0)
\end{array}
\end{displaymath}
\caption{Abstract evaluator $\mathcal{A}_0$ for 0CFA, imperative style.}
\label{fig-0cfa-imperative}
\end{figure}

The abstract evaluator $\avz\cdot$ is iterated until the finite cache
reaches a fixed point.

\paragraph{Fine Print:} A single iteration of $\avz{e}$ may in turn
make a recursive call $\avz{e}$ with no change in the cache, so care
must be taken to avoid looping.  This amounts to appealing to the
coinductive hypothesis $\cache,\aenv \models e$ in verifying
$\cache,\aenv \models e$.  However, we consider this inessential
detail, and it can safely be ignored for the purposes of obtaining our
main results in which this behavior is {\em never triggered}.

Since the cache size is polynomial in the program size, so is the
running time, as the cache is {\em monotonic}---values are put in, but
never taken out.  Thus the analysis and any decision problems
answered by the analysis are clearly computable within polynomial
time.

\begin{lemma}
The control flow problem for 0CFA is contained in \ptime.
\end{lemma}
\begin{proof}
0CFA computes a binary relation over a fixed structure.  The
computation of the relation is monotone: it begins as empty and is
added to incrementally.  Because the structure is finite, a fixed
point must be reached by this incremental computation.  The binary
relation can be at most polynomial in size, and each increment is
computed in polynomial time.
\end{proof}

\paragraph{An Example:} 

Consider the following program, which we will return to discuss
further in subsequent analyses:

\begin{displaymath}
((\lambda f.((f^1 f^2)^3(\lambda y.y^4)^5)^6)^7 (\lambda
x.x^8)^9)^{10}
\end{displaymath}

The least 0CFA is given by the following cache:
\begin{displaymath}
\begin{array}{l@{\:=\:}l@{\hspace{10pt}}l@{\:=\:}l}
\cache(1) & \{ \lambda x \} &
\cache(6) & \{ \lambda x, \lambda y \} \\
\cache(2) & \{ \lambda x \} &
\cache(7) & \{ \lambda f \} \\
\cache(3) & \{ \lambda x, \lambda y \} &
\cache(8) & \{ \lambda x, \lambda y \} \\
\cache(4) & \{ \lambda y \} &
\cache(9) & \{ \lambda x \} \\
\cache(5) & \{ \lambda y \} &
\cache(10)& \{ \lambda x, \lambda y \}
\end{array}
\hspace{10pt}
\begin{array}{l@{\:=\:}l}
\aenv(f) & \{ \lambda x \} \\
\aenv(x) & \{ \lambda x, \lambda y \}\\
\aenv(y) & \{ \lambda y \}
\end{array}
\end{displaymath}
where we write $\lambda x$ as shorthand for $\lambda x.x^8$, etc.

\section{Henglein's Simple Closure Analysis}
\label{sec:simple}

Simple closure analysis follows from an observation by Henglein some
15 years ago ``in an influential though not often credited technical
report'' \cite[page 4]{midtgaard-07}: he noted that the standard
control flow analysis can be computed in dramatically less time by
changing the specification of flow constraints to use equality rather
than containment \cite{henglein92d}.  The analysis bears a strong
resemblance to simple-type inference---analysis can be performed by
emitting a system of equality constraints and then solving them using
{\em unification}, which can be computed in almost linear time with a
union-find data structure.

Consider a program with both $(f\;x)$ and $(f\;y)$ as subexpressions.
Under 0CFA, whatever flows into $x$ and $y$ will also flow into the
formal parameter of all abstractions flowing into $f$, but it is not
necessarily true that whatever flows into $x$ {\em also} flows into
$y$ and {\em vice versa}.  However, under simple closure analysis,
this is the case.  For this reason, flows in simple closure analysis
are said to be {\em bidirectional}.

\paragraph{The Analysis:} 

The specification of the analysis is given in
\autoref{fig-simple-declarative}.

\begin{figure}[h]
\begin{displaymath}
\begin{array}{lcl}
\cache,\aenv \models x^\ell & 
\mbox{ iff } &
\aenv(x) = \cache(\ell)\\
\cache,\aenv \models (\lambda x .e)^\ell &
\mbox{ iff } &
\lambda x.e \in \cache(\ell)\\
\cache,\aenv \models (t^{\ell_1}\ t^{\ell_2})^\ell & 
\mbox{ iff } &
\cache,\aenv \models t^{\ell_1} \wedge \cache,\aenv \models t^{\ell_2}\ \wedge\\
\ & \ & \quad \forall \lambda x.t^{\ell_0} \in \cache(\ell_1) : \\
\ & \ & \qquad \cache(\ell_2) = \aenv(x)\ \wedge\\
\ & \ & \qquad \cache,\aenv \models t^{\ell_0}\ \wedge \\
\ & \ & \qquad \cache(\ell_0) = \cache(\ell)
\end{array}
\end{displaymath}
\caption{Simple closure analysis abstract cache acceptability.}
\label{fig-simple-declarative}
\end{figure}

\paragraph{The Algorithm:} 

We write $\cache[\ell \leftrightarrow \ell']$ to mean
$\cache[\ell\mapsto^+ \cache(\ell')][\ell'\mapsto^+ \cache(\ell)]$.

\begin{displaymath}
\begin{array}{lcl}
\avz{x^\ell}               & = & \cache(\ell) \leftrightarrow \aenv(x) \\
\avz{(\lambda x.e)^\ell} & = & \cache(\ell) \leftarrow \{ \lambda x.e \}\\
\avz{(t^{\ell_1}_1 t^{\ell_2}_2)^\ell} & = & \avz{t^{\ell_1}_1};\ \ \avz{t^{\ell_2}_2};\\
\ &\ & \mbox{\bf for each }\lambda x.t^{\ell_0}_0 \mbox{\bf\ in\ } \cache(\ell_1)\mbox{\bf\ do }\\
\ &\ & \quad \aenv(x)\leftrightarrow\cache(\ell_2);\\
\ &\ & \quad \avz{t^{\ell_0}_0};\\
\ &\ & \quad \cache(\ell)\leftrightarrow\cache(\ell_0)
\end{array}
\end{displaymath}
The abstract evaluator $\avz\cdot$ is iterated until a fixed point
is reached.\footnote{The fine print of \autoref{sec:0cfa} applies as
well.}  By similar reasoning to that given for 0CFA, simple closure
analysis is clearly computable within polynomial time.

\begin{lemma}
The control flow problem for simple closure analysis is contained in \ptime.
\end{lemma}

\paragraph{An Example:}

Recall the example program of the previous section:
\begin{displaymath}
((\lambda f.((f^1 f^2)^3(\lambda y.y^4)^5)^6)^7 (\lambda
x.x^8)^9)^{10}
\end{displaymath}

Notice that $\lambda x.x$ is applied to itself and then to $\lambda
y.y$, so $x$ will be bound to both $\lambda x.x$ and $\lambda y.y$,
which induces an equality between these two terms.  Consequently,
everywhere that 0CFA was able to deduce a flow set of $\{ \lambda x
\}$ or $\{ \lambda y \}$ will be replaced by $\{ \lambda x, \lambda y
\}$ under a simple closure analysis.  The least simple closure
analysis is given by the following cache (new flows are underlined):
\begin{displaymath}
\begin{array}{l@{\:=\:}l@{\hspace{10pt}}l@{\:=\:}l}
\cache(1) & \{ \lambda x, \underline{\lambda y} \} &
\cache(6) & \{ \lambda x, \lambda y \} \\
\cache(2) & \{ \lambda x, \underline{\lambda y} \} &
\cache(7) & \{ \lambda f \} \\
\cache(3) & \{ \lambda x, \lambda y \} &
\cache(8) & \{ \lambda x, \lambda y \} \\
\cache(4) & \{ \lambda y, \underline{\lambda x} \} &
\cache(9) & \{ \lambda x, \underline{\lambda y} \} \\
\cache(5) & \{ \lambda y, \underline{\lambda x} \} &
\cache(10)& \{ \lambda x, \lambda y \}
\end{array}
\hspace{10pt}
\begin{array}{l@{\:=\:}l}
\aenv(f) & \{ \lambda x, \underline{\lambda y} \} \\
\aenv(x) & \{ \lambda x, \lambda y \}\\
\aenv(y) & \{ \lambda y, \underline{\lambda x} \}
\end{array}
\end{displaymath}

\section{Linearity: Analysis is Evaluation}
\label{sec:linearity}

It is straightforward to observe that in a {\em linear}
$\lambda$-term, each abstraction $\lambda x.e$ can be applied to at
most one argument, and hence the abstracted value can be bound to at
most one argument.\footnote{Note that this observation is clearly
untrue for the {\em nonlinear} $\lambda$-term $(\lambda f.f (a (f b)))
(\lambda x.x)$, as $x$ is bound to $b$, and also to $ab$.}
Generalizing this observation, analysis of a linear $\lambda$-term
coincides exactly with its evaluation.  So not only are the analyses
equivalent on linear terms, but they are also synonymous with
evaluation.

A natural and expressive class of such linear terms are the ones which
implement Boolean logic.  When analyzing the coding of a Boolean
circuit and its inputs, the Boolean output will flow to a
predetermined place in the (abstract) cache.  By placing that value in
an appropriate context, we construct an instance of the control flow
problem: a function $f$ flows to a call site $a$ iff the Boolean
output is $\True$.

Since the circuit value problem \cite{ladner-75}, which is complete
for \ptime, can be reduced to an instance of the 0CFA control flow
problem, we conclude this control flow problem is \ptime-hard.
Further, as 0CFA can be computed in polynomial time, the control flow
problem for 0CFA is \ptime-complete.

One way to realize the computational potency of a static analysis is
to subvert this loss of information, making the analysis an {\em
exact} computational tool.  Lower bounds on the expressiveness of an
analysis thus become exercises in hacking, armed with this newfound
tool.  Clearly the more approximate the analysis, the less we have to
work with, computationally speaking, and the more we have to do to
undermine the approximation.  But a fundamental technique has emerged
in understanding expressivity in static analysis---{\em linearity}.

In this section, we show that when the program is {\em linear}---every
bound variable occurs exactly once---analysis and evaluation are
synonymous.

First, we start by considering an alternative evaluator, given in
\autoref{figure-eval-alt}, which is slightly modified from the one
given in \autoref{figure-eval}.  Notice that this evaluator
``tightens'' the environment in the case of an application, thus
maintaining throughout evaluation that the domain of the environment
is exactly the set of free variables in the expression.  When
evaluating a variable occurrence, there is only one mapping in the
environment: the binding for this variable. Likewise, when
constructing a closure, the environment does not need to be
restricted: it already is.

\begin{figure}
\begin{displaymath}
\begin{array}{lcl}
\evalf['] & : & \Exp \times \Env \rightharpoonup \Val\\
\\
\eval[']{x^\ell}[x \mapsto v]   & = & v\\
\eval[']{(\lambda x.e)^\ell}\rho & = & \langle\lambda x.e,\rho\rangle \\
\eval[']{(e_1\ e_2)^\ell}\rho & = &
       \mbox{\bf let }\langle\lambda x.e_0,\rho'\rangle = 
                      \eval[']{e_1}\rho\restrict\fv{e_1}\mbox{\bf\ in }\\
\ &\ & \mbox{\bf let }v = 
                      \eval[']{e_2}\rho\restrict\fv{e_2}\mbox{\bf\ in }\\
\ &\ &   \quad   \eval[']{e_0}{\rho'[x \mapsto v]}
\end{array}
\end{displaymath}
\caption{Evaluator $\evalf[']$.}
\label{figure-eval-alt}
\end{figure}

% We use $\rho$ to range over {\em environments}, $\Env = \Var
% \rightharpoonup \langle\Term,\Env\rangle$, and let $v$ range over {\em
% closures}, each comprising a term and an environment that closes the
% term.  

This alternative evaluator $\evalf[']$ will be useful in reasoning
about linear programs, but it should be clear that it is equivalent to
the original, standard evaluator $\evalf$ of \autoref{figure-eval}.

\begin{lemma}
$\eval{e}\rho \Longleftrightarrow \eval[']{e}\rho$, when $\dom{\rho} = \fv{e}$.
\end{lemma}

In a linear program, each mapping in the environment corresponds to
the single occurrence of a bound variable.  So when evaluating an
application, this tightening {\em splits} the environment $\rho$ into
$(\rho_1,\rho_2)$, where $\rho_1$ closes the operator, $\rho_2$ closes
the operand, and $\dom{\rho_1} \cap \dom{\rho_2} = \emptyset$.

\begin{definition}
Environment $\rho$ {\em linearly closes} $t$ (or $\langle
t,\rho\rangle$ is a {\em linear closure}) iff $t$ is linear, $\rho$
closes $t$, and for all $x\in\dom{\rho}$, $x$ occurs exactly once
(free) in $t$, $\rho(x)$ is a linear closure, and for all
$y\in\dom{\rho}, x$ does not occur (free or bound) in $\rho(y)$. The
{\em size} of a linear closure $\langle t,\rho\rangle$ is defined as:
\begin{eqnarray*}
|t,\rho| & = & |t|+|\rho|\\
|x| & = & 1\\
|(\lambda x.t^\ell)| & = & 1+|t|\\
|(t_1^{\ell_1}\; t_2^{\ell_2})| & = & 1+|t_1|+|t_2|\\
|[x_1\mapsto c_1,\dots,x_n\mapsto c_n]| & = & n+\sum_i |c_i|
\end{eqnarray*}
\end{definition}

The following lemma states that evaluation of a linear closure cannot
produce a larger value.  This is the environment-based analog to the
easy observation that $\beta$-reduction {\em strictly} decreases the
size of a linear term.
\begin{lemma}\label{lem:smaller}
If $\rho$ linearly closes $t$ and $\eval[']{t^\ell}\rho = c$, then
$|c|\leq|t,\rho|$.
\end{lemma}
\begin{proof}
Straightforward by induction on $|t,\rho|$, reasoning by case analysis
on $t$.  Observe that the size strictly decreases in the application
and variable case, and remains the same in the abstraction case.
\end{proof}

The function $\lab{\cdot}$ is extended to closures and environments by
taking the union of all labels in the closure or in the range of the
environment, respectively.

\begin{definition}
The set of labels in a given term, expression, environment, or closure
is defined as follows:
\begin{displaymath}
\begin{array}{rcl@{\qquad}rcl}
  \lab{t^\ell}      &=& \lab{t}\cup\{\ell\}      & 
  \lab{e_1\;e_2}    &=& \lab{e_1} \cup \lab{e_2} \\
  \lab{x}           &=& \{x\}                    & 
  \lab{\lambda x.e} &=& \lab{e}\cup\{x\}         \\
  \lab{t,\rho}      &=& \lab{t}\cup\lab{\rho}   & 
  \lab{\rho}        &=& \bigcup_{x\in\dom\rho} \lab{\rho(x)}
\end{array}
\end{displaymath}
\end{definition}

\begin{definition}
  A cache $\cache,\aenv$ {\em respects} $\langle t,\rho\rangle$
  (written $\cache,\aenv\loves t,\rho$) when,
\begin{enumerate}
\item $\rho$ linearly closes $t$,
\item $\forall x \in \dom{\rho} . \rho(x) = \langle t',\rho'\rangle
\Rightarrow \aenv(x) = \{t'\} \mbox { and } \cache,\aenv\loves t',\rho'$,
\item $\forall \ell \in \lab{t}, \cache(\ell) = \emptyset$, and
\item $\forall x \in \bv{t}, \aenv(x) = \emptyset$.
\end{enumerate}
\end{definition}
Clearly, $\emptyset\loves t,\emptyset$ when $t$ is closed and linear,
i.e.~$t$ is a linear 

\begin{figure}[h]
\begin{displaymath}
\begin{array}{lcl}
\avf_0 & : & \Exp \times \ACache\rightarrow \ACache\\ 
\\
%%%%%%%%%%%%%%%%%%
%% Var
\avz{x^\ell}\;\cache,\aenv & = & 
\cache[\ell \mapsto \aenv(x)], \aenv\\
%%%%%%%%%%%%%%%%%%
%% Lam
\avz{(\lambda x.e)^\ell}\;\cache,\aenv & = & 
\cache[\ell \mapsto \{\lambda x.e\}], \aenv\\
%%%%%%%%%%%%%%%%%%
%% App
\avz{(t^{\ell_1}\ t^{\ell_2})^\ell}\;\cache,\aenv & = &
\cache_3[\ell\mapsto \cache_3(\ell_0)],\aenv_3,\mbox{ where}\\
\ &\ &\begin{array}{lcl}
  \delta' & = & \lceil\delta\ell\rceil_k\\
  \cache_1,\aenv_1 & = & \avz{t^{\ell_1}}\;\cache,\aenv\\
  \cache_2,\aenv_2 & = & \avz{t^{\ell_2}}\;\cache_1,\aenv_1\\
  \cache_3,\aenv_3 & = & \\
\multicolumn{3}{c}{
\bigsqcup\nolimits_{\lambda x.t^{\ell_0}}^{\cache_2(\ell)}
\left(\avz{t^{\ell_0}}\;\cache_2,\aenv_2[x\mapsto \cache_2(\ell_2)]\right)}
\end{array}
\end{array}
\end{displaymath}
\caption{Abstract evaluator $\mathcal{A}_0$ for 0CFA, functional style.}
\label{fig-0cfa-functional}
\end{figure}

\autoref{fig-0cfa-functional} gives a ``cache-passing'' functional
algorithm for $\avz{\cdot}$ of \autoref{sec:simple}.  It is equivalent
to the functional style abstract evaluator of
\autoref{fig-a-functional} specialized by letting $k = 0$.
We now state and prove the main theorem of this section in terms of
this abstract evaluator.

\begin{theorem}\label{thm:main}
  If $\cache,\aenv\loves t,\rho$, $\cache(\ell)=\emptyset$,
  $\ell\notin\lab{t,\rho}$, $\eval[']{t^\ell}{\rho} = \langle
  t',\rho'\rangle$, and $\avz{t^\ell}{\cache,\aenv} = \cache',\aenv'$,
  then $\cache'(\ell) = \{ t' \}$, $\cache'\loves t',\rho'$, and
  $\cache',\aenv'\models t^\ell$.
\end{theorem}
An important consequence is noted in Corollary \ref{cor:normal}.

\begin{proof} By induction on $|t,\rho|$, reasoning by case analysis on $t$.
\begin{itemize}
\item Case $t\equiv x$.

Since $\cache\loves x,\rho$ and $\rho$ linearly closes $x$, thus $\rho
= [x\mapsto \langle t',\rho'\rangle]$ and $\rho'$ linearly closes
$t'$.  By definition,
\begin{eqnarray*}
\eval[']{x^\ell}\rho &=& \langle t',\rho'\rangle, \mbox{ and}\\
\avz{x^\ell}\cache &=& \cache[x\leftrightarrow \ell].
\end{eqnarray*}
Again since $\cache\loves x,\rho$, $\cache(x) = \{t'\}$, with which the
assumption $\cache(\ell)=\emptyset$ implies
\begin{displaymath}
\cache[x\leftrightarrow\ell](x) =
\cache[x\leftrightarrow\ell](\ell) = 
\{t'\},
\end{displaymath}
and therefore $\cache[x\leftrightarrow\ell]\models x^\ell$.  It
remains to show that $\cache[x\leftrightarrow\ell]\loves t',\rho'$.
By definition, $\cache\loves t',\rho'$.  Since $x$ and $\ell$ do not
occur in $t',\rho'$ by linearity and assumption, respectively, it
follows that $\cache[x\mapsto\ell]\loves t',\rho'$ and the case
holds.

\item Case $t\equiv \lambda x.e_0$.

By definition,
\begin{eqnarray*}
\eval[']{(\lambda x.e_0)^\ell}\rho & = & \langle\lambda x.e_0,\rho\rangle,\\
\avz{(\lambda x.e_0)^\ell}\cache & = & \cache[\ell\mapsto^+ \{\lambda x.e_0\}],
\end{eqnarray*}
and by assumption $\cache(\ell) = \emptyset$, so $\cache[\ell\mapsto^+
\{\lambda x.e_0\}](\ell) = \{\lambda x.e_0\}$ and therefore
$\cache[\ell\mapsto^+ \{\lambda x.e_0\}]\models (\lambda x.e_0)^\ell$.
By assumptions $\ell\notin\lab{\lambda x.e_0,\rho}$ and
$\cache\loves\lambda x.e_0,\rho$, it follows that
$\cache[\ell\mapsto^+ \{\lambda x.e_0\}]\loves\lambda x.e_0,\rho$ and
the case holds.

\item Case $t\equiv t_1^{\ell_1}\; t_2^{\ell_2}$. Let
\begin{eqnarray*}
\eval[']{t_1}\rho\restrict\fv{t_1^{\ell_1}} &=& \langle v_1,\rho_1\rangle = \langle\lambda x.t_0^{\ell_0},\rho_1\rangle,\\
\eval[']{t_2}\rho\restrict\fv{t_2^{\ell_2}} &=& \langle v_2,\rho_2\rangle,\\
\avz{t_1}\cache &=& \cache_1, \mbox{ and}\\
\avz{t_2}\cache &=& \cache_2.
\end{eqnarray*}
Clearly, for $i \in \{1,2\}$, $\cache\loves t_i,\rho\restrict\fv{t_i}$ and
\begin{eqnarray*}
1+\sum_i |t_i^{\ell_i},\rho\restrict\fv{t_i^{\ell_i}}| &=& |(t_1^{\ell_1}\;t_2^{\ell_2}),\rho|.
\end{eqnarray*}

By induction, for $i\in\{1,2\} : \cache_i(\ell_i) = \{v_i\},
\cache_i\loves\langle v_i,\rho_i\rangle,$ and $\cache_i\models
t_i^{\ell_i}$.  From this, it is straightforward to observe that
$\cache_1 = \cache \cup \cache'_1$ and $\cache_2 = \cache \cup
\cache'_2$ where $\cache'_1$ and $\cache'_2$ are disjoint.  So let
$\cache_3 = (\cache_1 \cup \cache_2)[x\leftrightarrow \ell_2]$.  It is
clear that $\cache_3\models t_i^{\ell_i}$.  Furthermore,
\begin{eqnarray*}
\cache_3 &\loves& t_0,\rho_1[x\mapsto \langle v_2,\rho_2\rangle],\\
\cache_3(\ell_0) &=& \emptyset,\mbox{ and}\\
\ell_0 &\notin& \lab{t_0,\rho_1[x\mapsto \langle v_2,\rho_2\rangle]}.
\end{eqnarray*}

By Lemma \ref{lem:smaller},
$|v_i,\rho_i| \leq |t_i,\rho\restrict\fv{t_i}|$, therefore 
\begin{eqnarray*}
|t_0,\rho_1[x\mapsto \langle v_2,\rho_2\rangle]| &<& |(t_1^{\ell_1}\;t_2^{\ell_2})|.
\end{eqnarray*}
Let
\begin{eqnarray*}
\eval[']{t_0^{\ell_0}}\rho_1[x\mapsto \langle v_2,\rho_2\rangle] &=& \langle v',\rho'\rangle,\\
\avz{t_0^{\ell_0}}\cache_3 &=& \cache_4,
\end{eqnarray*}
and by induction, $\cache_4(\ell_0) = \{v'\}$, $\cache_4\loves
v',\rho'$, and $\cache_4\models v'$.  Finally, observe that
$\cache_4[\ell\leftrightarrow\ell_0](\ell) =
\cache_4[\ell\leftrightarrow\ell_0](\ell_0) = \{v'\}$,
$\cache_4[\ell\leftrightarrow\ell_0]\loves v',\rho'$, and
$\cache_4[\ell\leftrightarrow\ell_0]\models
(t_1^{\ell_1}\;t_2^{\ell_2})^\ell$, so the case holds.
\end{itemize}
\end{proof}
We can now establish the correspondence between analysis and
evaluation.

\begin{corollary}\label{cor:normal}
If $\cache$ is the simple closure analysis of a linear program
$t^\ell$, then $\eval[']{t^\ell}\emptyset = \langle v,\rho'\rangle$ where
$\cache(\ell) = \{v\}$ and $\cache\loves v,\rho'$.
\end{corollary}

By a simple replaying of the proof substituting the containment
constraints of 0CFA for the equality constraints of simple closure
analysis, it is clear that the same correspondence can be established,
and therefore 0CFA and simple closure analysis are identical for
linear programs.

\begin{corollary}
If $e$ is a linear program, then $\cache$ is the simple closure
analysis of $e$ iff $\cache$ is the 0CFA of $e$.
\end{corollary}

\paragraph{Discussion:}

Returning to our earlier question of the computationally potent
ingredients in a static analysis, we can now see that when the term is
linear, whether flows are directional and bidirectional is irrelevant.
For these terms, simple closure analysis, 0CFA, and evaluation are
equivalent.  And, as we will see, when an analysis is {\em exact} for
linear terms, the analysis will have a \ptime-hardness bound.
%% The key to designing sub-\ptime\ analyses therefore is to do something
%% less precise on linear terms.

\section{Lower Bounds for Flow Analysis}
\label{sec:circuits}

There are at least two fundamental ways to reduce the complexity of
analysis.  One is to compute more approximate answers, the other is to
analyze a syntactically restricted language.

We use {\em linearity} as the key ingredient in proving lower bounds
on analysis.  This shows not only that simple closure analysis and
other flow analyses are \ptime-complete, but the result is rather
robust in the face of analysis design based on syntactic restrictions.
This is because we are able to prove the lower bound via a highly
restricted programming language---the linear $\lambda$-calculus.  So
long as the subject language of an analysis includes the linear
$\lambda$-calculus, and is exact for this subset, the analysis must be
at least \ptime-hard.

The decision problem answered by flow analysis, described in
\autoref{chap:foundations}, is formulated for monovariant analyses as
follows:
\begin{description}
\item[Flow Analysis Problem:] Given a closed expression $e$, a term
$v$, and label $\ell$, is $v \in \cache(\ell)$ in the analysis of $e$?
\end{description}

\begin{theorem}
If analysis corresponds to evaluation on linear terms, it is
\ptime-hard.
\end{theorem}
The proof is by reduction from the canonical \ptime-complete problem
of circuit evaluation \cite{ladner-75}:
\begin{description}
\item[Circuit Value Problem:] Given a Boolean circuit $C$ of $n$
inputs and one output, and truth values $\vec{x} = x_1,\dots,x_n$, is
$\vec{x}$ accepted by $C$?
\end{description}

An instance of the circuit value problem can be compiled, using only
logarithmic space, into an instance of the flow analysis problem.  The
circuit and its inputs are compiled into a linear $\lambda$-term,
which simulates $C$ on $\vec{x}$ via {\em evaluation}---it normalizes
to true if $C$ accepts $\vec{x}$ and false otherwise.  But since the
analysis faithfully captures evaluation of linear terms, and our
encoding is linear, the circuit can be simulated by flow analysis.

The encodings work like this: \TT\ is the identity on pairs, and \FF\
is the swap.  Boolean values are either $\langle\TT,\FF\rangle$ or
$\langle\FF,\TT\rangle$, where the first component is the ``real''
value, and the second component is the complement.  
\begin{displaymath}
\begin{array}{rcl@{\qquad}rcl}
\TT & \equiv & \lambda p . \mbox{let }\langle x,y \rangle = p\mbox{ in } \langle x, y\rangle &
\True & \equiv &\langle \TT,\FF \rangle\\
\FF & \equiv & \lambda p . \mbox{let }\langle x,y \rangle = p\mbox{ in } \langle y, x\rangle &
\False & \equiv &\langle \FF,\TT \rangle
\end{array}
\end{displaymath}

The simplest connective is \Not, which is an inversion on pairs, like
\FF.  A {\em linear} copy connective is defined as:
\begin{eqnarray*}
\Copy & \equiv & \lambda b.\mbox{let }\langle u,v\rangle = b\mbox{ in }
\langle u\langle \TT,\FF\rangle, v\langle \FF,\TT\rangle\rangle.
\end{eqnarray*}
The coding is easily explained: suppose $b$ is \True, then $u$ is
identity and $v$ twists; so we get the pair
$\langle\True,\True\rangle$.  Suppose $b$ is \False, then $u$ twists
and $v$ is identity; we get $\langle\False,\False\rangle$.  We write
$\Copy_n$ to mean $n$-ary fan-out---a straightforward extension of the
above.

The \AND\ connective is defined as follows:
\begin{displaymath}
\begin{array}{rcl}
\AND & \equiv & \lambda b_1.\lambda b_2.\\
\ & \ & \quad\mbox{let }\langle u_1,v_1\rangle = b_1\mbox{ in}\\
\ & \ & \quad\mbox{let }\langle u_2,v_2\rangle = b_2\mbox{ in}\\
\ & \ & \quad\mbox{let }\langle p_1,p_2\rangle = u_1\langle u_2, \FF\rangle \mbox{ in}\\
\ & \ & \quad\mbox{let }\langle q_1,q_2\rangle = v_1\langle \TT, v_2\rangle \mbox{ in}\\
\ & \ & \qquad\langle p_1, q_1 \circ p_2 \circ q_2 \circ \FF \rangle.
\end{array}
\end{displaymath}
Conjunction works by computing pairs $\langle p_1,p_2\rangle$ and
$\langle q_1,q_2\rangle$.  The former is the usual conjunction on the
first components of the Booleans $b_1,b_2$: $u_1\langle u_2,
\FF\rangle$ can be read as ``if $u_1$ then $u_2$, otherwise false
(\FF).''  The latter is (exploiting De Morgan duality) the disjunction
of the complement components of the Booleans:
$v_1\langle\TT,v_2\rangle$ is read as ``if $v_1$ (i.e.~if not $u_1$)
then true (\TT), otherwise $v_2$ (i.e.~not $u_2$).''  The result of
the computation is equal to $\langle p_1,q_1\rangle$, but this leaves
$p_2,q_2$ unused, which would violate linearity.  However, there is
symmetry to this {\em garbage}, which allows for its disposal.  Notice
that, while we do not know whether $p_2$ is \TT\ or \FF\ and similarly
for $q_2$, we do know that {\em one of them is \TT\ while the other is
  \FF}.  Composing the two together, we are guaranteed that $p_2 \circ
q_2 = \FF$.  Composing this again with another twist (\FF) results in
the identity function $p_2 \circ q_2 \circ \FF = \TT$.  Finally,
composing this with $q_1$ is just equal to $q_1$, so $\langle p_1, q_1
\circ p_2 \circ q_2 \circ \FF \rangle = \langle p_1, q_1\rangle$,
which is the desired result, but the symmetric garbage has been {\em
  annihilated}, maintaining linearity.

Similarly, we define truth-table implication:
\begin{displaymath}
\begin{array}{rcl}
\Implies & \equiv & \lambda b_1.\lambda b_2.\\
\ & \ & \quad\mbox{let }\langle u_1,v_1\rangle = b_1\mbox{ in}\\
\ & \ & \quad\mbox{let }\langle u_2,v_2\rangle = b_2\mbox{ in}\\
\ & \ & \quad\mbox{let }\langle p_1,p_2\rangle = u_1\langle u_2, \TT\rangle \mbox{ in}\\
\ & \ & \quad\mbox{let }\langle q_1,q_2\rangle = v_1\langle \FF, v_2\rangle \mbox{ in}\\
\ & \ & \qquad\langle p_1, q_1 \circ p_2 \circ q_2 \circ \FF \rangle
\end{array}
\end{displaymath}
Let us work through the construction once more: Notice that if $b_1$
is \True, then $u_1$ is \TT, so $p_1$ is \TT\ iff $b_2$ is \True.  And
if $b_1$ is \True, then $v_1$ is \FF, so $q_1$ is \FF\ iff $b_2$ is
\False.  On the other hand, if $b_1$ is \False, $u_1$ is \FF, so $p_1$
is \TT, and $v_1$ is \TT, so $q_1$ is \FF.  Therefore $\langle
p_1,q_1\rangle$ is \True\ iff $b_1 \supset b_2$, and \False\
otherwise. Or, if you prefer, $u_1 \langle u_2,\TT\rangle$ can be read
as ``if $u_1$, then $u_2$ else $\TT$''---the if-then-else description
of the implication $u_1\supset u_2$ ---and $v_1 \langle \FF,
v_2\rangle$ as its De Morgan dual $\neg(v_2\supset v_1)$.  Thus
$\langle p_1,q_1\rangle$ is the answer we want---and we need only
dispense with the ``garbage'' $p_2$ and $q_2$.  De Morgan duality
ensures that one is $\TT$, and the other is $\FF$ (though we do not
know which), so they always compose to $\FF$.

However, simply returning $\langle p_1,q_1\rangle$ violates linearity since
$p_2,q_2$ go unused.  We know that $p_2 = \TT$ iff $q_2 = \FF$ and
$p_2 = \FF$ iff $q_2 = \TT$.  We do not know which is which, but
clearly $p_2 \circ q_2 = \FF\circ\TT = \TT\circ\FF = \FF$.  Composing $p_2\circ
q_2$ with $\FF$, we are guaranteed to get $\TT$.  Therefore $q_1
\circ p_2 \circ q_2 \circ \FF = q_1$, and we have used all bound
variables exactly once.

This hacking, with its self-annihilating garbage, is an improvement
over that given by \citet{mairson-jfp04} and allows Boolean computation
without K-redexes, making the lower bound stronger, but also
preserving all flows.  In addition, it is the best way to do circuit
computation in multiplicative linear logic, and is how you compute
similarly in non-affine typed $\lambda$-calculus \cite{mairson-geocal06}.

By writing continuation-passing style variants of the logic gates, we
can encode circuits that look like straight-line code.  For example,
define CPS logic gates as follows:
\begin{eqnarray*}
  \Andgate  & \equiv & \lambda b_1.\lambda b_2.\lambda k.k(\AND\;b_1\;b_2)\\
  \Orgate   & \equiv & \lambda b_1.\lambda b_2.\lambda k.k(\Or\;b_1\;b_2)\\
  \Implgate & \equiv & \lambda b_1.\lambda b_2.\lambda k.k(\Implies\;b_1\;b_2)\\
  \Notgate  & \equiv & \lambda b.\lambda k. k(\Not\;b)\\
  \Copygate & \equiv & \lambda b.\lambda k.k(\Copy\;b)
\end{eqnarray*}

Continuation-passing style code such as $\Andgate\;b_1\;b_2\;(\lambda
r. e)$ can be read colloquially as a kind of low-level, straight-line
assembly language: ``compute the $\AND$ of registers $b_1$ and $b_2$,
write the result into register $r$, and goto $e$.''

\begin{figure} 
  \centering 
  \includegraphics{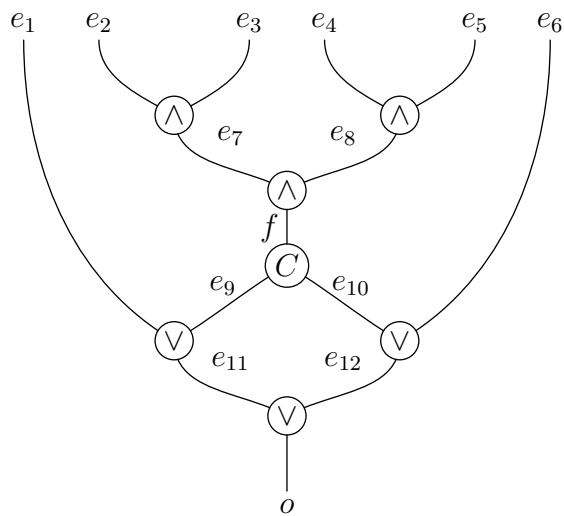}
  \caption{An example circuit.}
  \label{fig-circuit}
\end{figure}

An example circuit is given in \autoref{fig-circuit}, which can be
encoded as:
\begin{eqnarray*}
\mbox{\tt Circuit} & \equiv & 
\lambda e_1.\lambda e_2.\lambda e_3\lambda e_4.\lambda e_5.\lambda e_6.\\
& & \Andgate\; e_2\; e_3\; (\lambda e_7.\\
& & \Andgate\; e_4\; e_5\; (\lambda e_8.\\
& & \Copygate\; f\; (\lambda e_9.\lambda e_{10}.\\
& & \Orgate\; e_1\; e_9\; (\lambda e_{11}.\\
& & \Orgate\; e_{10}\; e_6\; (\lambda e_{12}.\\
& & \Orgate\; e_{11}\; e_{12}\; (\lambda o.o))))))
\end{eqnarray*}
Notice that each variable in this CPS encoding corresponds to a wire
in the circuit.

The above code says:
\begin{itemize} 
\item compute the \AND\ of $e_2$ and $e_3$, putting the result in
  register $e_7$,
\item compute the \AND\ of $e_4$ and $e_5$, putting the result in
  register $e_8$,
\item compute the \AND\ of $e_7$ and $e_8$, putting the result in
  register $f$,
\item make two copies of register $f$, putting the values in registers
  $e_9$ and $e_{10}$,
\item compute the \Or\ of $e_1$ and $e_9$, putting the result in
  register $e_{11}$,
\item compute the \Or\ of $e_{10}$ and $e_6$, putting the result in
  register $e_{12}$,
\item compute the \Or\ of $e_{11}$ and $e_{12}$, putting the result in
  the $o$ (``output'') register.
\end{itemize}

We know from corollary~\ref{cor:normal} that evaluation and analysis
of linear programs are synonymous, and our encoding of circuits will
faithfully simulate a given circuit on its inputs, evaluating to true
iff the circuit accepts its inputs.  But it does not immediately
follow that the circuit value problem can be reduced to the flow
analysis problem.  Let $||C,\vec{x}||$ be the encoding of the circuit
and its inputs.  It is tempting to think the instance of the flow
analysis problem could be stated:
\begin{center}
is \True\ in $\cache(\ell)$ in the analysis
of $||C,\vec{x}||^\ell$?
\end{center}
The problem with this is there may be many syntactic instances of
``\True.''  Since the flow analysis problem must ask about a particular
one, this reduction will not work.  The fix is to use a context which
expects a Boolean expression and induces a particular flow (that can
be asked about in the flow analysis problem) iff that expression
evaluates to a true value.

We use The Widget to this effect.  It is a term expecting a Boolean
value.  It evaluates as though it were the identity function on
Booleans, $\Widget\; b = b$, but it induces a specific flow we can ask
about. If a true value flows out of $b$, then $\True_W$ flows out of
$\Widget\; b$.  If a false value flows out of $b$, then $\False_W$
flows out of $\Widget\; b$, where $\True_W$ and $\False_W$ are
distinguished terms, and the only possible terms that can flow out.
We usually drop the subscripts and say ``does \True\ flow out of
$\Widget\; b$?''  without much ado.
\begin{eqnarray*}
\Widget & \equiv & \lambda b.\\
\ & \ & \quad \mbox{let }\langle u,v\rangle = b\mbox{ in}\\
\ & \ & \quad \mbox{let }\langle x,y\rangle = u\langle f,g\rangle\mbox{ in}\\
\ & \ & \quad \mbox{let }\langle x',y'\rangle = u'\langle f',g'\rangle\mbox{ in}\\
\ & \ & \qquad \langle\langle x a, y n\rangle, \langle x' a', y' b'\rangle\rangle
\end{eqnarray*}

%% Affine widget from SAS.
%% \begin{eqnarray*}
%% \Widget & \equiv & \lambda b.\mbox{let }\langle u,v\rangle=b\mbox{ in }
%% \pi_1 (u\langle\True_W,\False_W\rangle)
%% \end{eqnarray*}

Because the circuit value problem is complete for \ptime, we
conclude:

\begin{theorem}
The control flow problem for 0CFA is complete for \ptime.
\end{theorem}

\begin{corollary}
The control flow problem for simple closure analysis is complete for \ptime.
\end{corollary}

\section{Other Monovariant Analyses}

In this section, we survey some of the existing monovariant analyses
that either approximate or restrict 0CFA to obtain faster analysis
times.  In each case, we sketch why these analyses are complete for
\ptime.

\citet{shivers-sigplan04} noted in his retrospective on control flow
analysis that ``in the ensuing years [since 1988], researchers have
expended a great deal of effort deriving clever ways to tame the cost
of the analysis.''  Such an effort prompts a fundamental question: to
what extent is this possible?

Algorithms to compute 0CFA were long believed to be at least cubic in
the size of the program, proving impractical for the analysis of large
programs, and \citet{heintze-mcallester-lics97} provided strong
evidence to suggest that in general, this could not be improved.  They
reduced the problem of computing 0CFA to that of deciding two-way
nondeterministic push-down automata acceptance (2NPDA); a problem
whose best known algorithm was cubic and had remained so since its
discovery \cite{aho-hopcroft-ullman-ic68}---or so it was believed; see
\autoref{sec-2npda} for a discussion.

In the face of this likely insurmountable bottleneck, researchers
derived ways of further approximating 0CFA, thereby giving up
information in the service of quickly computing a necessarily less
precise analysis in order to avoid the ``cubic bottleneck.''

Such further approximations enjoy linear or near linear algorithms and
have become widely used for the analysis of large programs where the
more precise 0CFA would be to expensive to compute.  But it is natural
to wonder if the algorithms for these simpler analyses could be
improved.  Owing to 0CFA's \ptime-lower bound, its algorithms are
unlikely to be effectively parallelized or made memory efficient.  But
what about these other analyses?

\subsection{Ashley and Dybvig's Sub-0CFA}
\label{sec:sub0}

\citet{ashley-dybvig-toplas98} developed a general framework for
specifying and computing flow analyses; instantiations of the
framework include 0CFA and the polynomial 1CFA of
\citet{jagannathan-weeks-popl95}, for example.  They also developed a
class of instantiations, dubbed {\em sub-0CFA}, that are faster to
compute, but less accurate than 0CFA.

This analysis works by explicitly bounding the number of times the
cache can be updated for any given program point.  After this
threshold has been crossed, the cache is updated with a distinguished
$\unknown$ value that represents all possible $\lambda$-abstractions
in the program.  Bounding the number of updates to the cache for any
given location effectively bounds the number of passes over the
program an analyzer must make, producing an analysis that is $O(n)$ in
the size of the program.  Empirically, Ashley and Dybvig observe that
setting the bound to 1 yields an inexpensive analysis with no
significant difference in enabling optimizations with respect to 0CFA.
%% Following Ashley and Dybvig, we refer to the analysis with the bound
%% set to 1 as sub-0CFA and will refer explicitly to the {\em class} of
%% sub-0CFA analyses when speaking of all possible bounds.

The idea is the cache gets updated once ($n$ times in general) before
giving up and saying all $\lambda$-abstractions flow out of this
point.  But for a linear term, the cache is only updated at most once
for each program point.  Thus we conclude even when the sub-0CFA bound
is 1, the problem is \ptime-complete.

As Ashley and Dybvig note, for any given program, there exists an
analysis in the sub-0CFA class that is identical to 0CFA (namely by
setting $n$ to the number of passes 0CFA makes over the given
program).  We can further clarify this relationship by noting that for
all linear programs, all analyses in the sub-0CFA class are identical
to 0CFA (and thus simple closure analysis).

\subsection{Subtransitive 0CFA}

\citet{heintze-mcallester-lics97} have shown the ``cubic bottleneck''
of computing full 0CFA---that is, computing all the flows in a
program---cannot be avoided in general without combinatorial
breakthroughs: the problem is {\sc{2npda}}-hard, for which the ``the
cubic time decision procedure [\dots] has not been improved since its
discovery in 1968.''

Forty years later, that decision procedure was improved to be slightly
subcubic by \citet{chaudhuri-popl08}.  However, given the strong
evidence at the time that the situation was unlikely to improve in
general, \citet{heintze-mcallester-pldi97} identified several simpler
flow questions\footnote{Including the decision problem discussed in
  this dissertation, which is the simplest; answers to any of the
  other questions imply an answer to this problem} and designed
algorithms to answer them for simply-typed programs. Under certain
typing conditions, namely that the type is within a bounded size,
these algorithms compute in less than cubic time.

The algorithm constructs a graph structure and runs in time linear in a
program's graph.  The graph, in turn, is bounded by the size of the
program's type.  Thus, bounding the size of a program's type results
in a linear bound on the running times of these algorithms.  

If this type bound is removed, though, it is clear that even these
simplified flow problems (and their bidirectional-flow analogs), are
complete for \ptime: observe that every linear term is simply typable,
however in our lower bound construction, the type size is proportional
to the size of the circuit being simulated.  As they point out, when
type size is not bounded, the flow graph may be exponentially larger
than the program, in which case the standard cubic algorithm is
preferred.

Independently, \citet{mossin-njc98} developed a type-based analysis
that, under the assumption of a constant bound on the size of a
program's type, can answer restricted flow questions such as single
source/use in linear time with respect to the size of the explicitly
typed program.  But again, removing this imposed bound results in
\ptime-completeness.

As \citet{hankin-games} point out: both Heintze and McAllester's and
Mossin's algorithms operate on type structure (or structure isomorphic
to type structure), but with either implicit or explicit
$\eta$-expansion.  For simply-typed terms, this can result in an
exponential blow-up in type size.  It is not surprising then, that
given a much richer graph structure, the analysis can be computed
quickly.  

In this light, the results of \autoref{chap:linear-logic} on 0CFA of
$\eta$-expanded, simply-typed programs can be seen as an improvement
of the subtransitive flow analysis since it works equally well for
languages with first-class control and can be performed with only a
fixed number of pointers into the program structure, i.e.~it is
computable in \logspace\ (and in other words, \ptime\ $=$ \logspace\
up to $\eta$).

% \dvhnote{Something should be said, perhaps not here, about the
% practical importance of the sub-cubic analysis, ie.~they are the only
% ones used in practice (Might said this).}

% \dvhnote{Is it worth saying the following? ``Generally speaking, it's
% become inevitable for every CFA, there's an alias analysis, and for
% every alias analysis, there's a CFA.  (I bet it's possible to
% formalize this.)'' \cite{might-correspondence}.  It gives us the
% opportunity to 1) tie in (ubiquitous) work on alias analysis in first
% order languages, 2) re-iterate the case for complexity theoretic
% understanding of analyses in order to relate what is seemingly
% disparate. }

%% But what is the usefulness of an analysis if there is no difference
%% between analyzing and running a program?

%% It remains open whether there are useful analyses for linear programs
%% that can be computed with less resources than it takes to evaluate the
%% program.

\section{Conclusions}

When an analysis is {\em exact}, it will be possible to establish a
correspondence with evaluation.  The richer the language for which
analysis is exact, the harder it will be to compute the analysis.  As
an example in the extreme, \citet{mossin-sas97} developed a flow
analysis that is exact for simply-typed terms.  The computational
resources that may be expended to compute this analysis are {\em ipso
facto} not bounded by any elementary recursive function
\cite{statman79}.  However, most flow analyses do not approach this kind
of expressivity.  By way of comparison, 0CFA only captures \ptime, and
yet researchers have still expending a great deal of effort deriving
approximations to 0CFA that are faster to compute.  But as we have
shown for a number of them, they all coincide on linear terms, and so
they too capture \ptime.

We should be clear about what is being said, and not said.  There is a
considerable difference in practice between linear algorithms
(nominally considered efficient) and cubic---or near
cubic---algorithms (still feasible, but taxing for large inputs), even
though both are polynomial-time.  \ptime-completeness does not
distinguish the two.  But if a sub-polynomial (e.g., \logspace)
algorithm was found for this sort of flow analysis, it would depend on
(or lead to) things we do not know (\logspace\ $=$ \ptime).  

Likewise, were a parallel implementation of this flow analysis to run
in logarithmic time (i.e., \nc), we would consequently be able to
parallelize every polynomial time algorithm.  \ptime-complete problems
are considered to be the least likely to be in \nc.  This is because
logarithmic-space reductions (such as our compiler from circuits to
$\lambda$-terms) preserve parallel complexity, and so by composing
this reduction with a (hypothetical) logarithmic-time 0CFA analyzer
(or equivalently, a logarithmic-time linear $\lambda$-calculus
evaluator) would yield a fast parallel algorithm for {\em all}
problems in \ptime, which are by definition, logspace-reducible to the
circuit value problem \cite[page 377]{Papadimitriou94}.

The practical consequences of the \ptime-hardness result is that we
can conclude any analysis which is exact for linear programs, which
includes 0CFA, and many further approximations, does not have a fast
parallel algorithm unless \ptime\ $=$ \nc.

\chapter{Linear Logic and Static Analysis}
\label{chap:linear-logic}

If you want to understand exactly how and where static analysis is
computationally difficult, you need to know about linearity.
%
%% When I hear the word `Curry-Howard', I reach for my gun.
%
In this chapter, we develop an alternative, graphical representation
of programs that makes explicit both non-linearity and control, and is
suitable for static analysis.

This alternative representation offers the following benefits:

\begin{itemize}
\item It provides clear intuitions on the essence of 0CFA and forms
the basis for a transparent proof of the correspondence between 0CFA
and evaluation for linear programs.

\item As a consequence of symmetries in the notation, it is equally
well-suited for representing programs with first-class control.

\item It based on the technology of linear logic.  Insights gleaned
from linear logic, viewed through the lens of a Curry-Howard
correspondence, can inform program analysis and {\em vice versa}.

\item As an application of the above, a novel and efficient algorithm
for analyzing typed programs (\autoref{sec:eta}) is derived from
recent results on the efficient normalization of linear logic proofs.

\end{itemize}

We give a reformulation of 0CFA in this setting and then transparently
{\em reprove} the main result of \autoref{sec:linearity}: analysis and
evaluation are synonymous for linear programs.

%% \paragraph{Note to the reader:} 
%% \begin{quote}
%% So much recent research effort has been directed at the
%% effective communication of mathematical proofs to computers via the
%% mechanizing meta-theory program (and it may be that the technology of
%% sharing graphs can contribute to that program), but proofs (and
%% programs, naturally) are meant also to communicate ideas {\em between}
%% humans (cf. Frege's concept script). It is primarily for this reason
%% that we choose this representation: it is effective in communicating
%% the {\em ideas} of this research.

%% Simple proofs are hard to find in recent literature.  Rigor has lead
%% to rigor mortis.  Instead, simple proofs exist in an oral history
%% among researchers, proofs that can be (and often are) told at the
%% coffeehouse over napkins and a pen.  There is a place for 50,000 lines
%% of Coq code, but likewise there is a place for the napkin and the pen.

%% It has technical advantages as well, but these can be regarded as
%% icing on the cake (and are perhaps due to the conceptual clarity of
%% the setting). 

%% Come on in, the water is fine.
%% \end{quote}

\section{Sharing Graphs for Static Analysis}

In general, the sharing graph of a term will consist of a
distinguished {\em root} wire from which the rest of the term's graph
``hangs.''
\begin{displaymath}
\includegraphics{figure.13}
\end{displaymath}
At the bottom of the graph, the dangling wires represent
free variables and connect to occurrences of the free variable within
in term.

Graphs consist of ternary abstraction ($\lambda$), apply (@),
sharing ($\triangledown$) nodes, and unary weakening ($\odot$)
nodes.  Each node has a distinguished {\em principal} port.  For unary
nodes, this is the only port.  The ternary nodes have two {\em
auxiliary} ports, distinguished as the {\em white} and {\em black}
ports.

\begin{itemize}
\item A variable occurrence is represented simply as a wire from the
root to the free occurrence of the variable.
\begin{displaymath}
\includegraphics{figure.14}
\end{displaymath}

\item Given the graph for $M$, where $x$ occurs free,
\begin{displaymath}
\includegraphics{figure.15}
\end{displaymath}
the abstraction $\lambda x.M$ is formed as,
\begin{displaymath}
\includegraphics{figure.16}
\end{displaymath}

Supposing $x$ does not occur in $M$, the weakening node ($\odot$) is
used to ``plug'' the $\lambda$ variable wire.
\begin{displaymath}
\includegraphics{figure.17}
\end{displaymath}

\item Given graphs for $M$ and $N$,
\begin{displaymath}
\includegraphics{figure.12},
\end{displaymath}
the application $MN$ is formed as,
\begin{displaymath}
  \includegraphics{figure.11}.
\end{displaymath}
An application node is introduced.  The operator $M$ is connected to
the function port and the operand $N$ is connected to the argument
port.  The continuation wire becomes the root wire for the
application.  Free variables shared between both $M$ and $N$ are
fanned out with sharing nodes.
\end{itemize}

% \dvhnote{Explain the linear logic program.  What it sought to
% accomplish, where does it succeed, what does this have to do with
% static analysis?  See \citet{girard-proofs-and-types}.}

\section{Graphical 0CFA}

We now describe an algorithm for performing control flow analysis that
is based on the graph coding of terms.  The graphical formulation
consists of generating a set of {\em virtual paths} for a program
graph.  Virtual paths describe an approximation of the real paths that
will arise during program execution.

\begin{figure} 
  \centering 
  \scalebox{1.5}{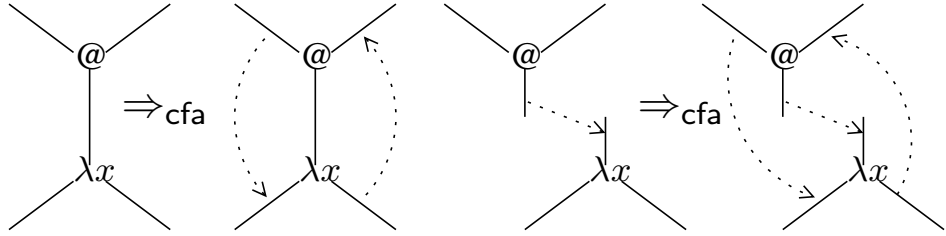}
  \caption{CFA virtual wire propagation rules.} 
  \label{figure-cfa-rules} 
\end{figure}

\autoref{figure-cfa-rules} defines the virtual path propagation rules.
Note that a wire can be identified by its label or a variable
name.\footnote{We implicitly let $\ell$ range over both in the
  following definitions.}  The left hand rule states that a virtual
wire is added from the continuation wire to the body wire and from the
variable wire to the argument wire of each $\beta$-redex.  The right
hand rule states analogous wires are added to each {\em virtual
  $\beta$-redex}---an apply and lambda node connected by a virtual
path.  There is a {\em virtual path} between two wires $\ell$ and
$\ell'$, written $\ell \vp \ell'$ in a CFA-graph iff:
%% 1) $\ell \equiv \ell'$, 2) there is a
%% virtual wire from $\ell$ to $\ell'$, 3) $\ell$ connects to an
%% auxiliary port and $\ell'$ connects to the principal port of a sharing
%% node, or 4) $\ell\vp \ell''$ and $\ell''\vp \ell'$.
\begin{enumerate}
\item $\ell \equiv \ell'$.
\item There is a virtual wire from $\ell$ to $\ell'$.
\item $\ell$ connects to an auxiliary port and $\ell'$ connects to the root
port of a sharing node.
\item There is a virtual path from $\ell$ to $\ell''$ and from
$\ell''$ and $\ell'$.
\end{enumerate}

\paragraph{Reachability:} Some care must be taken to ensure leastness
when propagating virtual wires.  In particular, wires are added only
when there is a virtual path between a {\em reachable} apply and a
lambda node.  An apply node is reachable if it is on the spine of the
program, i.e., if $e=(\cdots((e_0e_1)^{\ell_1}e_2)^{\ell_2}\cdots
e_n)^{\ell_n}$ then the apply nodes with continuation wires labeled
$\ell_1,\dots,\ell_n$ are reachable, or it is on the spine of an
expression with a virtual path from a reachable apply node.

Reachability is usually explained as a known improvement to flow
analysis; precision is increased by avoiding parts of the program that
cannot be reached \cite{ayers-phd93, palsberg-schwarzbach-ic95,
  biswas-popl97, heintze-mcallester-icfp97, midtgaard-jensen-sas-08,
  midtgaard-jensen-icfp09}.

But reachability can also be understood as an analysis analog to weak
normalization.
Reachability says roughly: ``don't analyze under $\lambda$ until the
analysis determines it may be applied.''
On the other hand, weak normalization says: ``don't evaluate under
$\lambda$ until the evaluator determines it is applied.''
The analyzers of \autoref{chap:foundations} implicitly include
reachability since they are based on a evaluation function that
performs weak normalization.

% \citet{wand-02} investigates these issues (question posed by
% Palsberg).

% \citet{wand-02} introduces a very course view of observation. 

% Soundness: if $\cache \models t^\ell$ and $t^\ell \rightarrow
% t^{\ell'}$, then $\cache \models t^{\ell'}$, and $\cache(\ell')
% \subseteq \cache(\ell)$.

% Prediction: if $\cache \models t^\ell$ and $t^\ell \rightarrow^\star
% (\lambda x.e)^{\ell'}$, then

The graph-based analysis can now be performed in the following way:
construct the CFA graph according to the rules in
\autoref{figure-cfa-rules}, then define $\cache(\ell)$ as $\{ (\lambda
x.e)^{\ell'}\; |\; \ell\vp \ell' \}$ and $\aenv(x)$ as $\{ (\lambda
x.e)^{\ell}\; |\; x\vp \ell \}$.  It is easy to see that the
algorithm constructs answers that satisfy the acceptability relation
specifying the analysis.  Moreover, this algorithm constructs least
solutions according to the partial order given in \autoref{sec:ai}.

\begin{lemma}
  $\cache',\aenv'\models e$ implies $\cache,\aenv \sqsubseteq
  \cache',\aenv'$ for $\cache,\aenv$ constructed for $e$ as described
  above.
\end{lemma}

\begin{figure} 
  \centering \scalebox{1.5}{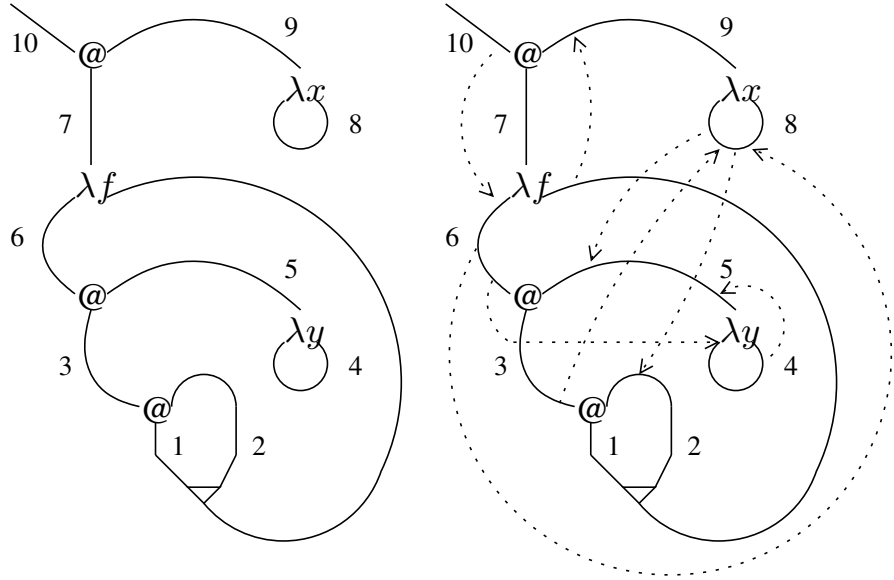}
  \caption{Graph coding and CFA graph of $(\lambda f. ff(\lambda y.y)) (\lambda x.x)$.}
  \label{figure-cfa-example} 
\end{figure}

We now consider an example of use of the algorithm.  Consider the
labeled program: 
\begin{displaymath}
((\lambda f.((f^1 f^2)^3(\lambda y.y^4)^5)^6)^7 (\lambda
x.x^8)^9)^{10}
\end{displaymath}
\autoref{figure-cfa-example} shows the graph coding of the program and
the corresponding CFA graph.  The CFA graph is constructed by adding
virtual wires $10 \vp 6$ and $f\vp 9$, induced by the actual
$\beta$-redex on wire 7.  Adding the virtual path $f\vp 9$ to the
graph creates a virtual $\beta$-redex via the route $1\vp f$ (through
the sharing node), and $f\vp 9$ (through the virtual wire).  This
induces $3 \vp 8$ and $8\vp 2$.  There is now a virtual $\beta$-redex
via $3 \vp 8 \vp 2 \vp f \vp 9$, so wires $6 \vp 8$ and $8 \vp 5$ are
added.  This addition creates another virtual redex via $3 \vp 8 \vp 2
\vp 5$, which induces virtual wires $6 \vp 4$ and $4 \vp 5$.  No
further wires can be added, so the CFA graph is complete.  The
resulting abstract cache gives:
\begin{displaymath}
\begin{array}{l@{\:=\:}ll@{\:=\:}l}
\cache(1) & \{ \lambda x \} &
\cache(6) & \{ \lambda x, \lambda y \} \\
\cache(2) & \{ \lambda x \} &
\cache(7) & \{ \lambda f \} \\
\cache(3) & \{ \lambda x, \lambda y \} &
\cache(8) & \{ \lambda x, \lambda y \} \\
\cache(4) & \{ \lambda y \} &
\cache(9) & \{ \lambda x \} \\
\cache(5) & \{ \lambda y \} &
\cache(10)& \{ \lambda x, \lambda y \}
\end{array}
\begin{array}{l@{\:=\:}l}
\aenv(f) & \{ \lambda x \} \\
\aenv(x) & \{ \lambda x, \lambda y \}\\
\aenv(y) & \{ \lambda y \}
\end{array}
\end{displaymath}

\section{Multiplicative Linear Logic}

The Curry-Howard isomorphism states a correspondence between logical
systems and computational calculi \cite{howard}.  The fundamental idea
is that data types are theorems and typed programs are proofs of
theorems.
\begin{quotation}%
It begins with the observation that an implication $A \rightarrow B$
corresponds to a type of functions from $A$ to $B$, because inferring
$B$ from $A\rightarrow B$ and $A$ can be seen as {\em applying} the
first assumption to the second one---just like a function from $A$ to
$B$ applied to an element of $A$ yields an element of $B$. \cite[p.~v]{sorensen-urzyczyn}
\end{quotation}

For the functional programmer, the most immediate correspondence is
between proofs in propositional intuitionistic logic and simply typed
$\lambda$-terms.  But the correspondence extends considerably further.
\begin{quotation}%
Virtually all proof-related concepts can be interpreted in terms of
computations, and virtually all syntactic features of various
lambda-calculi and similar systems can be formulated in the language
of proof theory.
\end{quotation}

In this section we want to develop the ``proofs-as-programs''
correspondence for linear programs, an important class of programs to
consider for lower bounds on program analysis.  Because analysis and
evaluation are synonymous for linear programs, insights from proof
evaluation can guide new algorithms for program analysis.

The correspondence between simply typed (nonlinear) terms and
intuitionistic logic can be seen by looking at the familiar typing
rules:
\begin{mathpar}
\inferrule*[left=Var]{ }{\Gamma,x : A \vdash x : A}
\and
\inferrule*[left=Abs]{\Gamma,x:A \vdash M : B}{\Gamma \vdash \lambda x.M : A \rightarrow B}
\and
\inferrule*[left=App]{\Gamma\vdash M:A \rightarrow B\\ \Gamma\vdash N:A}{\Gamma \vdash MN:B}
\end{mathpar}

If you ignore the ``proof terms'' (i.e. the programs), you get intuitionsitic sequent calculus:
\begin{mathpar}
\inferrule*[left=Ax]{ }{\Gamma,A \vdash A}
\and
\inferrule*[left=$\rightarrow$I]{\Gamma,A \vdash B}{\Gamma \vdash A \rightarrow B}
\and
\inferrule*[left=$\rightarrow$E]{\Gamma\vdash A \rightarrow B\\ \Gamma\vdash A}{\Gamma \vdash B}
\end{mathpar}

Likewise, {\em linear programs} have their own logical avatar, namely
{\em multiplicative linear logic}.

\subsection{Proofs}

Each atomic formula is given in two forms: positive ($A$) and negative
($A^\perp$) and the {\em linear negation} of $A$ is $A^\perp$ and {\em
vice versa}.  Negation is extended to compound formulae via De
Morgan laws:
\begin{mathpar}
(A\tensor B)^\perp = A^\perp \parr B^\perp

(A\parr B)^\perp = A^\perp \tensor B^\perp
\end{mathpar}

A two sided sequent
\begin{mathpar}
A_1,\dots,A_n \vdash B_1,\dots,B_m
\end{mathpar}
is replaced by
\begin{mathpar}
\vdash A_1^\perp,\dots,A_n^\perp,B_1,\dots,B_m
\end{mathpar}

%% \dvhnote{This section should be flushed out.  See \citet{lafont-ll95}
%% and Harry's notes on that paper for more details.}

The interested reader is referred to \citet{girard-tcs87} for more
details on linear logic.

For each derivation in MLL, there is a proofnet, which abstracts away
much of the needless sequentialization of sequent derivations, ``like
the order of application of independent logical rules: for example,
there are many inessintailly different ways to obtain $\vdash A_1
\parr A_2,\dots A_{n-1} \parr A_n$ from $\vdash A_1,\dots A_n$, while
there is only one proof net representing all these derivations''
\cite{di-cosmo-etal-mscs03}.
There is strong connection with calculus of explicit substitutions
\citet{di-cosmo-etal-mscs03}.

The sequent rules of multiplicative linear logic (MLL) are given in
\autoref{fig-mll-sequents}.

\begin{figure}[h]
\begin{mathpar}
\inferrule*[left=Ax,rightskip=10pt]{ }{A,A^\perp}
\and
\inferrule*[left=Cut]{\Gamma,A \\ A^\perp,\Delta}{\Gamma,\Delta}
\and
\inferrule*[Left=$\parr$]{\Gamma,A,B}{\Gamma,A\parr B}
\and
\inferrule*[Left=$\tensor$]{\Gamma,A \\ \Delta,B}{\Gamma,\Delta,A\tensor B}
\end{mathpar}
\caption{MLL sequent rules.}
\label{fig-mll-sequents}
\end{figure}

\subsection{Programs}

These rules have an easy functional programming interpretation as the
types of a linear programming language (eg.~linear ML), following the
intuitions of the Curry-Howard correspondence
\cite{girard-proofs-and-types,sorensen-urzyczyn}.\footnote{For a more
detailed discussion of the C.-H.~correspondence between linear ML and
MLL, see \citet{mairson-jfp04}.}

(These are written in the more conventional (to functional programmers)
two-sided sequents, but just remember that $A^\perp$ on the left is
like $A$ on the right).

\begin{mathpar}
\inferrule*{ }{x:A\vdash x:A}
\and
\inferrule*{\Gamma \vdash M:A \\ \Delta \vdash N:B}{\Gamma,\Delta \vdash (M,N):A\tensor B}
\and
\inferrule*{\Gamma,x:A\vdash M:B}{\Gamma\vdash \lambda x.M : A\lolli B}
\and 
\inferrule*{\Gamma\vdash M:A\lolli B \\ \Delta\vdash N:A}{\Gamma,\Delta\vdash MN : B}
\and 
\inferrule*{\Gamma\vdash M:A\tensor B\\ \Delta, x:A,y:B \vdash N : C}
           {\Gamma,\Delta \vdash \mbox{let }\langle x,y\rangle = M\mbox{ in }N : C}
\end{mathpar}

The {\sc Axiom} rule says that a variable can be viewed simultaneously
as a continuation ($A^\perp$) or as an expression ($A$)---one man's
ceiling is another man's floor.  Thus we say ``input of type $A$'' and
``output of type $A^\perp$'' interchangeably, along with similar
dualisms.  We also regard $(A^\perp)^\perp$ synonymous with $A$: for
example, {\tt Int} is an integer, and ${\tt Int}^\perp$ is a request
(need) for an integer, and if you need to need an integer---$({\tt
  Int}^\perp)^\perp$---then you have an integer.

The {\sc Cut} rule says that if you have two computations, one with an
output of type $A$, another with an input of type $A$, you can plug
them together.  

The $\tensor$-rule is about pairing: it says that if you have separate
computations producing outputs of types $A$ and $B$ respectively, you
can combine the computations to produce a paired output of type
$A\tensor B$.  Alternatively, given two computations with $A$ an
output in one, and $B$ an input (equivalently, continuation $B^\perp$
an output) in the other, they get paired as a {\em call site}
``waiting'' for a function which produces an {\em output} of type $B$
with an {\em input} of type $A$.  Thus $\tensor$ is both {\tt cons}
and function call (@).

The $\parr$-rule is the linear unpairing of this $\tensor$-formation.
When a computation uses inputs of types $A$ and $B$, these can be
combined as a single input pair, e.g., {\tt let (x,y)=p in...}.
Alternatively, when a computation has an input of type $A$ (output of
continuation of type $A^\perp$) and an output of type $B$, these can
be combined to construct a function which inputs a call site pair, and
unpairs them appropriately.  Thus $\parr$ is both unpairing and $\lambda$.

\section{\texorpdfstring{$\eta$}{Eta}-Expansion and \logspace}
\label{sec:eta}
%% ICFP

\subsection{Atomic versus Non-Atomic Axioms}

The above {\sc Axiom} rule does not make clear whether the formula $A$
is an atomic type variable or a more complex type formula.  When a {\em
  linear} program only has atomic formulas in the ``axiom'' position,
then we can evaluate (normalize) it in logarithmic space.  When the
program is not linear, we can similarly compute a 0CFA analysis in
\logspace.  Moreover, these problems are complete for \logspace.

MLL proofs with non-atomic axioms can be easily converted to ones with
atomic axioms using the following transformation, analogous to
$\eta$-expansion:
\begin{mathpar}
\inferrule
{ }
{\alpha \otimes \beta, \alpha^\perp \parr \beta^\perp}
\and
\Rightarrow 
\and
\inferrule
{\inferrule
  {\inferrule{ }{\alpha,\alpha^\perp}\\ 
   \inferrule{ }{\beta,\beta^\perp}}
  {\alpha \otimes \beta, \alpha^\perp, \beta^\perp}}
{\alpha \otimes \beta, \alpha^\perp \parr \beta^\perp}
\end{mathpar}

\begin{figure} 
  \centering 
  \scalebox{1.5}{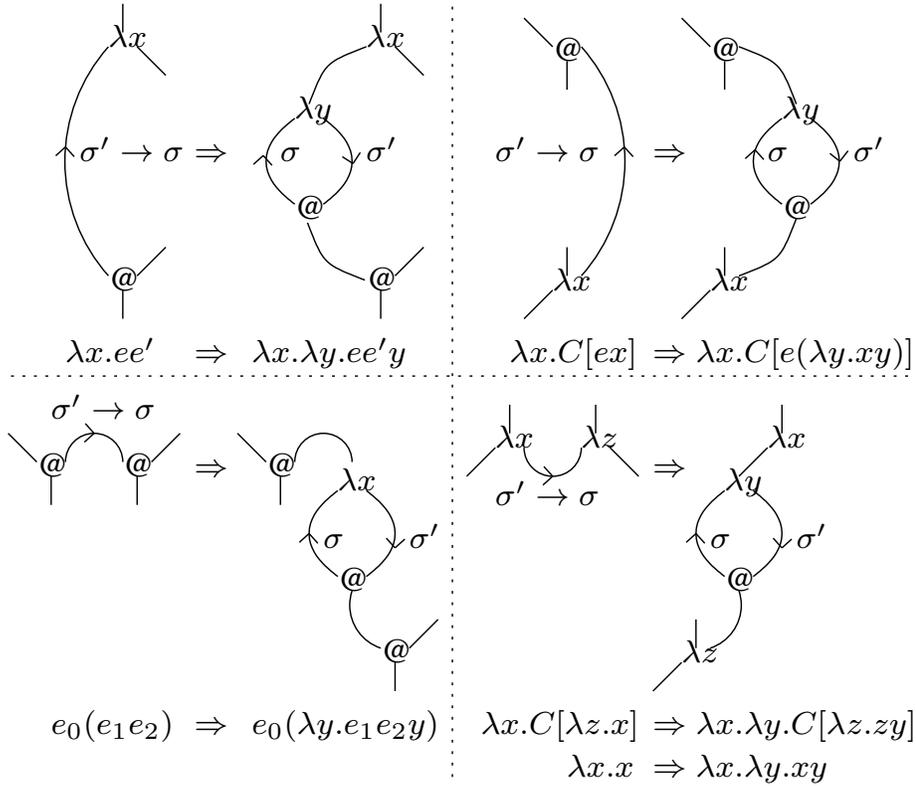}
  \caption{Expansion algorithm.}
  \label{figure-expansion}
\end{figure}
This transformation can increase the size of the proof.  For example,
in the circuit examples of the previous section (which are evidence
for \ptime-completeness), $\eta$-expansion causes an exponential
increase in the number of proof rules used.\footnote{It is linear in
  the formulas used, whose length increases exponentially (not so if
  the formulas are represented by directed acyclic graphs).}  A {\sc
  logspace} evaluation is then polynomial-time and -space in the
original circuit description.

The program transformation corresponding to the above proof expansion
is a version of $\eta$-expansion: see \autoref{figure-expansion}.  The
left hand expansion rule is simply $\eta$, dualized in the unusual
right hand rule. 
% The left rule applies to terms of the
% form $\lambda x . ee'$ where $ee':\sigma \rightarrow \sigma'$.  The
% right rule applies to terms of the form $\lambda x. C[ex]$ where $C$
% is an arbitrary program context and $x : \sigma \rightarrow \sigma'$.
The right rule is written with the @ above the $\lambda$ only to
emphasis its duality with the left rule.  Although not shown in the
graphs, but implied by the term rewriting rules, an axiom may pass
through any number of sharing nodes.

\subsection{Proof Normalization with Non-Atomic Axioms: \ptime}

A normalized {\em linear} program has no redexes. From the type of the
program, one can reconstruct---in a totally syntax-directed way---what
the structure of the term is \cite{mairson-jfp04}.  It is only the
position of the {\em axioms} that is not revealed.  For example, both
{\tt TT} and {\tt FF} from the above circuit example have type {\tt 'a
  * 'a -> 'a * 'a}.\footnote{The linear logic equivalent is
  $(\alpha^\perp\parr\alpha^\perp)\parr (\alpha\tensor\alpha)$.  The
  $\lambda$ is represented by the outer $\parr$, the unpairing by the
  inner $\parr$, and the {\tt cons}ing by the $\tensor$.}
%  $\alpha\tensor\alpha\lolli\alpha\tensor\alphs$}
From this type, we can see that the term is a $\lambda$-abstraction,
the parameter is unpaired---and then, are the two components of type
{\tt a} repaired as before, or ``twisted''?  To twist or not to twist
is what distinguishes {\tt TT} from {\tt FF}.

An MLL {\em proofnet} is a graphical analogue of an MLL proof, where
various sequentialization in the proof is ignored.  The proofnet
consists of axiom, cut, $\tensor$, and $\parr$ nodes with various
dangling edges corresponding to conclusions.  Rules for proofnet
formation (\autoref{fig-mll-proofnets}) follow the rules for sequent
formation (\autoref{fig-mll-sequents}) almost identically.

\begin{figure}[h]
\begin{displaymath}
\scalebox{1.5}{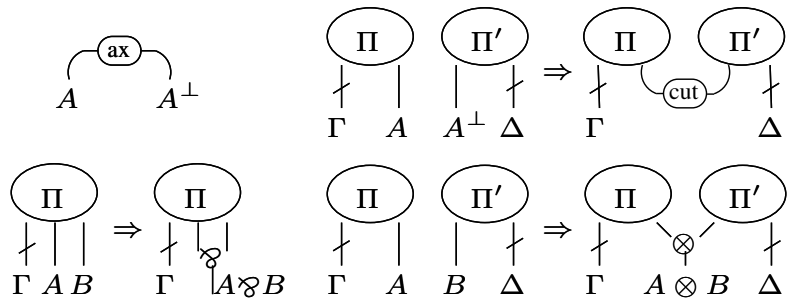}
\end{displaymath}
\caption{MLL proofnets.}
\label{fig-mll-proofnets}
\end{figure}

A binary axiom node has two dangling edges, typed $A$ and $A^\perp$.
Given two disjoint proofnets with dangling edges (conclusions) typed
$\Gamma, A$ and $A^\perp,\Delta$, the edges typed $A,A^\perp$ can be
connected to a binary cut node, and the resulting connected proofnet
has dangling edges typed $\Gamma,\Delta$.  Given a connected proofnet
with dangling wires typed $\Gamma,A,B$, the edges typed $A,B$ can be
connected to the two auxiliary port of a $\parr$ node and the dangling
edge connected to the principal port will have type $A\parr B$.
Finally, given two disjoint proofnets with dangling edges typed
$\Gamma,A$ and $\Delta,B$, the edges typed $A,B$ can be connected to
the two auxiliary ports of a ternary $\tensor$ node; the principal
port then has a dangling wire of type $A\tensor B$.  The intuition is
that $\tensor$ is pairing and $\parr$ is linear unpairing.

The geometry of interaction \cite{girard-goi89,
  gonthier-abadi-levy-popl92}---the semantics of linear logic---and
the notion of paths provide a way to calculate normal forms, and may
be viewed as the logician's way of talking about static program
analysis.\footnote{See \citet{mairson-fsttcs02} for an introduction to
  context semantics and normalization by static analysis in the
  geometry of interaction.}  To understand how this analysis works, we
need to have a graphical picture of what a linear functional program
looks like.

Without loss of generality, such a program has a type $\phi$.  Nodes
in its graphical picture are either $\lambda$ or linear unpairing
($\parr$ in MLL), or application/call site or linear pairing
($\tensor$ in MLL).  We draw the graphical picture so that axioms are
on top, and cuts (redexes, either $\beta$-redexes or pair-unpair
redexes) are on the bottom as shown in \autoref{fig-mll-normal}.

\begin{figure}[h]
\begin{center}
\scalebox{1.5}{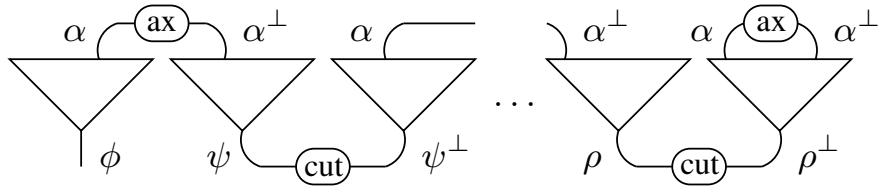}
\end{center}
\caption{MLL proofnet with atomic axioms.}
\label{fig-mll-normal}
\end{figure}

Because the axioms all have atomic type, the graph has the following
nice property:
\begin{lemma}
  Begin at an axiom $\alpha$ and ``descend'' to a cut-link, saving in
  an (initially empty) stack whether nodes are encountered on their
  left or right auxiliary port.  Once a cut is reached, ``ascend'' the
  accompanying structure, popping the stack and continuing left or
  right as specified by the stack token.  Then (1) the stack empties
  exactly when the next axiom $\alpha'$ is reached, and (2) if the
  $k$-th node from the start traversed is a $\tensor$, the $k$-th node
  from the end traversed is a $\parr$, and vice versa.
\end{lemma}
The path traced in the Lemma, using the stack, is
geometry of interaction (GoI), also known as static
analysis.  The correspondence between the $k$-th node from the start and
end of the traversal is precisely that between a {\em call site} ($\tensor$) and
a {\em called function} ($\parr$), or between a {\tt cons} ($\tensor$)
and a linear unpairing ($\parr$).

\subsection{Proof Normalization with Atomic Axioms: \logspace}

\noindent {\sf A sketch of the ``four finger'' normalization
  algorithm:} The stack height may be polynomial, but we do not need
the stack! Put fingers $\alpha,\beta$ on the axiom where the path begins, and
iterate over all possible choices of another two fingers $\alpha',\beta'$ at
another axiom.  Now move $\beta$ and $\beta'$ towards the cut link, where if
$\beta$ encounters a node on the left (right), then $\beta'$ must move left
(right) also.  If $\alpha',\beta'$ were correctly placed initially, then when
$\beta$ arrives at the cut link, it must be met by $\beta'$.  If $\beta'$ isn't
there, or got stuck somehow, then $\alpha',\beta'$ were incorrectly placed, and
we iterate to another placement and try again.

\begin{lemma}
  Any path from axiom $\alpha$ to axiom $\alpha'$ traced by the stack algorithm
  of the previous lemma is also traversed by the ``four finger''
  normalization algorithm.
\end{lemma}

Normalization by static analysis is synonymous with traversing these
paths.  Because these fingers can be stored in logarithmic space, we
conclude \cite{terui-ll2002,mairson-unpub06,mairson-geocal06}:

\begin{theorem}
  Normalization of linear, simply-typed, and fully $\eta$-expanded
  functional programs is contained in \logspace.
\end{theorem}

That 0CFA is then contained in \logspace\ is a casual byproduct of
this theorem, due to the following observation: if application site
$\chi$ calls function $\phi$, then the $\tensor$ and $\parr$
(synonymously, @ and $\lambda$) denoting call site and function are
in distinct trees connected by a {\sc Cut} link.  As a consequence the
0CFA computation is a subcase of the four-finger algorithm: traverse
the two paths from the nodes to the cut link, checking that the paths
are isomorphic, as described above.  The full 0CFA calculation then
iterates over all such pairs of nodes.
\begin{corollary}
  0CFA of linear, simply-typed, and fully $\eta$-expanded
  functional programs is contained in \logspace.
\end{corollary}

\subsection{0CFA in \logspace}

Now let us remove the linearity constraint, while continuing to insist
on full $\eta$-expansion as described above, and simple typing.  The
normalization problem is no longer contained in \logspace, but rather
non-elementary recursive, \cite{statman79,mairson92,asperti-mairson}.
However, 0CFA remains contained in \logspace, because it is now an
{\em approximation}.  This result follows from the following
observation:
\begin{lemma}
\label{lemma-iso-paths}
Suppose $(t^\ell\;e)$ occurs in a simply typed, fully $\eta$-expanded
program and $\lambda x.e \in \cache(\ell)$.  Then the corresponding
$\tensor$ and $\parr$ occur in adjacent trees connected at their roots
by a {\sc cut}-link and on dual, isomorphic paths modulo placement of
sharing nodes.
\end{lemma}
Here ``modulo placement'' means: follow the paths to the {\sc}
cut---then we encounter $\tensor$ (resp., $\parr$) on one path when we
encounter $\parr$ (resp., $\tensor$) on the other, on the same (left,
right) auxiliary ports.  We thus {\em ignore} traversal of sharing
nodes on each path in judging whether the paths are isomorphic.
(Without sharing nodes, the $\tensor$ and $\parr$ would
annihilate---i.e., a $\beta$-redex---during normalization.)

\begin{theorem}
0CFA of a simply-typed, fully $\eta$-expanded program is contained in
\logspace.
\end{theorem}

Observe that 0CFA defines an {\em approximate} form of normalization
which is suggested by simply {\em ignoring} where sharing occurs.
Thus we may define the {\em set} of $\lambda$-terms to which that a term
might evaluate.  Call this {\em 0CFA-normalization}. 
%% \dvhnote{It's not
%% really a set of normal forms... it's a single proofnet.}
\begin{theorem}
  For fully $\eta$-expanded, simply-typed terms, 0CFA-normalization
  can be computed in {\em nondeterministic} \logspace.
\end{theorem}
\begin{conjecture}
  For fully $\eta$-expanded, simply-typed terms, 0CFA-normalization is
  complete for {\em nondeterministic} \logspace.
\end{conjecture}
The proof of the above conjecture likely depends on a coding of
arbitrary directed graphs and the consideration of commensurate path
problems.
\begin{conjecture}
  An algorithm for 0CFA normalization can be realized by {\em optimal
  reduction}, where sharing nodes {\em always} duplicate, and never
  annihilate.
\end{conjecture}

%% \dvhnote{If we are going to talk about 0CFA ``normalization'' it seems
%% it should be introduced before the $\eta$ expansion section.}

\subsection{\logspace-hardness of Normalization and 0CFA:  linear,
  simply-typed, fully \texorpdfstring{$\eta$}{eta}-expanded programs}

That the normalization and 0CFA problem for this class of programs is
as hard as any \logspace\ problem follows from the {\sc
  logspace}-hardness of the {\em permutation problem}: given a
permutation $\pi$ on $1,\ldots, n$ and integer $1\le i\le n$, are $1$
and $i$ on the same cycle in $\pi$?  That is, is there a $k$ where
$1\le k\le n$ and $\pi^k(1)=i$?

Briefly, the \logspace-hardness of the permutation problem is as
follows.\footnote{This presentation closely follows
  \citet{mairson-geocal06}.}  Given an arbitrary \logspace\ Turing
machine $M$ and an input $x$ to it, visualize a graph where the nodes
are machine IDs, with directed edges connecting successive
configurations.  Assume that $M$ always accepts or rejects in unique
configurations.  Then the graph has two connected components: the
``accept'' component, and the ``reject'' component.  Each component is
a directed tree with edges pointing towards the root (final
configuration).  Take an Euler tour around each component (like
tracing the fingers on your hand) to derive two {\em cycles}, and thus
a {\em permutation} on machine IDs.  Each cycle is polynomial size,
because the configurations only take logarithmic space.  The
equivalent permutation problem is then: does the initial configuration
and the accept configuration sit on the same cycle?

The following linear ML code describes the ``target'' code of a
transformation of an instance of the permutation problem.  For a
permutation on $n$ letters, we take here an example where $n=3$.
Begin with a vector of length $n$ set to {\tt False}, and a
permutation on $n$ letters:
\begin{quote}
{\small
\begin{alltt}
- val V= (False,False,False);
val V = ((fn,fn),(fn,fn),(fn,fn))
  : (('a * 'a -> 'a * 'a) * ('a * 'a -> 'a * 'a))
  * (('a * 'a -> 'a * 'a) * ('a * 'a -> 'a * 'a))
  * (('a * 'a -> 'a * 'a) * ('a * 'a -> 'a * 'a))
\end{alltt}}
\end{quote}
Denote as $\nu$ the type of vector {\tt V}.  
\begin{quote}
{\small
\begin{alltt}
- fun Perm (P,Q,R)= (Q,R,P);
val Perm = fn : \ensuremath{\nu} -> \ensuremath{\nu}
\end{alltt}}
\end{quote}
The function {\tt Insert} {\em linearly} inserts {\tt True} in the
first vector component, using all input exactly once:
\begin{quote}
{\small
\begin{alltt}
- fun Insert ((p,p'),Q,R)= ((TT,Compose(p,p')),Q,R);
val Insert = fn : \ensuremath{\nu} -> \ensuremath{\nu}
\end{alltt}}
\end{quote}
The function {\tt Select} {\em linearly} selects the third vector
component:
\begin{quote}
{\small
\begin{alltt}
- fun Select (P,Q,(r,r'))= 
    (Compose (r,Compose (Compose P, Compose Q)),r');
val Select = fn
  : \ensuremath{\nu} -> (('a * 'a -> 'a * 'a) * ('a * 'a -> 'a * 'a))
\end{alltt}}
\end{quote}
Because {\tt Perm} and {\tt Insert} have the same flat type, they can
be composed iteratively in ML without changing the type.  (This
clearly is {\em not} true in our coding of circuits, where the size of
the type increases with the circuit.  A careful coding limits the type
size to be polynomial in the circuit size, regardless of circuit
depth.)
\begin{lemma}
Let $\pi$ be coded as permutation {\tt Perm}.  Define {\tt Foo} to be
\begin{center}
{\tt Compose(Insert,Perm)} 
\end{center}
composed with itself $n$ times.  Then 1 and
$i$ are on the same cycle of $\pi$ iff {\tt Select (Foo V)} normalizes
to {\tt True}.
\end{lemma}

Because 0CFA of a linear program is identical with normalization, we conclude:

\begin{theorem}
0CFA of a simply-typed, fully $\eta$-expanded program is complete for
\logspace.
\end{theorem}

% As a consequence of the hardness construction having a constant type,
% we may conclude the 0CFA of any bounded type program ($\eta$-expanded
% or not) is \logspace-hard:

% \begin{theorem}
% 0CFA of a bounded, simply-typed program is \logspace-hard.
% \end{theorem}

The usefulness of $\eta$-expansion has been noted in the context of
partial evaluation \cite{jones-et-al-93,danvy-et-al-toplas96}.  In
that setting, $\eta$-redexes serve to syntactically embed binding-time
coercions.  In our case, the type-based $\eta$-expansion does the
trick of placing the analysis in \logspace\ by embedding the type
structure into the syntax of the program.\footnote{Or, in slogan form:
  \logspace\ $=$ \ptime\ upto $\eta$.}

\section{Graphical Flow Analysis and Control}
%% ICFP
\citet{shivers-sigplan04} argues that ``CPS provide[s] a uniform
representation of control structure,'' allowing ``this machinery to be
employed to reason about context, as well,'' and that ``without CPS,
separate contextual analyses and transforms must be also
implemented---redundantly,'' in his view.  Although our formulation of
flow analysis is a ``direct-style'' formulation, a graph
representation enjoys the same benefits of a CPS representation,
namely that control structures are made explicit---in a graph a
continuation is simply a wire.  Control constructs such as
\verb|call/cc| can be expressed directly \cite{lawall-mairson-esop00}
and our graphical formulation of control flow analysis carries over
without modification.

\citet{lawall-mairson-esop00} derive graph representations of programs
with control operators such as \verb|call/cc| by first translating
programs into continuation passing style (CPS).  They observed that
when edges in the CPS graphs carrying answer values (of type $\bot$)
are eliminated, the original (direct-style) graph is regained, modulo
placement of boxes and croissants that control sharing.  Composing the
two transformations results in a direct-style graph coding for
languages with \verb|call/cc| (hereafter, $\lambda_\mathcal{K}$).  The
approach applies equally well to languages such as Filinski's
symmetric $\lambda$-calculus \citeyearpar{filinski-ctcs89}, Parigot's
$\lambda_\mu$ calculus \citeyearpar{parigot-lpar92}, and most any
language expressible in CPS.

Languages such as $\lambda_\xi$, which contains the ``delimited
control'' operators {\em shift} and {\em reset}
\cite{danvy-filinksi-lfp90}, are not immediately amenable to this
approach since the direct-style transformation requires all calls to
functions or continuations be in tail position.  Adapting this
approach to such languages constitutes an open area of research.

\newcommand\graph[1]{\ensuremath{#1}}
\begin{figure} 
  \centering \scalebox{1.5}{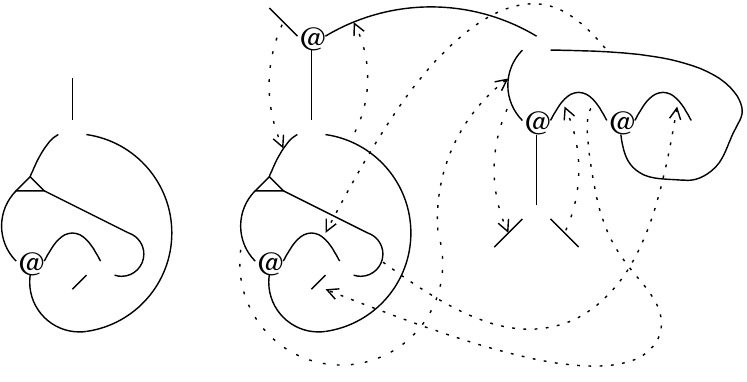}
  \caption{Graph coding of {\tt call/cc} and example CFA graph.}
  \label{figure-callcc}
\end{figure}

The left side of \autoref{figure-callcc} shows the graph coding of
\verb|call/cc|.  Examining this graph, we can read of an
interpretation of \verb|call/cc|, namely: \verb|call/cc| is a function
that when applied, copies the current continuation ($\triangle$) and
applies the given function $f$ to a function ($\lambda v\dots$) that
when applied abandons the continuation at that point ($\odot$) and
gives its argument $v$ to a copy of the continuation where
\verb|call/cc| was applied.  If $f$ never applies the function it is
given, then control returns ``normally'' and the value $f$ returns is
given to the other copy of the continuation where \verb|call/cc| was
applied.

%% \dvhnote{I started to write about boxing and confluence but I'm not
%% sure they're relevant here.}

%% Of course, just as the graph coding of terms given in
%% \autoref{section-preliminaries} is insufficient for performing
%% reduction, so too is this graph coding of \verb|call/cc|.  What is
%% missing is a boxing strategy like that given in
%% \autoref{figure-red-graph}...

%% It is well-known that term reduction for $\lambda_\mathcal{K}$ is
%% non-confluent, however the graph reduction technology developed by
%% \citet{lawall-mairson-esop00} resolves this issue---confluence is
%% regained by program's giving multiple answers to a shared
%% continuation.

The right side of \autoref{figure-callcc} gives the CFA graph for the
program:
\begin{displaymath}
(\mbox{{\tt call/cc }} (\lambda k.(\lambda x.\overline{1})(k\overline{2})))^\ell
\end{displaymath}

From the CFA graph we see that $\cache(\ell) = \{ \overline{1},
\overline{2} \}$, reflecting the fact that the program will return
$\overline{1}$ under a call-by-name reduction strategy and
$\overline{2}$ under call-by-value.  Thus, the analysis is indifferent
to the reduction strategy.  Note that whereas before, approximation
was introduced through nonlinearity of bound variables, approximation
can now be introduced via nonlinear use of continuations, as seen in
the example.  In the same way that 0CFA considers all occurrences of a
bound variable ``the same'', 0CFA considers all continuations obtained
with each instance of \verb|call/cc| ``the same''.

Note that we can ask new kinds of interesting questions in this
analysis.  For example, in \autoref{figure-callcc}, we can compute
which continuations are potentially {\em discarded}, by computing
which continuations flow into the weakening node of the \verb|call/cc|
term.  (The answer is the continuation $((\lambda
x.\overline{1})[\;])$.)  Likewise, it is possible to ask which
continuations are potentially {\em copied}, by computing which
continuations flow into the principal port of the sharing node in the
\verb|call/cc| term (in this case, the top-level empty continuation
$[\;]$).  Because continuations are used linearly in
\verb|call/cc|-free programs, the questions were uninteresting
before---the answer is always {\em none}.

Our proofs for the \ptime-completeness of 0CFA for the untyped
$\lambda$-calculus
%% ---and likewise for the results on $k$CFA---
carry over without modification languages such as
$\lambda_\mathcal{K}$, $\lambda_\mu$ and the symmetric
$\lambda$-calculus.  In other words, first-class control operators
such as \verb|call/cc| increase the expressivity of the language,
%% in the sense of \citet{felleisen-esop90}, 
but add nothing to the computational complexity of control flow
analysis.  In the case of simply-typed, fully $\eta$-expanded
programs, the same can be said.  A suitable notion of ``simply-typed''
programs is needed, such as that provided by \citet{griffin-popl90}
for $\lambda_\mathcal{K}$.  The type-based expansion algorithm of
\autoref{figure-expansion} applies without modification and lemma
\autoref{lemma-iso-paths} holds,
%% the paths are simply ``modulo more sharing nodes,''
allowing 0CFA for this class of programs to be done in \logspace.
Linear logic provides a foundation for (classical) $\lambda$-calculi
with control; related logical insights allow control flow analysis in
this setting.

The graph coding of terms in our development is based on the
technology of {\em sharing graphs} in the untyped case, and {\em proof
  nets} in the typed case \cite{lafont-ll95}.  The technology of
proofnets have previously been extended to intersection types
\cite{regnier-phd,neergaard-phd}, which have a close connection to
flow analysis \cite{amtoft-turbak-esop00, palsberg-pavlopoulou-jfp01,
  wells-etal-jfp02, banerjee-jensen-mscs03}.
 
The graph codings, CFA graphs, and virtual wire propagation rules
share a strong resemblance to the ``pre-flow'' graphs, flow graphs,
and graph ``closing rules'', respectively, of
\citet{Mossin:97:FlowAnalysis}.  Casting the analysis in this light
leads to insights from linear logic and optimal reduction.  For
example, as \citet[page 78]{Mossin:97:FlowAnalysis} notes, the CFA
virtual paths computed by 0CFA are an approximation of the actual
run-time paths and correspond exactly to the ``well-balanced paths'' of
\citet{asperti-laneve-tcs95} as an approximation to ``legal paths''
\cite{levy-phd} and results on proof normalization in linear logic
\cite{DBLP:conf/ictcs/MairsonT03} informed the novel flow analysis algorithms
presented here.

%% \section{Flow analysis and optimal evaluation}
%% ICFP

\chapter{\kcfa\ and \exptime}
\chaptermark{KCFA AND EXPTIME}
\label{chapter-kcfa}

In this chapter, we give an exact characterization of the
computational complexity of the $k$CFA hierarchy.  For any $k > 0$, we
prove that the control flow decision problem is complete for
deterministic exponential time.  This theorem validates empirical
observations that such control flow analysis is intractable.  It also
provides more general insight into the complexity of abstract
interpretation.

\section{Shivers' \kcfa}

As noted in \autoref{sec-overview}, practical flow analyses must
negotiate a compromise between complexity and precision, and their
{\em expressiveness} can be characterized by the computational
resources required to compute their results.

Examples of simple yet useful flow analyses include Shivers' 0CFA
\citeyearpar{shivers-pldi88} and Henglein's simple closure analysis
\citeyearpar{henglein92d}, which are {\em monovariant}---functions
that are closed over the same $\lambda$-expression are identified.
Their expressiveness is characterized by the class \ptime\
(\autoref{chapter-0cfa}).

As described in \autoref{chapter-0cfa}, a monovariant analysis is one
that approximates at points of nonlinearity.  When a variable appears
multiple times, flow information is merged together for all sites.

So for example, in analyzing the program from \autoref{sec:0cfa},
\begin{displaymath}
(\lambda f. (ff)(\lambda y.y))(\lambda x.x),
\end{displaymath}
a monovariant analysis such as 0CFA or simple closure analysis will
merge the flow information for the two occurrences of $f$.
Consequently both $\lambda x.x$ and $\lambda y.y$ are deemed to flow
out of the whole expression.

More precise analyses can be obtained by incorporating
context-sensitivity to distinguish multiple closures over the same
$\lambda$-term, resulting in ``finer grained approximations, expending
more work to gain more information'' \cite{shivers-pldi88,shivers-phd}.
This context-sensitivity will allow the two occurrences of $f$ to be
analyzed independently.  Consequently, such an analysis will determine
that only $\lambda y.y$ flows out of the expression.

To put it another way, a context-sensitive analysis is capable of
evaluating this program.

As a first approximation to understanding, the added precision of
$k$CFA can be thought of as the ability to do partial reductions
before analysis.  If were to first reduce all of the apparent redexes
in the program, and {\em then} do 0CFA on the residual, our example
program would look like
\begin{displaymath}
(\lambda x_1.x_1)(\lambda x_2.x_2)(\lambda y.y).
\end{displaymath}
Being a linear program, 0CFA is sufficient to prove only $\lambda y.y$
flows out of this residual.  The polyvariance of $k$CFA is powerful
enough to prove the same, however it is important to note that it is
{\em not} done by a bounded reduction of the program.  Instead, the
$k$CFA hierarchy uses the last $k$ calling contexts to distinguish
closures.

The increased precision comes with an empirically observed increase in
cost. As Shivers noted in his retrospective on the $k$CFA work
\citeyearpar{shivers-sigplan04}:
\begin{quotation}%
It did not take long to discover that the basic analysis, for any $k >
0$, was intractably slow for large programs. In the ensuing years,
researchers have expended a great deal of effort deriving clever ways
to tame the cost of the analysis.
\end{quotation}

A fairly straightforward calculation---see, for example,
\citet{nielson-nielson-hankin}---shows that 0CFA can be computed in
polynomial time, and for any $k>0$, $k$CFA can be computed in
exponential time.

These naive upper bounds suggest that the $k$CFA hierarchy is
essentially {\em flat}; researchers subsequently ``expended a great
deal of effort'' trying to improve them.\footnote{Even so, there is a
big difference between algorithms that run in $2^n$ and $2^{n^2}$
steps, though both are nominally in \exptime.}  For example, it seemed
plausible (at least, to us) that the $k$CFA problem could be in \np\
by {\em guessing} flows appropriately during analysis.

As this dissertation shows, the naive algorithm is essentially the
best one, and the {\em lower} bounds are what needed improving.  We
prove that for all $k>0$, computing the $k$CFA analysis requires (and
is thus complete for) deterministic exponential time.  There is, in
the worst case---and plausibly, in practice---no way to tame the cost
of the analysis.  Exponential time is required.

\paragraph{Why should this result matter to functional programmers?}

\begin{itemize}

\item This result concerns a fundamental and ubiquitous static
  analysis {\em of} functional programs.

  The theorem gives an analytic, scientific characterization of the
  expressive power of $k$CFA.  As a consequence, the {\em empirically
    observed} intractability of the cost of this analysis can be
  understood as being {\em inherent in the approximation problem being
    solved}, rather than reflecting unfortunate gaps in our
  programming abilities.

  Good science depends on having relevant theoretical understandings
  of what we observe empirically in practice.

  This connection between theory and experience contrasts with the
  similar result for ML-type inference \cite{mairson-popl90}: while
  the problem of recognizing ML-typable terms is complete for
  exponential time, programmers have happily gone on programming.  It
  is likely that their need of higher-order procedures, essential for
  the lower bound, is not
  considerable.\footnote{\citet{kuan-macqueen-ml-07} have recently
    provided a refined perspective on the complexity of ML-type
    inference that explains why it works so quickly in practice.}

  But static flow analysis really has been costly, and this theorem
  explains why.

\item The theorem is proved {\em by} functional programming.

  We take the view that the analysis itself is a functional
  programming language, albeit with implicit bounds on the available
  computational resources.  Our result harnesses the approximation
  inherent in $k$CFA as a computational tool to hack exponential time
  Turing machines within this unconventional language.  The hack used
  here is completely unlike the one used for the ML analysis, which
  depended on complete developments of {\tt let}-redexes.  The theorem
  we prove in this paper uses approximation in a way that has little
  to do with normalization.
\end{itemize}

We proceed by first bounding the complexity of $k$CFA from above,
showing that $k$CFA can be solved in exponential time
(\autoref{sec-kcfa-in-exptime}).
This is easy to calculate and is known \cite{nielson-nielson-hankin}.
Next, we bound the complexity from below by using $k$CFA as a
SAT-solver.
This shows $k$CFA is at least \np-hard\ (\autoref{sec-kcfa-np-hard}).
The intuitions developed in the \np-hardness proof can be improved to
construct a kind of exponential iterator. 
A small, elucidative example is developed in \autoref{sec-toy}.
These ideas are then scaled up and applied in
\autoref{sec-kcfa-exptime-hard} to close the gap between the \exptime\
upper bound and \np\ lower bound by giving a construction to simulate
Turing machines for an exponential number of steps using $k$CFA, thus
showing $k$CFA to be complete for \exptime.

\section{\kcfa\ is in \exptime}
\label{sec-kcfa-in-exptime}

Recall the definition of $k$CFA from \autoref{sec:ai}.
The cache, $\cache,\aenv$, is a finite mapping and has $n^{k+1}$
entries.
Each entry contains  a set of closures.
The environment component of each closure maps $p$ free variables to
any one of $n^k$ contours.
There are $n$ possible $\lambda$-terms and $n^{kp}$ environments, so
each entry contains at most $n^{1+kp}$ closures.
Analysis is monotonic, and there are at most $n^{1+(k+1)p}$ updates to
the cache.
Since $p\leq n$, we conclude:

\begin{lemma}
The control flow problem for $k$CFA is contained in \exptime.
\end{lemma}

It is worth noting that this result shows, from a complexity
perspective, the flatness of the $k$CFA hierarchy: {\em for any
  constant} $k$, $k$CFA is decidable in exponential time.
It is not the case, for example, that 1CFA requires exponential time
(for all $j$, $\mbox{\dtime}(2^{n^j}) \subseteq \mbox{\exptime}$),
while 2CFA requires {\em doubly} exponential time (for all $j$,
$\mbox{\dtime}(2^{2^{n^j}}) \subseteq \mbox{2\exptime}$), 3CFA
requires {\em triply} exponential time, etc.
There are strict separation results for these classes, $\mbox\exptime
\subset \mbox{2\exptime} \subset \mbox{3\exptime}$, etc., so we know
from the above lemma there is no need to go searching for lower bounds
greater than \exptime.

\section{\kcfa\ is \np-hard}
\label{sec-kcfa-np-hard}

Because $k$CFA makes approximations, many closures can flow to a
single program point and contour.  In 1CFA, for example, $\lambda
w.wx_1x_2\cdots x_n$ has $n$ free variables, with an exponential
number of possible associated environments mapping these variables to
program points (contours of length 1).  Approximation allows us to
bind each $x_i$, independently, to either of the closed
$\lambda$-terms for \True\ or \False\ that we saw in the
\ptime-completeness proof for 0CFA.  In turn, application to an
$n$-ary Boolean function necessitates computation of all $2^n$ such
bindings in order to compute the flow out from the application site.
The term \True\ can only flow out if the Boolean function is
satisfiable by some truth valuation.
\begin{figure}[h]
\begin{eqnarray*}
&& (\lambda f_1. (f_1\ \True) (f_1\ \False))\\
&& (\lambda x_1.\\
&& \TB (\lambda f_2. (f_2\ \True) (f_2\ \False))\\
&& \TB (\lambda x_2.\\
&& \TB\TB (\lambda f_3. (f_3\ \True) (f_3\ \False))\\
&& \TB\TB (\lambda x_3.\\
&& \TB\TB\TB\TB \cdots\\
&& \TB\TB\TB\TB (\lambda f_n. (f_n\ \True) (f_n\ \False))\\
&& \TB\TB\TB\TB (\lambda x_n.\\
&& \TB\TB\TB\TB\TB C[(\lambda v.\phi\ v)(\lambda w.wx_1x_2\cdots x_n)])\cdots ))))
\end{eqnarray*}
\caption{\np-hard construction for $k$CFA.}
\end{figure}
For an appropriately chosen program point (label) $\ell$, the cache
location $\cache(v,\ell)$ will contain the set of all possible
closures which are approximated to flow to $v$.  This set is that of
all closures
\begin{displaymath}
\langle(\lambda w.wx_1x_2\cdots x_n),\ce\rangle
\end{displaymath} 
where $\ce$ ranges over all assignments of {\tt True} and {\tt False}
to the free variables (or more precisely assignments of locations in
the table containing {\tt True} and {\tt False} to the free
variables).  The Boolean function $\phi$ is completely linear, as in
the \ptime-completeness proof; the context $C$ uses the Boolean
output(s) as in the conclusion to that proof: mixing in some ML, the
context is:
\begin{quote}
{\small
\begin{verbatim}
- let val (u,u')= [---] in
  let val ((x,y),(x',y'))= (u (f,g), u' (f',g')) in 
      ((x a, y b),(x' a', y' b')) end end;
\end{verbatim}
}
\end{quote}
Again, {\tt a} can only flow as an argument to {\tt f} if {\tt True} flows to {\tt (u,u')},
leaving {\tt (f,g)} unchanged, which can only happen if {\em some} closure
$\langle(\lambda w.wx_1x_2\cdots x_n),\ce\rangle$ provides a satisfying truth valuation for $\phi$.
We have as a consequence:

\begin{theorem}
The control flow problem for 1CFA is \np-hard.
%% Determining the existence of a single flow in 1CFA is \np-hard.
\end{theorem}

Having established this lower bound for 1CFA, we now argue the result
generalizes to all values of $k > 0$.
Observe that by going from $k$CFA to $(k+1)$CFA, further
context-sensitivity is introduced.
But, this added precision can be undone by inserting an identity
function application at the point relevant to answering the flow
question.
This added calling context consumes the added bit of precision in the
analysis and renders the analysis of rest of the program equivalently
to the courser analysis.
Thus, it is easy to insert an identity function into the above
construction such that 2CFA on this program produces the same results
as 1CFA on the original.
So for any $k > 1$, we can construct an \np-hard computation by
following the above construction and inserting $k-1$ application sites
to eat up the precision added beyond 1CFA.
The result is equivalent to 1CFA on the original term, so we conclude:

\begin{theorem}
The control flow problem for $k$CFA is \np-hard, for any $k > 0$.
%% Determining the existence of a single flow in 1CFA is \np-hard.
\end{theorem}

At this point, there is a tension in the results.  On the one hand,
$k$CFA is contained in \exptime; on the other, $k$CFA requires at
least \np-time to compute.
So a gap remains; either the algorithm for computing $k$CFA can be
improved and put into \np, or the lower bound can be strengthened by
exploiting more computational power from the analysis.

We observe that while the computation of the {\em entire} cache
requires exponential time, perhaps the existence of a {\em specific}
flow in it may well be computable in \np.
A non-deterministic algorithm might compute using the ``collection
semantics'' $\ev{t^\ell}{\ce}{\delta}$, but rather than compute entire
sets, {\em choose} the element of the set that bears witness to the
flow.
If so we could conclude $k$CFA is \np-complete.

However, this is not the case.
We show that the lower bound can be improved and $k$CFA is complete
for \exptime.
The improvement relies on simulating an exponential iterator using
analysis.
The following section demonstrates the core of the idea.

\section{Nonlinearity and Cartesian Products:\texorpdfstring{\\ \qquad}{} a toy calculation, with insights}
\label{sec-toy}

A good proof has, at its heart, a small and simple idea that makes it
work.  For our proof, the key idea is how the approximation of
analysis can be {\em leveraged} to provide computing power {\em above
and beyond} that provided by evaluation.  The difference between the
two can be illustrated by the following term:
\begin{displaymath}
\begin{array}{l}
(\lambda f. (f\;\mbox{\tt True}) (f\;\mbox{\tt False}))\\
(\lambda x.\Implies\, x\, x)
\end{array}
\end{displaymath}
Consider evaluation: Here $\Implies\, x\, x$ (a tautology) is evaluated
twice, once with $x$ bound to \True, once with $x$ bound to \False.
But in both cases, the result is \True.  Since $x$ is bound to \True\
or \False\, both occurrences of $x$ are bound to \True\ or to
\False---but it is never the case, for example, that the first
occurrence is bound to \True, while the second is bound to \False.
The values of each occurrence of $x$ is dependent on the other.

On the other hand, consider what flows out of $\Implies\,x\, x$ according
1CFA: both \True\ and \False.  Why? The approximation incurs analysis
of $\Implies\,x\, x$ for $x$ bound to \True\ and \False, but it considers
{\em each occurrence of $x$ as ranging over \True\ and \False,
independently}.  In other words, for the set of values bound to $x$,
we consider their {\em cross product} when $x$ appears nonlinearly.
The approximation permits one occurrence of $x$ be bound to \True\
while the other occurrence is bound to \False; and somewhat alarmingly,
$\Implies\,\True\,\False$ causes \False\ to flow out.  Unlike in normal
evaluation, where within a given scope we know that multiple
occurrences of the same variable refer to the same value, in the
approximation of analysis, multiple occurrences of the same variable
range over {\em all} values that they are possible bound to {\em
independent of each other}.  

Now consider what happens when the program is expanded as follows:
\begin{displaymath}
\begin{array}{l}
(\lambda f. (f\,\mbox{\tt True}) (f\,\mbox{\tt False}))\\
(\lambda x. (\lambda p. p (\lambda u. p (\lambda v. \Implies\, u v))) (\lambda w.wx))
\end{array}
\end{displaymath}
Here, rather than pass $x$ directly to \Implies, we construct a unary
tuple $\lambda w.wx$.  The tuple is used nonlinearly, so $p$ will
range over {\em closures} of $\lambda w.wx$ with $x$ bound to \True\
and \False, again, independently.

A closure can be approximated by an exponential number of values.  For
example, $\lambda w.wz_1z_2\dots z_n$ has $n$ free variables, so there
are an exponential number of possible environments mapping these
variables to program points (contours of length 1).  If we could apply
a Boolean function to this tuple, we would effectively be evaluating
all rows of a truth table; following this intuition leads to
\np-hardness of the 1CFA control flow problem.

Generalizing from unary to $n$-ary tuples in the above example, an
exponential number of closures can flow out of the tuple.  For a
function taking two $n$-tuples, we can compute the function on the
cross product of the exponential number of closures.

This insight is the key computational ingredient in simulating exponential
time, as we describe in the following section.

\section{\kcfa\ is \exptime-hard}
\label{sec-kcfa-exptime-hard}

\subsection{Approximation and \exptime}
%% \subsection{Simulating Turing machines for exponential time}

Recall the formal definition of a Turing machine: a 7-tuple
\[
\langle Q,\Sigma,\Gamma,\delta,q_0,q_a,q_r\rangle
\]
where $Q$, $\Sigma$, and $\Gamma$ are finite sets,  $Q$ is the set of
machine states (and $\{q_0,q_a,q_r\}\subseteq Q$),  $\Sigma$ is the
input alphabet, and $\Gamma$ the tape alphabet, where
$\Sigma\subseteq\Gamma$.  The states $q_0$, $q_a$, and $q_r$ are the
machine's initial, accept, and reject states, respectively.
The complexity class \exptime\ denotes the languages
that can be decided by a Turing machine in time  exponential in the
input length.  

Suppose we have a deterministic Turing machine $M$ that accepts or
rejects its input $x$ in time $2^{p(n)}$, where $p$ is a polynomial
and $n=|x|$.  We want to simulate the computation of $M$ on $x$ by
$k$CFA analysis of a $\lambda$-term $E$ dependent on $M,x,p$, where a
particular closure will flow to a specific program point iff $M$
accepts $x$.  It turns out that $k=1$ suffices to carry out this
simulation.  The construction, computed in logarithmic space, is
similar for all constant $k>1$ modulo a certain amount of padding
as described in \autoref{sec-kcfa-np-hard}.

\subsection{Coding Machine IDs}

The first task is to code machine IDs.  Observe that each
value stored in the abstract cache $\cache$ is a {\em closure}---a
$\lambda$-abstraction, together with an environment for its free variables.
The number of such abstractions is bounded by the program size, as is the {\em domain}
of the environment---while the number of such {\em environments} is exponential in the program size.
(Just consider a program of size $n$ with, say, $n/2$ free variables mapped to only 2 program points
denoting bindings.)

Since a closure only has polynomial size, and a Turing machine ID has
exponential size, we represent the latter by splitting its information
into an exponential number of closures.  Each closure represents a
tuple $\langle T,S,H,C,b\rangle$, which can be read as
\begin{quotation}\it
``At time $T$, Turing machine $M$ was in state $S$, the tape position
was at cell $H$, and cell $C$ held contents $b$.''
\end{quotation}
$T$, $S$, $H$, and $C$ are blocks of bits ($\Zero \equiv \True$,
$\One\equiv\False$) of size polynomial in the input to the Turing machine.  As such, each block
can represent an exponential number of values.
A single machine ID is represented
by an exponential number of tuples (varying $C$ and $b$).  Each such tuple can in turn be
coded as a $\lambda$-term $\lambda w.wz_1z_2\cdots z_N$, where $N=O(p(n))$.

We still need to be able to generate an exponential number of closures
for such an $N$-ary tuple.  The construction is only a modest,
iterative generalization of the construction in our toy calculation
above:
\begin{figure}[h]
\begin{displaymath}
\begin{array}{rcl}
\ & \ & (\lambda f_1.(f_1\;\Zero)(f_1\;\One))\\
\ & \ & (\lambda z_1.\\
\ & \ &\quad(\lambda f_2.(f_2\;\Zero)(f_2\;\One))\\
\ & \ &\quad(\lambda z_2.\\
\ & \ &\qquad\cdots\\
\ & \ &\quad\qquad(\lambda f_N. (f_N\;\Zero) (f_N\;\One))\\
\ & \ &\quad\qquad(\lambda z_N.((\lambda x.x)(\lambda w.wz_1z_2\cdots z_N))^{\ell})\cdots))
\end{array}
\end{displaymath}
\caption{Generalization of toy calculation for $k$CFA.}
\label{fig-kcfa-construction}
\end{figure}

In the inner subterm,
\begin{displaymath}
((\lambda x.x)(\lambda w.wz_1z_2\cdots z_N))^{\ell},
\end{displaymath}
the function $\lambda x.x$ acts as a very important form of {\em
  padding}.  Recall that this is $k$CFA with $k=1$---the expression
$(\lambda w.wz_1z_2\cdots z_N)$ is evaluated an exponential number of
times---to see why, normalize the term---but in each instance, the
contour is always $\ell$.  (For $k>1$, we would just need more padding
to evade the {\em polyvariance} of the flow analyzer.)  As a
consequence, each of the (exponential number of) closures gets put in
the {\em same} location of the abstract cache $\cache$, while they are
placed in unique, {\em different} locations of the exact cache
$\ecache$.  In other words, the approximation mechanism of $k$CFA
treats them as if they are all the same.  (That is why they are put in
the same cache location.)

\subsection{Transition Function}

Now we define a binary transition function $\delta$, which does a {\em
  piecemeal} transition of the machine ID.  The transition function is
represented by three rules, identified uniquely by the time stamps $T$
on the input tuples.

The first {\em transition rule} is used when the
tuples agree on the time stamp $T$, and the head and cell address of the
first tuple coincide:
\begin{displaymath}
\begin{array}{l}
\delta\langle T,S,H,H,b\rangle \langle T,S',H',C',b'\rangle\; = \\
\qquad\qquad\qquad\qquad % hack
\langle T+1, \delta_Q(S,b), \delta_{LR}(S,H,b), H, \delta_\Sigma(S,b) \rangle
\end{array}
\end{displaymath}
This rule {\em computes} the transition to the next ID.  The first
tuple has the head address and cell address coinciding, so it has all
the information needed to compute the next state, head movement, and
what to write in that tape cell.  The second tuple just marks that
this is an instance of the {\em computation} rule, simply indicated by
having the time stamps in the tuples to be identical.  The Boolean functions
$\delta_Q,\delta_{LR},\delta_\Sigma$ compute the next state, head position, and 
what to write on the tape.

The second {\em communication rule} is used when the tuples have time
stamps $T+1$ and $T$: in other words, the first tuple has information
about state and head position which needs to be communicated to every
tuple with time stamp $T$ holding tape cell information for an arbitrary such cell, as it
gets updated to time stamp $T+1$:
\begin{displaymath}
\begin{array}{l}
\delta\langle T+1,S,H,C,b\rangle \langle T,S',H',C',b'\rangle = \langle T+1, S, H, C',b'\rangle\\
\qquad\qquad\qquad\qquad\qquad\qquad\qquad\qquad\qquad\qquad % hack
(H'\not=C')
\end{array}
\end{displaymath}
(Note that when $H'=C'$, we have already written the salient tuple using the transition rule.)
This rule {\em communicates} state and head position (for the first
tuple computed with time stamp $T+1$, where the head and cell address
coincided) to all the other tuples coding the rest of the Turing
machine tape.  

Finally, we define a {\em catch-all rule}, mapping any other
pairs of tuples (say, with time stamps $T$ and $T+42$) to some
distinguished null value (say, the initial ID).  We need this rule just to make sure that
$\delta$ is a totally defined function.
\begin{displaymath}
\begin{array}{l}
\delta\langle T,S,H,C,b\rangle \langle T',S',H',C',b'\rangle\; =\ \Null\qquad\qquad\qquad\qquad\\
\qquad\qquad\qquad\qquad\qquad\qquad\qquad\qquad % hack
(T\not=T'\mbox{\ and\ } T\not=T'+1)
\end{array}
\end{displaymath}

Clearly, these three rules can be coded by a single Boolean circuit,
and we have all the required Boolean logic at our disposal from
\autoref{sec:circuits}.

Because $\delta$ is a binary function, we need to compute a {\em cross
  product} on the coding of IDs to provide its input.  The transition
function is therefore defined as in \autoref{fig-transition}.
\begin{figure}[h]
\begin{displaymath}
\begin{array}{rcl}
\Phi &\equiv& \lambda p.\\
\ &\ &\quad\mbox{let }\langle u_1,u_2,u_3,u_4,u_5\rangle = \Copy_5\; p\mbox{ in}\\
\ &\ &\quad\mbox{let }\langle v_1,v_2,v_3.v_4,v_5\rangle = \Copy_5\; p\mbox{ in}\\
\ &\ &\qquad(\lambda w . w(\phi_T u_1 v_1) (\phi_S u_2 v_2) \dots (\phi_b u_5 v_5))\\
\ &\ &\qquad(\lambda w_T. \lambda w_S. \lambda w_H. \lambda w_C. \lambda w_b.\\
\ &\ &\quad\qquad w_T(\lambda z_1.\lambda z_2\dots\lambda z_T.\\
\ &\ &\qquad\qquad w_S(\lambda z_{T+1}.\lambda z_{T+2}\dots\lambda z_{T+S}.\\
\ &\ &\quad\qquad\qquad \dots\\
\ &\ &\qquad\qquad\qquad w_b(\lambda z_{C+1}.\lambda z_{C+2}\dots\lambda z_{C+b=m}.\\
\ &\ &\quad\qquad\qquad\qquad \lambda w.wz_1z_2\dots z_m)\dots)))
\end{array}
\end{displaymath}
\caption{Turing machine transition function construction.}
\label{fig-transition}
\end{figure}
The {\tt Copy} functions just copy enough of the input for the separate calculations to
be implemented in a linear way.
Observe that this $\lambda$-term is entirely linear {\em except} for
the two occurrences of its parameter $p$.  In that sense, it serves a
function analogous to $\lambda x.\Implies\, x\, x$ in the toy
calculation.  Just as $x$ ranges there over the closures for $\True$
and for $\False$, $p$ ranges over all possible IDs flowing to the
argument position.  Since there are two occurrences of $p$, we have
two entirely separate iterations in the $k$CFA analysis.  These
separate iterations, like nested ``for'' loops, create the equivalent
of a cross product of IDs in the ``inner loop'' of the flow analysis.

\subsection{Context and Widget}

The context for the Turing machine simulation needs to set up the
initial ID and associated machinery, extract the Boolean value telling
whether the machine accepted its input, and feed it into the flow
widget that causes different flows depending on whether the value flowing in is
$\True$ or $\False$.
\begin{figure}
\begin{displaymath}
\begin{array}{rcl}
C &\equiv & (\lambda f_1.(f_1\;\Zero)(f_1\;\One))\\
\ & \ & (\lambda z_1.\\
\ & \ &\quad(\lambda f_2.(f_2\;\Zero)(f_2\;\One))\\
\ & \ &\quad(\lambda z_2.\\
\ & \ &\qquad\cdots\\
\ & \ &\quad\qquad(\lambda f_N. (f_N\;\Zero) (f_N\;\One))\\
\ & \ &\quad\qquad(\lambda z_N.((\lambda x.x) (\Widget (\Extract[\;]))^\ell)^{\ell'})\cdots))
\end{array}
\end{displaymath}
\caption{\exptime-hard construction for $k$CFA.}
\end{figure}
In this code, the $\lambda x.x$ (with label $\ell'$ on its application) serve as padding, so that the term within is always applied in the same contour.
\Extract\ extracts a final ID, with its time stamp, and checks if it
codes an accepting state, returning \True\ or \False\ accordingly.
\Widget\ is our standard control flow test.  The context is
instantiated with the coding of the transition function, iterated over
an initial machine ID,
\begin{displaymath}
\begin{array}{l}
%% 2^n\;\Phi\;\langle\Zero\dots\Zero, Q_0, H_0, \lambda w.wz_1 z_2\dots z_N, \Zero\rangle,
2^n\;\Phi\; \lambda w.w\Zero\dots\Zero\cdots Q_0\cdots H_0\cdots z_1 z_2\dots z_N \Zero,
\end{array}
\end{displaymath}
where $\Phi$ is a coding of transition function for $M$.  The
$\lambda$-term $2^n$ is a fixed point operator for $k$CFA, which can
be assumed to be either $\mathbf{Y}$, or an exponential function
composer.  There just has to be enough iteration of the transition
function to produce a fixed point for the flow analysis.

To make the coding easy, we just assume without loss of generality
that $M$ starts by writing $x$ on the tape, and then begins the
generic exponential-time computation.  Then we can just have all
zeroes on the initial tape configuration.

\begin{lemma}
  For any Turing machine $M$ and input $x$ of length $n$, where $M$
  accepts or rejects $x$ in $2^{p(n)}$ steps, there exists a
  logspace-constructable, closed, labeled $\lambda$-term $e$ with distinguished label
  $\ell$ such that in the $k$CFA analysis of $e$ ($k>0$), $\True$ flows into
  $\ell$ iff $M$ accepts $x$.
\end{lemma}
\begin{theorem}
The control flow problem for $k$CFA is complete
for \exptime\ for any $k > 0$.
\end{theorem}

\section{Exact \kcfa\ is \ptime-complete}

At the heart of the \exptime-completeness result is the idea that the
{\em approximation} inherent in abstract interpretation is being
harnessed for computational power, quite apart from the power of {\em
exact} evaluation.  To get a good lower bound, this is necessary: it
turns out there is a dearth of computation power when $k$CFA
corresponds with evaluation, i.e.~when the analysis is exact.

As noted earlier, approximation arises from the truncation of contours
during analysis.  Consequently, if truncation never occurs, the
instrumented interpreter and the abstract interpreter produce
identical results for the given program.  But what can we say about
the complexity of these programs?  In other words, what kind of
computations can $k$CFA analyze exactly when $k$ is a constant,
independent of the program analyzed?  What is the intersection between
the abstract and concrete interpreter?

An answer to this question provides another point in the
characterization of the expressiveness of an analysis.  For 0CFA, the
answer is \ptime\ since the evaluation of linear terms is captured.
For $k$CFA, the answer remains the same.

For any fixed $k$, $k$CFA can only analyze polynomial time programs
exactly, since, in order for an analysis to be exact, there can only
one entry in each cache location, and there are only $n^{k+1}$
locations. But from this it is clear that only through the use of
approximation that a exponential time computation can be simulated,
but this computation has little to do with the actual running of the
program.  A program that runs for exponential time cannot be analyzed
exactly by $k$CFA for any constant $k$.

Contrast this with ML-typability, for example, where the evaluation of
programs that run for exponential time can be simulated via type
inference.

Note that if the contour is never truncated, every program point is
now approximated by at most one closure (rather than an exponential
number of closures).  The size of the cache is then bounded by a
polynomial in $n$; since the cache is computed monotonically, the
analysis and the natural related decision problem is constrained by
the size and use of the cache.

\begin{proposition}
Deciding the control flow problem for exact $k$CFA is complete for
\ptime.
\end{proposition}

This proposition provides a characterization of the computational
complexity (or expressivity) of the language evaluated by the
instrumented evaluator $\mathcal{E}$ of section
\autoref{section-instrumented} as a function of the contour length.

It also provides an analytic understanding of the empirical
observation researchers have made: computing a more precise analysis
is often cheaper than performing a less precise one, which ``yields
coarser approximations, and thus induces more merging. More merging
leads to more propagation, which in turn leads to more reevaluation''
\cite{wright-jagannathan-toplas98}.  \citet{might-shivers-icfp06} make
a similar observation: ``imprecision reinforces itself during a flow
analysis through an ever-worsening feedback loop.''  This
ever-worsening feedback loop, in which we can make \False\
(spuriously) flow out of $\Implies\,x\,x$, is the critical ingredient
in our \exptime\ lower bound.

Finally, the asymptotic differential between the complexity of exact
and abstract interpretation shows that abstract interpretation is
strictly more expressive, for any fixed $k$.

%\dvhnote{Need also to insert $n$CFA stuff from ICFP'07}

\section{Discussions}

% \paragraph{Discussion:} 

We observe an ``exponential jump'' between contour length and
complexity of the control flow decision problem for every
polynomial-length contour, including contours of constant length.
Once $k=n$ (contour length equals program size), an exponential-time
hardness result can be proved which is essentially a linear circuit
with an exponential iterator---very much like \citet{mairson-popl90}.
When the contours are exponential in program length, the decision
problem is doubly exponential, and so on.

The reason for this exponential jump is the cardinality of
environments in closures.  This, in fact, is the bottleneck for
control flow analysis---it is the reason that 0CFA (without closures)
is tractable, while 1CFA is not.  If $f(n)$ is the
contour length and $n$ is the program length, then
\begin{displaymath}
|\CEnv| = |\Var \rightarrow \Delta^{\leq f(n)}| = (n^{f(n)})^n = 2^{f(n)n \lg n}
\end{displaymath}
This cardinality of environments effectively determines the size of
the universe of values for the abstract interpretation realized by CFA.

When $k$ is a constant, one might ask why the inherent complexity is
exponential time, and not more---especially since one can iterate (in
an untyped world) with the $\mathbf{Y}$ combinator.  Exponential time
is the ``limit'' because with a polynomial-length tuple (as
constrained by a logspace reduction), you can only code an exponential
number of closures.

The idea behind $k$CFA is that the precision of could {\em dialed up},
but there are essentially two settings to the $k$CFA hierarchy: {\em
  high} ($k > 0$, \exptime) and {\em low} $(k = 0)$.  We can see, from
a computational complexity perspective, that 0CFA is strictly less
expressive than $k$CFA.  In turn, $k$CFA is strictly less expressive
than, for example, Mossin's flow analysis \citeyearpar{mossin-sas97}.
Mossin's analysis is a stronger analysis in the sense that it is exact
for a larger class of programs than 0CFA or $k$CFA---it exact not only
for linear terms, but for all simply-typed terms.  In other words, the
flow analysis of simply-typed programs is synonymous with running the
program, and hence non-elementary.  This kind of expressivity is also
found in Burn-Hankin-Abramsky-style strictness analysis
\citeyearpar{burn-hankin-abramsky}.  But there is a considerable gap
between $k$CFA and these more expressive analyses.  What is in between
and how can we build a real {\em hierarchy} of static analyses that
occupy positions within this gap?

This argues that the relationship between dial level $N$ and $N+1$
should be exact. This is the case with say simple-typing and
ML-typing.  (ML = simple + let reduction).  There is no analogous
relationship known between $k$ and $k+1$CFA.  A major computational
expense in $k$CFA is the approximation engendering further
approximation and re-evaluation.  Perhaps by staging analysis into
polyvariance and approximation phases, the feedback loop of spurious
flows can be avoided. 

If you had an analysis that did some kind of exact, bounded,
evaluation of the program and then analyzed the residual with 0CFA,
you may have a far more usable analysis than with the $k$CFA
hierarchy.

The precision of $k$CFA is highly sensitive to syntactic structure.
Simple program refactorings such as $\eta$-expansion have drastic
effects on the results of $k$CFA and can easily undermine the added
work of a more and more precise analysis.  Indeed, we utilize these
simple refactorings to undermine the added precision of $k$CFA to
generalize the hardness results from the case of 1CFA to all $k > 0$
CFA.  But an analysis that was robust in the face of these
refactorings could undermine these lower bounds.

In general, techniques that lead to increased precision will take
computational power {\em away} from our lower bound constructions.
For instance, it is not clear what could be said about lower bounds on
the complexity of a variant of $k$CFA that employed abstract garbage
collection \cite{might-shivers-icfp06}, which allows for the safe
removal of values from the cache during computation.  It is critical
in the lower bound construction that what goes into the cache, stays
in the cache.

L\'evy's notion of labeled reduction \citeyearpar{levy-phd, levy-80}
provides a richer notion of ``instrumented evaluation'' coupled with a
richer theory of exact flow analysis, namely the geometry of
interaction \cite{girard-goi89, gonthier-abadi-levy-popl92}.  With the
proper notion of abstraction and simulated reduction, we should be
able to design more powerful flow analyses, filling out the hierarchy
from 0CFA up to the expressivity of Mossin's analysis in the limit.

\section{Conclusions}

Empirically observed increases in costs can be understood analytically
as {\em inherent in the approximation problem being solved}.

We have given an exact characterization of the $k$CFA approximation
problem.  The \exptime\ lower bound validates empirical observations
and shows that there is no tractable algorithm for $k$CFA.

The proof relies on previous insights about linearity, static
analysis, and normalization (namely, when a term is linear, static
analysis and normalization are synonymous); coupled with new insights
about using nonlinearity to realize the full computational power of
approximate, or abstract, interpretation.

Shivers wrote in his best of PLDI retrospective
\citeyearpar{shivers-sigplan04},
\begin{quotation}%%
Despite all this work on formalising CFA and speeding it up, I have
been disappointed in the dearth of work extending its {\em
power}.
\end{quotation}
This work has shown that work spent on speeding up $k$CFA is an
exercise in futility; there is no getting around the exponential
bottleneck of $k$CFA.  The one-word description of the bottleneck is {\em closures},
which do not exist in 0CFA, because free variables in a closure would necessarily map to $\epsilon$,
and hence the environments are useless.

This detailed accounting of the ingredients that combine to make
$k$CFA hard, when $k > 0$, should provide guidance in designing new
abstractions that avoid computationally expensive components of
analysis.
A lesson learned has been that {\em closures}, as they exist when
$k>0$, result in an exponential value space that can be harnessed for
the \exptime\ lower-bound construction.

% Finally, we need to emphasize the importance of linearity in static
% analysis.  Static analysis makes approximations to be tractable, but
% with linear terms, there is no approximation.  We carefully admitted a
% certain, limited nonlinearity in order to increase the lower bound.

\chapter{Related Work}

This dissertation draws upon several large veins of research.  At the
highest level, this includes complexity, semantics, logic, and program
analysis.  This chapter surveys related work to sketch applications
and draw parallels with existing work.

\section{Monovariant Flow Analysis}

In the setting of first-order programming languages, \citet{reps-ai96}
gives a complexity investigation of program analyses and shows
interprocedural slicing to be complete for \ptime\ and that obtaining
``meet-over-all-valid-paths'' solutions of distributive data-flow
analysis problems \cite{hecht} is \ptime-hard in general, and
\ptime-complete when there are only a finite number of data-flow
facts.  A circuit-value construction by interprocedural data-flow
analysis is given using Boolean circuitry encoded as call graph
gadgets, similar in spirit to our constructions in
\autoref{chapter-0cfa}.

In the setting of higher-order programming languages,
\citet{melski-reps00} give a complexity investigation of 0CFA-like,
inclusion-based monovariant flow analysis for a functional language
with pattern matching.
The analysis takes the form of a constraint satisfaction problem and
this satisfaction problem is shown to be complete for \ptime.
See  \autoref{sec-context-free-lang-reach} for further discussion.

The impact of pattern matching on analysis complexity is further
examined by \citet{heintze-mcallester-icfp97}, which shows how deep
pattern matching affects monovariant analysis, making it complete for
\exptime.

\section{Linearity and Static Analysis}

% \subsection{Linearity and Flow analysis}

\citet{jagannathan-etal-popl98} observe that flow analysis, which is a
{\em may} analysis, can be adapted to answer {\em must} analysis
questions by incorporating a ``per-program-point {\em variable
  cardinality map}, which indicates whether all reachable environments
binding a variable $x$ hold the same value.  If so, $x$ is marked
single at that point; otherwise $x$ is marked multiple.''  The
resulting must-alias information facilities program optimization such
as lightweight closure conversion \cite{steckler-wand-toplas97}.  This
must analysis is a simple instance of tracking linearity information
in order to increase the precision of the analysis.
\citet{might-shivers-icfp06} use a similar approach of {\em abstract
  counting}, which distinguish singleton and non-singleton flow sets,
to improve flow analysis precision.

Something similar can be observed in 0CFA without cardinality maps;
singleton flow sets $\cache(\ell) = \{\lambda x.e\}$, which are
interpreted as ``the expression labelled $\ell$ {\em may} evaluate to
one of $\{ \lambda x.e \}$,'' convey {\em must} information.  The
expression labelled $\ell$ either diverges or evaluates to $\lambda
x.e$.  When $\lambda x.e$ is linearly closed---the variables map to
singleton sets containing linear closures---then the run-time value
produced by the expression labelled $\ell$ can be determined
completely at analysis time.  The idea of taking this special case of
{\em must} analysis within a {\em may} analysis to its logical
conclusion is the basis of \autoref{chapter-0cfa}.

\citet{damian-danvy-jfp03} have investigated the impact of linear
$\beta$-reduction on the result of flow analysis and show how
leastness is preserved.  The result is used to show that leastness is
preserved through CPS and administrative reductions, which are linear.

%\subsection{Linearity and Type analysis}

An old, but key, observation about the type inference problem for
simply typed $\lambda$-terms is that, when the term is linear (every
bound variable occurs exactly once), the most general type and normal
form are isomorphic
\cite{hindley-tcs89,hirokawa-tacs91,henglein-mairson-popl91,mairson-jfp04}.%
\footnote{The seed of inspiration for this work came from a close
  study of \citet{mairson-jfp04} in the Spring of 2005 for a seminar
  presentation given in a graduate course on advanced topics in
  complexity theory at the University of Vermont.}

The observation translates to flow analysis, as shown in
\autoref{chapter-0cfa}, but in a typed setting, it also scales to
richer systems.
The insight leads to an elegant reproof of the \exptime-hardness of
ML-type inference result from \citet{mairson-popl90}
\cite{henglein-unpub90}.
It was used to prove novel lower bounds on type inference for System
$F_\omega$ \cite{henglein-mairson-popl91} and rank-bound intersection
type inference \cite{neergaard-mairson-icfp04}.
See \autoref{sec-quantifier-elim} for further discussion.

\section{Context-Free-Language Reachability}
\label{sec-context-free-lang-reach}

\citet{melski-reps00} show the interconvertibility between a number of
set-constraint problems and the context-free-language (CFL)
reachability problem, which is known to be complete for \ptime\
\cite{ullman-allen-sfcs86}.  \citet{heintze-lfp94} develops a
set-based approach to flow analysis for a simple untyped functional
language with functions, applications, pattern-matching, and
recursion.  The analysis works by making a pass over the program,
generating set constraints, which can then be solved to compute flow
analysis results.  Following Melski and Reps, we refer to this
constraint system as ML set-constraints.  For the subset of the
language considered in this dissertation, solving these constraints
computes a monovariant flow analysis that coincides with 0CFA.

In addition to the many set-constraint problems considered, which have
applications to static analysis of first-order programming languages,
\citet[section 5]{melski-reps00} also investigate the problem of
solving the ML set-constraints used by Heintze.  They show this class
of set-constraint problems can be solved in cubic time with respect to
the size of the input constraints.  Since \citet{heintze-lfp94} gave a
$O(n^3)$ algorithm for solving these constraints, Melski and Reps'
result demonstrates the conversion to CFL-reachability preserves
cubic-solvability, while allowing CFL-reachability formulations of
static analyses, such as program slicing and shape analysis, to be
brought to bear on higher-order languages, where previously they had
only been applied in a first-order setting.

After showing ML set-constraints can be solved using CFL-reachability,
\citet[section 6]{melski-reps00} also prove the converse holds:
CFL-reachability problems can be solved by reduction to ML
set-constraint problems while preserving the worse-case asymptotic
complexity.  By the known \ptime-hardness of CFL-reachability, this
implies ML set-constraint satisfaction is \ptime-complete.  It does
not follow, however, that 0CFA is also \ptime-complete.

It is worth noting that Melski and Reps are concerned with constraint
satisfaction, and not directly with flow analysis---the two are
intimately related, but the distinction is important.  It follows as a
corollary that since ML set-constraints can be solved, through a
reduction to CFL-reachability, flow analysis can be performed in cubic
time.  \citet[page 314]{heintze-lfp94} observes that the size of the
set-constraint problem generated by the initial pass of the program is
linear in the size of the program being analyzed.  Therefore it is
straightforward to derive from the ML set-constraint to
CFL-reachability reduction the (known) inclusion of 0CFA in \ptime.

In the other direction, it is not clear that it follows from the
\ptime-hardness of ML set-constraint satisfaction that flow analysis
of Heintze's subject language is \ptime-hard.  Melski and Reps use the
constraint language directly in their encoding of CFL-reachability.
What remains to be seen is whether there are programs which could be
constructed that would induce these constraints.  Moreover, their
reduction relies soley on the ``case'' constraints of Heintze, which
are set constraints induced by pattern matching expressions in the
source language.

If the source language lacks pattern matching, the Boolean circuit
machinery of Melski and Reps can no longer be constructed since no
expressions induce the needed ``case'' constraints.  For this
language, the \ptime-hardness of constraint satisfaction and 0CFA does
not follow from the results of Melski and Reps.

This reiterates the importance of Reps' own observation that analysis
problems should be formulated in ``trimmed-down form,'' which both
leads to a wider applicability of the lower bounds and ``allows one to
gain greater insight into exactly what aspects of an [\dots] analysis
problem introduce what computational limitations on algorithms for
these problems,'' \cite[section 2]{reps-ai96}.

By considering only the core subset of every higher-order programming
language and relying on the specification of analysis, rather than its
implementation technology, the 0CFA \ptime-completeness result implies
as an immediate corollary the \ptime-completeness of the ML
set-constraint problem considered by Melski and Reps.  Moreover, as we
have seen, our proof technique of using linearity to subvert
approximation is broadly applicable to further analysis
approximations, whereas CFL-reachability reductions must be replayed
{\em mutatis mutandis}.

%\section{Two-way non-deterministic pushdown automata}
\section{2NPDA and the Cubic Bottleneck}
\label{sec-2npda}

The class \twonpda\ contains all languages that are recognizable by a
two-way non-deterministic push-down automaton.\footnote{This section
  is derived from material in \citet{midtgaard-vanhorn-tr09}.}
The familiar PDAs found in undergraduate textbooks~\citep{martin-97},
both deterministic and non-deterministic, are one-way: consuming their
input from left-to-right.
In contrast, two-way NPDAs accept their input on a read-only input
tape marked with special begin and end markers, on which they can move
the read-head forwards, backwards, or not at all.

Over a decade ago, \citet{heintze-mcallester-lics97} proved deciding a
monovariant flow analysis problem to be at least as hard as \twonpda,
and argued this provided evidence the ``cubic bottleneck'' of flow
analysis was unlikely to be overcome since the best known algorithm
for \twonpda\ was cubic and had not been improved since its
formulation by \citet{aho-hopcroft-ullman-ic68}.
This statement was made by several other %subsequent followup
papers~\citep{neal-89,heintze-mcallester-lics97,heintze-mcallester-pldi97,
%McAllester:SAS99,
  melski-reps00,mcallester-jacm02,vanhorn-mairson-sas08}.
Yet collectively, this is simply an oversight in the history of
events; \citet{rytter-ic85} improved the cubic bound by a logarithmic
factor.

Since then, Rytter's technique has been used in various contexts: in
diameter verification, in Boolean matrix multiplication, and for the
all pairs shortest paths
problem~\citep{basch-etal-95,zwick-alg06,chan-stoc07}, as
well as for reachability in recursive state
machines~\citep{chaudhuri-popl08},
and for maximum node-weighted $k$-clique~\citep{vassilevska-ipl09} to name
a few.
In particular, \citet{chaudhuri-popl08} recently used Rytter's
techniques to formulate a subcubic algorithm for the related problem
of context-free language (CFL) reachability.
Perhaps unknown to most, indirectly this constitutes the first
subcubic inclusion-based flow analysis algorithm when combined with a
reduction due to \citet{melski-reps00}.

The logarithmic improvement can be carried over to the flow analysis
problem directly, by applying the same known set compression
techniques \citet{rytter-ic85} applies to improve deciding \twonpda.
Moreover, refined analyses similar to
\citet{heintze-mcallester-icfp97} that incorporate notions of
reachability to improve precision remain subcubic.  See
\citet{midtgaard-vanhorn-tr09} for details.

0CFA is complete for both \twonpda\ \citep{heintze-mcallester-lics97} and
\ptime\ (\autoref{chapter-0cfa}).  Yet, it is not clear what relation
these class have to each other.

The \twonpda\ inclusion proof of Heintze and McAllester is sensitive to
representation choices and problem formulations. They use an encoding
of programs that requires a non-standard bit string labelling scheme
in which identical subterms have the same labels. The authors remark
that without this labelling scheme, the problem ``appears not to be in
\twonpda.'' 

Moreover, the notions of reduction employed in the definitions of
\twonpda-hardness and \ptime-hardness rely on different computational
models.  For a problem to be \twonpda-hard, all problems in the class
must be reducible to it in $O(nR(\log n))$ time on a RAM, where $R$ is
a polynomial. Whereas for a problem to be \ptime-hard, all problems in
the class must be reducible to it using a $O(\log n)$ space work-tape
on a Turing machine.

\section{\kcfa}
\label{sec-rel-kcfa}

Our coding of Turing machines is descended from work on Datalog
(Prolog with variables, but without constants or function symbols), a
programming language that was of considerable interest to researchers
in database theory during the 1980s; see \citet{hkmv,gmsv}.

In $k$CFA and abstract interpretation more generally, an expression
can evaluate to a set of values from a finite universe, clearly motivating
the idiom of programming with sets.  Relational database queries
take as input a finite set of tuples, and compute new tuples from
them; since the universe of tuples is finite and the computation is
monotone, a fixed-point is reached in a finite number of iterations.  The
machine simulation here follows that framework very closely.  Even the
idea of splitting a machine configuration among many tuples has its
ancestor in \citet{hkmv}, where a ternary ${\tt cons}(A,L,R)$ is used
to simulate a {\tt cons}-cell at memory address $A$, with pointers
$L,R$.  It needs emphasis that the computing with sets described in
this paper has little to do with normalization, and everything to do
with the approximation inherent in the abstract interpretation.

Although $k$CFA and ML-type inference are two static analyses complete
for \exptime\ \cite{mairson-popl90}, the proofs of these respective
theorems is fundamentally different.  The ML proof relies on type
inference simulating exact normalization (analogous to the
\ptime-completeness proof for 0CFA), hence subverting the approximation
of the analysis.  In contrast, the $k$CFA proof harnesses the
approximation that results from nonlinearity.

\section{Class Analysis}

Flow analysis of functional languages is complicated by the fact that
{\em computations are expressible values}.
This makes basic questions about control flow undecidable in the
general case.
But the same is true in object-oriented programs---computations may be
package up as values, passed as arguments, stored in data-structures,
etc.---and so program analyses in object-oriented settings often deal
with the same issues as flow analysis. 
A close analogue of flow analysis is {\em class analysis}.

Expressions in object-oriented languages may have a declared class (or
type) but, at run-time, they can evaluate to objects of every subclass
of the class.
Class analysis computes the actual set of classes that an expression
can have at run-time \cite{johnson-etal-oopsla88,
  chambers-ungar-pldi90, palsberg-schwartzbach-oopsla91,
  bacon-sweeney-oopsla96}.
Class analysis is sometimes called receiver class analysis, type
analysis, or concrete type inference; it informs static method
resolution, inlining, and other program optimizations.

An object-oriented language is higher-order in the same way as a
language with first-class functions and exactly the same circularity
noted by Shivers occurs in the class analysis of an object-oriented
language.

\citet{grove-chambers-toplas01}:
\begin{quote}
  In object-oriented languages, the method invoked by a dynamically
  dispatched message send depends on the class of the object receiving
  the message; in languages with function values, the procedure
  invoked by the application of a computed function value is
  determined by the function value itself. In general, determining the
  flow of values needed to build a useful call graph requires an
  interprocedural data and control flow analysis of the program. But
  interprocedural analysis in turn requires that a call graph be built
  prior to the analysis being performed.
\end{quote}

Ten years earlier, \citet[page 6]{shivers-phd}%
%% It never hurts to kiss a committee member's ass a bit.  Especially
%% when the praise is true.
\footnote{It is a testament to Shivers' power as a writer that his
  original story has been told over and over again in so many places,
  usually with half the style.} had written essentially the same:
\begin{quote}
  So, if we wish to have a control-flow graph for a piece of Scheme
  code, we need to answer the following question: for every procedure
  call in the program, what are the possible lambda expressions that
  call could be a jump to? But this is a flow analysis question! So
  with regard to flow analysis in an HOL, we are faced with the
  following unfortunate situation:
\begin{itemize}
\item In order to do flow analysis, we need a control-flow graph. 
\item In order to determine control-flow graphs, we need to do flow analysis.
\end{itemize}
\end{quote}

Class analysis is often presented using the terminology of type
inference, however these type systems typically more closely resemble
flow analysis: types are finite sets of classes appearing
syntactically in the program and subtyping is interpreted as set
inclusion.

In other words, objects are treated much like functions in the flow
analysis of a functional language---typically both are approximated by
a set of definition sites, i.e. an object is approximated by a set of
class names that appear in the program; a function is approximated by
a set of $\lambda$ occurrences that appear in the program.  In an
object-oriented program, we may ask of a subexpression, what classes
may the subexpression evaluate to?  In a functional language we may
ask, what $\lambda$ terms may this expression evaluate to?  Notice
both are general questions that analysis must answer in a higher order
setting if you want to know about control flow.  To know where control
may transfer to from $(f\;x)$ we have to know what $f$ may be.  To
know where control may transfer to from \verb|f.apply(x)| we have to
know what \verb|f| may be.  In both cases, if we approximate functions
by sets of $\lambda$s and objects by sets of class names, we may
determine a set of possible places in code where control may transfer,
but we will not know about the {\em environment} of this code,
i.e. the environment component of a closure or the record component of
an object.

\citet{spoto-jensen-toplas03} give a reformulation of several class
analyses, including that of \citet{palsberg-schwartzbach-oopsla91,
  bacon-sweeney-oopsla96, diwan-etal-oopsla96}, using abstract
interpretation.

\citet{defouw-etal-popl98} presents a number of variations on the
theme of monovariant class analysis.  They develop a framework that
can be instantiated to obtain inclusion, equality, and optimistic
based class analyses with close analogies to 0CFA, simple closure
analysis, and rapid type analysis \cite{bacon-sweeney-oopsla96},
respectively.
Each of these instantiations enjoy the same asymptotic running times
as their functional language counterparts; cubic, near linear, and
linear, respectively.

Although some papers give upper bounds for the algorithms they
present, there are very few lower bound results in the
literature.\footnote{I was able to find zero papers that deal directly
  with lower bounds on class analysis complexity.}

Class analysis is closely related to {\em points-to} analysis in
object-oriented languages.
``{\em Points-to analysis} is a fundamental static analysis used by
optimizing compilers and software engineering tools to determine the
set of objects whose addresses may be stored in reference variables
and reference fields of objects,'' \cite{milanova-etal-tosem05}.
When a points-to analysis is {\em flow-sensitive}---``analyses take
into account the flow of control between program points inside a
method, and compute separate solutions for these points,''
\cite{milanova-etal-tosem05}---the analysis necessarily involves some
kind of class analysis.

% Interprocedural ``class'' analysis.  (??)

In object-oriented languages, context-sensitive is typically
distinguished as being object-sensitive \cite{milanova-etal-tosem05},
call-site sensitive \cite{grove-chambers-toplas01}, or partially flow
sensitivity \cite{rinetzky-etal-toplas08}.

\citet{grove-chambers-toplas01} provide a framework for a functional
and object-oriented hybrid language that can be instantiated to obtain
a $k$CFA analysis and an object-oriented analogue called {\em
  $k$-$l$-CFA}.  There is a discussion and references in Section
9.1.
In this discussion, \citet{grove-chambers-toplas01} cite
\citet{oxhoj-etal-ecoop92} as giving ``1-CFA extension to Palsberg and
Schwartzbach's algorithm,'' although the paper develops the analysis
as a type inference problem.
Grove and Chambers also cite \citet{vitek-etal-cc92} as one of several
``adaptations of $k$CFA to object-oriented programs,'' and although
this paper actually has analogies to $k$CFA in an object-oriented
setting (they give a call-string approach to call graph context
sensitivity in section 7), it seems to be developed completely
independently of Shivers' $k$CFA work or any functional flow analysis
work.  

% The paper presented a parameterized framework much in the
% spirit of \citet{nielson-nielson-popl97}.
%, which is abstracted over 3 sets\dots

\begin{figure}
{\small
\begin{verbatim}
new Fun<Fun<B,List<B>>,List<B>>() {
 public List<B> apply(Fun<B,List<B>> f1) {
   f1.apply(true);
   return f1.apply(false);
 }
}.apply(new Fun<B,List<B>>() {
  public List<B> apply(final B x1) {
    return
      new Fun<Fun<B,List<B>>,List<B>>() {
       public List<B> apply(Fun<B,List<B>> f2) {
         f2.apply(true);
         return f2.apply(false);
       }
      }.apply(new Fun<B,List<B>>() {
        public List<B> apply(final B x2) {
          return
            ...
            new Fun<Fun<B,List<B>>,List<B>>() {
             public List<B> apply(Fun<B,List<B>> fn) {
               fn.apply(true);
               return fn.apply(false);
             }
            }.apply(new Fun<B,List<B>>() {
              public List<B> apply(final B xn) {
                return
                  new List<B>{x1,x2,...xn};}}
\end{verbatim}}
  \caption{Translation of $k$CFA \exptime-construction into an
    object-oriented language.}
  \label{fig-kcfa-java-construction}
\end{figure}

The construction of \autoref{fig-kcfa-construction} can be translated
in an object-oriented language such as Java, as given in
\autoref{fig-kcfa-java-construction}.\footnote{This translation is
  Java except for the made up list constructor and some abbreviation
  in type names for brevity, i.e.~{\tt B} is shorthand for {\tt
    Boolean}.}  Functions are simulated as objects with an apply
method.  The crucial subterm in \autoref{fig-kcfa-java-construction}
is the construction of the list \verb|{x1,x2,..xn}|, where
\verb|x|$_i$ occur free with the context of the innermost ``lambda''
term, \verb|new Fun() {...}|. To be truly faithful to the original
construction, lists would be Church-encoded, and thus represented with
a function of one argument, which is applied to \verb|x1| through
\verb|xn|.  An analysis with a similar context abstraction to 1CFA
will approximate the term representing the list \verb|x1,x2,...,xn|
with an abstract object that includes 1 bit of context information for
each {\em instance variable}, and thus there would be $2^n$ values
flowing from this program point, one for each mapping \verb|x|$_i$ to
the calling context in which it was bound to either true or false for
all possible combinations.  \citet{grove-chambers-toplas01} develop a
framework for call-graph construction which can be instantiated in the
style of 1CFA and the construction above should be adaptable to show
this instantiation is \exptime-hard.

% \cite{milanova-etal-tosem05} has notes about the ``functional approach
% to context sensitivity'' vs. call chain context sensitivity, which are
% incomparable in terms of precision, Sharir and Pnueli 81.

% ``makes feasible the most precise context-sensitive analyses reported
% in the literature.''
% \cite{bravenboer-smaragdakis-oopsla09}.

A related question is whether the insights about linearity can be
carried over to the setting of pointer analysis in a first-order
language to obtain simple proofs of lower bounds.  If so, is it
possible higher-order constructions can be transformed systematically
to obtain first-order constructions?

Type hierarchy analysis is a kind of class analysis particularly
relevant to the discussion in \autoref{sec:approach} and the broader
applicability of the approach to proving lower bounds employed in
\autoref{chapter-0cfa}.  Type hierarchy analysis is an analysis of
statically typed object-oriented languages that bounds the set of
procedures a method invocation may call by examining the type
hierarchy declarations for method overrides.
``Type hierarchy analysis does not examine what the program actually
does, just its type and method declarations,''
\cite{diwan-etal-oopsla96}.
It seems unlikely that the technique of \autoref{sec:approach} can be
applied to prove lower bounds about this analysis since it has nothing
to do with approximating evaluation.

\section{Pointer Analysis}
\label{sec-pointer-analysis}

Just as flow analysis plays a fundamental role in the analysis of
higher-order functional programs, {\em pointer analysis}\footnote{Also
  known as {\em alias} and {\em points-to} analysis.} plays a
fundamental role in imperative languages with pointers
\cite{landi-phd92} and object-oriented languages, and informs later
program analyses such as live variables, available expressions, and
constant propagation.  Moreover, flow and alias analysis variants are
often developed along the same axes and have natural analogues with
each other.  

For example, Henglein's \citeyearpar{henglein92d} simple closure
analysis and Steensgaard's \citeyearpar{steensgaard-popl96} points-to
analysis are natural analogues.  Both operate in near linear time by
relying on equality-based (rather than inclusion-based) set
constraints, which can be implemented using a union-find
data-structure.  Steensgaard algorithm ``is inspired by Henglein's
\citeyearpar{henglein-fplca91} binding time analysis by type
inference,'' which also forms the conceptual basis for
\citet{henglein92d}.  Palsberg's \citeyearpar{palsberg-toplas95} and
Heintze's \citeyearpar{heintze-lfp94} constraint-based flow analysis
and Andersen's \citeyearpar{andersen-phd94} pointer analysis are
similarly analogous and bear a strong resemblance in their use of
subset constraints.

To get a full sense of the correspondence between pointer analysis and
flow analysis, read their respective surveys in parallel
\cite{hind-paste01,midtgaard-07}.  These comprise major, mostly
independent, lines of research.  Given the numerous analogies, it is
natural to wonder what the pointer analysis parallels are to the
results presented in this dissertation.  The landscape of the pointer
analysis literature is much like that of flow analysis; there are
hundreds of papers; similar, over-loaded, and abused terminology is
frequently used; it concerns a huge variety of tools, frameworks,
notations, proof techniques, implementation techniques, etc.  Without
delving into too much detail, we recall some of the fundamental
concepts of pointer analysis, cite relevant results, and try to more
fully develop the analogies between flow analysis and pointer
analysis.

A pointer analysis attempts to statically determine the possible
run-time values of a pointer.  Given a program and two variables $p$
and $q$, points-to analysis determines if $p$ can point to $q$
\cite{chakaravarthy-popl03}.  It is clear that in general, like all
interesting properties of programs, it is not decidable if $p$ can
point $q$.  A traditional assumption in this community is that all
paths in the program are executable.  However, even under this
conservative assumption, the problem is undecidable.  The history of
pointer analysis can be understood largely in terms of the trade-offs
between complexity and precision.

% There are two kinds of aliasing \cite{landi-loplas92}:
% \begin{itemize}
% \item {\em may alias} occur during some execution of the program
% \item {\em must alias} occur during all executions of the program
% \end{itemize}

Analyses are characterized along several dimensions
\cite{hind-paste01}, but of particular relevance are those of:
\begin{itemize}
\item {\em Equality-based}: assignment is treated as an undirected
  flow of values.
\item {\em Subset-based}: assignment is treated as a directed flow of values.

\item {\em Flow sensitivity}

  A points-to analysis is {\em flow-sensitive} analysis if it is given
  the control flow graph for the analyzed program.  The control flow
  graphs informs the paths considered when determining the points-to
  relation.  A {\em flow-insensitive} analysis is not given the
  control flow graph and it is assumed statements can be executed in
  any order.  See also section 4.4 of \citet{hind-paste01} and section
  2.3 of \citet{rinetzky-etal-toplas08}.

\item {\em Context sensitivity}

calling context is considered when
  analyzing a function so that calls return to their caller.
See also section 4.4 of \cite{hind-paste01}.

\citet{bravenboer-smaragdakis-oopsla09} remark:
\begin{quote}
  In full context-sensitive pointer analysis, there is an ongoing
  search for context abstractions that provide precise pointer
  information, and do not cause massive redundant
  computation.\footnote{That search has been reflected in the
    functional community as well, see for example, \citet{shivers-phd,
      jagannathan-weeks-popl95, banerjee-icfp97, faxen-lomaps97,
      nielson-nielson-popl97, sereni-icfp07, ashley-dybvig-toplas98,
      wright-jagannathan-toplas98, might-shivers-popl06, might-phd}.}
\end{quote}
\end{itemize}

The complexity of pointer analysis has been deeply studied
\cite{meyers-popl81, landi-ryder-popl91, landi-phd92, landi-loplas92,
  choi-etal-popl93,
  ramalingam-toplas94, % a re-proof of \cite{landi-loplas92}.
  horwitz-toplas97, muth-debray-popl00, chaterjee-etal-tse01,
  chakaravarthy-horwitz-ai02, chakaravarthy-popl03,
  rinetzky-etal-toplas08}.

Flow sensitive points-to analysis with dynamic memory is not decidable
\cite{landi-loplas92,ramalingam-toplas94,chakaravarthy-popl03}.  Flow
sensitive points-to analysis without dynamic memory is \pspace-hard
\cite{landi-phd92,muth-debray-popl00}, even when pointers are
well-typed and restricted to only two levels of dereferencing
\cite{chakaravarthy-popl03}.  
Context-sensitive pointer analysis can
be done efficiently in practice
\cite{emami-etal-pldi94,wilson-lam-pldi95}.  
Flow and context-sensitive points-to analysis for Java can be
efficient and practical even for large programs
\cite{milanova-etal-tosem05}.  

See \citet{muth-debray-popl00,chakaravarthy-popl03} for succinct
overview of complexity results and open problems.

\section{Logic Programming}

\citet{mcallester-jacm02} argues ``bottom-up logic program
presentations are clearer and simpler to analyze, for both correctness
and {\em complexity}'' and provides theorems for characterizing their
run-time.  McAllester argues bottom-up logic programming is especially
appropriate for static analysis algorithms.  The paper gives a
bottom-up logic presentation of evaluation (Fig. 4) and flow analysis
(Fig 5.) for the $\lambda$-calculus with pairing and uses the run-time
theorem to derive a cubic upper bound for the analysis.

Recent work by \citet{bravenboer-smaragdakis-oopsla09} demonstrates
how Datalog can be used to specify and efficiently implement pointer
analysis.  
By the \ptime-completeness of Datalog, any analysis that can be
specified is included in \ptime.

This bears a connection to the implicit computational complexity
program, which has sought to develop syntactic means of developing
programming languages that capture some complexity class
\cite{hofmann-hab98,leivant-popl93,hofmann-ic03,kristiansen-tcs04}.
Although this community has focused on general purpose programming
languages---with only limited success in producing usable systems---it
seems that restricting the domain of interest to program analyzers may
be a fruitful line of work to investigate.

The \exptime\ construction of \autoref{sec-kcfa-exptime-hard} has a
conceptual basis in Datalog complexity research \cite{hkmv,gmsv}.
See \autoref{sec-rel-kcfa} for a discussion.

\section{Termination Analysis}

Termination analysis of higher-order programs \cite{jones-bohr-lmcs08,
  sereni-jones-aplas05, giesl-etal-rta06, sereni-icfp07} is inherently
tied to some underlying flow analysis.

Recent work by Sereni and Jones on the termination analysis of
higher-order languages has relied on an initial control flow analysis
of a program, the result of which becomes input to the termination
analyzer \cite{sereni-jones-aplas05,sereni-icfp07}.  Once a call-graph
is constructed, the so-called ``size-change'' principle\footnote{The
  size-change principle has enjoyed a complexity investigation in its
  own right \cite{lee-etal-popl01, ben-amram-etal-toplas07}.} can be
used to show that there is no infinite path of decreasing size through
through the program's control graph, and therefore the program
eventually produces an answer.  This work has noted the inadequacies
of 0CFA for producing precise enough graphs for proving most
interesting programs terminating.  Motivated by more powerful
termination analyses, these researchers have designed more powerful
(i.e., more precise) control flow analyses, dubbed $k$-limited CFA.
These analyses are parametrized by a fixed bound on the depth of
environments, like Shivers' $k$CFA. So for example, in 1-limited CFA,
each variable is mapped to the program point in which it is bound, but
no information is retained about this value's environment.  But unlike
$k$CFA, this ``limited'' analysis is not polyvariant
(context-sensitive) with respect to the most recent $k$ calling
contexts.

A lesson of our investigation into the complexity of $k$CFA is that it
is {\em not} the polyvariance that makes the analysis difficult to
compute, but rather the environments.  Sereni notes that the
$k$-limited CFA hierarchy ``present[s] different characteristics, in
particular in the aspects of precision and complexity''
\cite{sereni-icfp07}, however no complexity characterization is given.

\section{Type Inference and Quantifier Elimination}
\label{sec-quantifier-elim}

Earlier work on the complexity of compile-time type inference is a
precursor of the research insights described here, and naturally so,
since type inference is a kind of static analysis
\cite{mairson-popl90,henglein-unpub90,henglein-mairson-popl91,mairson-jfp04}.
The decidability of type inference depends on the making of
approximations, necessarily rejecting programs without type errors; in
simply-typed $\lambda$-calculus, for instance, all occurrences of a
variable must have the same type.  (The same is, in effect, also true
for ML, modulo the finite development implicit in {\tt let}-bindings.)
The type constraints on these multiple occurrences are solved by
first-order unification.

As a consequence, we can understand the inherent complexity of type
inference by analyzing the expressive power of {\em linear} terms,
where no such constraints exist, since linear terms are always
simply-typable.  In these cases, type inference is synonymous with
normalization.\footnote{An aberrant case of this phenomenon is
  examined by \citet{neergaard-mairson-icfp04}, which analyzed a type
  system where normalization and type inference are synonymous in {\em
    every} case.  The tractability of type inference thus implied a
  certain inexpressiveness of the language.}  This observation
motivates the analysis of type inference described by
\citet{mairson-popl90, mairson-jfp04}.

Compared to flow analysis, type reconstruction has enjoyed a much more
thorough complexity analysis.

%%%%%%%%%%%%%%%%%%%%%%%%%%%%%%%%%%%%

A key observation about the type inference problem for simply typed
$\lambda$-terms is that, when the term is linear (every bound variable
occurs exactly once), the most general type and normal form are
isomorphic
\cite{hindley-tcs89,hirokawa-tacs91,henglein-mairson-popl91,mairson-jfp04}. So
given a linear term in normal form, we can construct its most general
type (no surprise there), but conversely, when given a most general
type, we can construct the normal form of all terms with that type.

This insight becomes the key ingredient in proving the lower bound
complexity of simple-type inference---when the program is linear,
static analysis is effectively ``running'' the program.  Lower bounds,
then, can be obtained by simply hacking within the linear
$\lambda$-calculus.

\begin{quotation}%
{\bf Aside:} 
The normal form of a linear program can be ``read back'' from its most
general type in the following way: given a type
$\sigma_1\rightarrow\sigma_2\rightarrow \dots \rightarrow \sigma_k
\rightarrow \alpha$, where $\alpha$ is a type variable, we can
conclude the normal form has the shape $\lambda x_1.\lambda x_2.\dots
\lambda x_k.e$.
Since the term is linear, and the type is most general, every type
variable occurs exactly twice: once positively and once negatively.
Furthermore, there exists a
unique $\sigma_i \equiv \tau_1\rightarrow\tau_2\rightarrow \dots
\rightarrow \tau_m \rightarrow \alpha$, so $x_i$ must be the head
variable of the normal form, i.e., we now know: $\lambda x_1.\lambda
x_2.\dots \lambda x_k.x_1e'$, and $x_i$ is applied to $m$ arguments,
each with type $\tau_1,\dots,\tau_m$, respectively. 
But now, by induction, we can recursively construct the normal forms
of the arguments.  The base case occurs when we get to a base type (a
type variable); here the term is just the occurrence of the
$\lambda$-bound variable that has this (unique) type.  In other words,
a negative type-variable occurrence marks a $\lambda$-binding, while
the corresponding positive type-variable occurrence marks the single
occurrence of the bound variable.  The rest of the term structure is
determined in a syntax-directed way by the arrow structure of the
type.
\end{quotation}

%%%%%%%%%%%%%%%%%%%%%%%%%%%%%%%%%%%%

It has been known for a long time that type reconstruction for the
simply typed $\lambda$-calculus is decidable
\cite{curry-dialectica69,hindley-ams69}, i.e.~it is decidable whether
a term of the untyped $\lambda$-calculus is the image under
type-erasing of a term of the simply typed
$\lambda$-calculus.\footnote{See \citet{tiuryn-mfcs90} for a survey of
  type inference problems, cited in \citet{cardone-hindley-06}.}
\citet{wand-fi87} gave the first direct reduction to the unification
problem
\cite{herbrand,robinson-jacm65,dwork-et-al-jlp84,kanellakis-etal-robinson}.
\citet{henglein-fplca91,henglein92d} used unification to develop
efficient type inference for binding time analysis and flow analysis,
respectively.  This work directly inspired the widely influential
\citet{steensgaard-popl96} algorithm.\footnote{See
  \autoref{sec-pointer-analysis} for more on the relation of pointer
  analysis and flow analysis.}

A lower bound on the complexity of type inference can often be
leveraged by the combinatorial power behind a quantifier elimination
procedure \cite{mairson-jfp92}.  These procedures are syntactic
transformations that map programs into potentially larger programs
that can been typed in a simpler, quantifier-free setting.

As an example, consider the case of ML polymorphism.  The universal
quantification introduced by {\tt let}-bound values can be eliminated
by reducing all {\tt let}-redexes.  The residual program is
simply-typable if, and only if, the original program is ML-typable.

This is embodied in the following inference rule:\footnote{In the
  survey, {\em Type systems for programming languages},
  \citet{mitchell-volb} attributes this observation to Albert Meyer.
  \citet[page~122]{henglein-mairson-popl91} point out in a footnote that
  it also appears in the thesis of \citet{damas-phd}, and is the
  subject of a question on the 1985 postgraduate examination in
  computing at Edinburgh University.}
\begin{displaymath}
\inferrule{\Gamma \vdash M : \tau_0\\ \Gamma \vdash [M/x]N:\tau_1}
          {\Gamma \vdash \mbox{\tt{let} } x = M \mbox{ \tt{in} } N:\tau_1}
\end{displaymath}

The residual may be exponentially larger due to nested {\tt let}
expressions that must all be eliminated.  From a Curry-Howard
perspective, this can be seen as a form of cut-elimination.  From a
computational perspective, this can be seen as a bounded running of
the program at compile time.  From a software engineering perspective,
this can be seen as code-reuse---the ML-type inference problem has
been reduced to the simple-type inference problem, and thus to
first-order unification.  But the price is that an exponential amount
of work may now be required.

Full polymorphism is undecidable, but ML offers a limit form of
outermost universal quantification.  But this restriction relegates
polymorphic functions to a second-class citizenship, so in particular,
functions passed as arguments to functions (a staple of higher-order
programming) can only be used monomorphically.

Intersection types restore first-class polymorphism by offering a
finite form of explicit quantification over simple types.  The type
$\tau_1 \wedge \tau_2$ is used for a term that is typable as both
$\tau_1$ and $\tau_2$.  This can formalized as the following inference
rule for $\wedge$:\footnote{This presentation closely follows the
  informal presentation of intersection types in Chapter 4 of
  \citet{neergaard-phd}.}
\begin{displaymath}
\inferrule{\Gamma_1 \vdash M:\tau_1\\ \Gamma_2 \vdash M:\tau_2}
          {\Gamma_1 \wedge \Gamma_2 \vdash M:\tau_1\wedge\tau_2}
\end{displaymath}
where $\wedge$ is lifted to environments in a straightforward way.
Notice that this allows expressions such as,
\begin{displaymath}
(\lambda f.\lambda z.z(f\; 2)(f\; \mathbf{false}))\;(\lambda x.x),
\end{displaymath}
to be typed where $x$ has type $\mathbf{int} \wedge \mathbf{bool}$.

The inference rule, as stated, breaks syntax-directed inference.
\citet{vanbakel-tcs92} observed that by limiting the rule to the
arguments of function application, syntax-direction can be recovered
without changing the set of typable terms (although some terms will
have fewer typings).  Such systems are called {\em strict
  intersections} since the $\wedge$ can occur only on the left of a
function type.

The finite $\wedge$-quantifiers of strict intersections too have an
elimination procedure, which can be understood as a program
transformation that eliminates $\wedge$-quant\-ification by {\em rank}.
A type is rank $r$ if there are no occurrences of $\wedge$ to the left
of $r$ occurrences of an arrow.  The highest rank intersections
can be eliminated by performing a {\em minimal complete development}.

Every strongly normalizing term has an intersection type, so type
inference in general is undecidable.  However, decidable fragments can
be regained by a standard approach of applying a {\em
  rank} restriction, limiting the depth of $\wedge$ to the left of a
function type.  

By bounding the rank, inference becomes decidable; if the rank is
bound at $k$, $k$ developments suffice to eliminate all intersections.
The residual program is simply-typable if, and only if, the original
program is rank-$k$ intersection typable.  Since each development can
cause the program to grow by an exponential factor, iteratively
performing $k$-MCD's results in an elementary lower bound
\cite{kfoury-etal-icfp99,neergaard-mairson-icfp04}.

The special case of rank-2 intersection types have proved to be an
important case with applications to modular flow analysis, dead-code
elimination, and typing polymorphic recursion, local definitions,
conditionals and pattern matching \cite{damiani-prost-types96,
  damiani-toplas03, banerjee-jensen-mscs03, damiani-fi07}.

% ``Quadratic bound'' on ML-type reconstruction \cite{leivant-popl83}.

% ML: \citet{mitchell-89}.  \citet{kanellakis-mitchell-popl89} observed
% that {\tt let} can be used to compose functions an exponential number
% of times with a polynomial-length formula.  \cite{mairson-popl90}.
% \cite{kanellakis-mairson-mitchell-90}. \cite{henglein-unpub90}.
% \cite{henglein-mairson-popl91}.

System F, the polymorphic typed $\lambda$-calculus
\cite{reynolds-pcp74,girard-proofs-and-types}, has an undecidable
Curry-style inference problem \cite{wells-apal99}.  Partial inference
in a Church-style system is investigated by
\citet{boehm-sfcs85,pfenning-fi93} and Pfenning's result shows even
partial inference for a simple predicative fragment is undecidable.

The quantifier-elimination approach to proving lower bounds was
extended to System F$_\omega$ by \citet{henglein-mairson-popl91}.
They prove a sequence of lower bounds on recognizing the System
F$_k$-typable terms, where the bound for F$_{k+1}$ is exponentially
larger than that for F$_k$.  This is analogous to intersection
quantifier elimination via complete developments at the term level.
The essence of \citet{henglein-mairson-popl91} is to compute
developments at the kind level to shift from System F$_{k+1}$ to
System F$_k$ typability.  This technique led to lower bounds on
System F$_i$ and the non-elementary bound on System F$_\omega$
\cite{henglein-mairson-popl91}.  \citet{urzyczyn-mscs97} showed
Curry-style inference for System F$_\omega$ is undecidable.

There are some interesting open complexity problems in the realm of
type inference and quantifier elimination.  Bounded polymorphic
recursion has recently been investigated \cite{comini-etal-sas08}, and
is decidable but with unknown complexity bounds, nor quantifier
elimination procedures.
Typed Scheme \cite{tobin-hochstadt-felleisen-popl08}, uses explicit
annotations, but with partial inference and flow sensitivity.  It
includes intersection rules for function types.  Complexity bounds on
type checking and partial inference are unknown.

The simple algorithm of \citet{wand-fi87}, which generates constraints
for type reconstruction, can also be seen as compiler for the linear
$\lambda$-calculus.  It compiles a linear term into a ``machine
language'' of first-order constraints of the form $a = b$ and $c =
d\rightarrow e$.  This machine language is the computational analog of
logic's own low-level machine language for first-order propositional
logic, the machine-oriented logic of \citet{robinson-jacm65}.

Unifying these constraints effectively runs the machine language,
evaluating the original program, producing an answer in the guise of a
solved form of the type, which is isomorphic to the normal form of the
program.

Viewed from this perspective, this is an instance of
normal\-ization-by-eval\-uation for the linear $\lambda$-calculus.  A
linear term is mapped into the domain of first-order logic, where
unification is used to evaluate to a canonical solved form, which can
be mapped to the normal form of the term.  Constraint-based
formulations of monovariant flow analyses analogously can be seen as
instances of {\em weak} normal\-ization-by-eval\-uation functions for
the linear $\lambda$-calculus.

% It should be noted that all linear terms have types
% \cite{hindley-tcs89}, so this compiler is complete.

% More complexity type inference papers: \cite{hoang-mitchell-popl95}.
\chapter{Conclusions and Perspective}

\section{Contributions}

Flow analysis is a fundamental static analysis of functional,
object-oriented, and other higher-order programming languages; it is a
ubiquitous and much-studied component of compiler technology with
nearly thirty years of research on the topic.
This dissertation has investigated the computational complexity of
flow analysis in higher-order programming languages, yielding novel
insights into the fundamental limitations on the cost of performing
flow analysis.

Monovariant flow analysis, such as 0CFA, is complete for polynomial
time.
Moreover, many further approximations to 0CFA from the literature,
such as Henglein's simple closure analysis, remain complete for
polynomial time.
These theorems rely on the fact that when a program is linear (each
bound variable occurs exactly once), the analysis makes no
approximation; abstract and concrete interpretation coincide.
More generally, we conjecture {\em any} abstract and concrete
interpretation will have some sublanguage of coincidence, and this
sublanguage may be useful in proving lower bounds.

The linear $\lambda$-calculus has been identified as an important
language subset to study in order to understand flow analysis.
Linearity is an equalizer among variants of static analysis, and a
powerful tool in proving lower bounds.
Analysis of linear programs coincide under both equality and
inclusion-based flow constraints, and moreover, concrete and abstract
interpretation coincide for this core language.
The inherently sequential nature of flow analysis can be understood as
a consequence of a lack of abstraction on this language subset.

Since linearity plays such a fruitful role in the study of program
analysis, we developed connections with linear logic and the
technology of sharing graphs.
Monovariant analysis can be formulated graphically, and the technology
of graph reduction and optimal evaluation can be applied to flow
analysis.
The explicit control representation of sharing graphs makes it easy to
extend flow analysis to languages with first-class control.

Simply-typed, $\eta$-expanded programs have a potentially simpler 0CFA
problem, which is complete for logarithmic space.
This discovery is based on analogies with proof normalization for
multiplicative linear logic with {\em atomic} axioms.

Shivers' polyvariant $k$CFA, for any $k>0$, is complete for
deterministic exponential time.
This theorem validates empirical observations that such control flow
analysis is intractable. 
A fairly straightforward calculation shows that $k$CFA can be computed
in exponential time.  We show that the naive algorithm is essentially
the best one.  There is, in the worst case---and plausibly, in
practice---no way to tame the cost of the analysis.  Exponential time
is required.

Collectively, these results provide general insight into the
complexity of abstract interpretation and program analysis.

\section{Future Work}

We end by outlining some new directions and open problems worth
pursuing, in approximately ascending order of ambition and import.

\subsection{Completing the Pointer Analysis Complexity Story}

Compared with flow analysis, pointer analysis has received a much more
thorough complexity investigation.  A series of important refinements
have been made by \citet{landi-ryder-popl91,
  landi-phd92,landi-loplas92,choi-etal-popl93,
  horwitz-toplas97,muth-debray-popl00,chaterjee-etal-tse01,
  chakaravarthy-popl03}, yet open problems persist.
\citet{chakaravarthy-popl03} leaves open the lower bound on the
complexity of pointer analysis with well-defined types with less than
two levels of dereference.  We believe our insights into linearity and
circuit construction can lead to an answer to this remaining problem.

\subsection{Polyvariant, Polynomial Flow Analyses}

To echo the remark of \citet{bravenboer-smaragdakis-oopsla09}, only
adapted to the setting of flow analysis rather than pointer analysis,
there is an ongoing search for polyvariant, or context-sensitive,
analyses that provide precise flow information without causing massive
redundant computation.
There has been important work in this area
\cite{jagannathan-weeks-popl95,nielson-nielson-popl97}, but the
landscape of tractable, context-sensitive flow analyses is mostly open
and in need of development.

The ingredients, detailed in \autoref{chapter-kcfa}, that combine to
make $k$CFA hard, when $k > 0$, should provide guidance in designing new
abstractions that avoid computationally expensive components of
analysis.
A lesson learned has been that {\em closures}, as they exist when
$k>0$, result in an exponential value space that can be harnessed for
the \exptime\ lower-bound construction.
It should be possible to design alternative closure abstractions while
remaining both polyvariant and polynomial (more below).

\subsection{An Expressive Hierarchy of Flow Analyses}

From the perspective of computational complexity, the $k$CFA hierarchy
is flat (for any fixed $k$, $k$CFA is in \exptime; see
\autoref{sec-kcfa-in-exptime}).  On the other hand, there are far more
powerful analyses such as those of \citet{burn-hankin-abramsky} and
\citet{mossin-njc98}.  How can we systematically bridge the gap
between these analyses to obtain a real expressivity hierarchy?

Flow analyses based on rank-bounded intersection types offers one
approach.  
It should also be possible to design such analyses by composing
notions of precise but bounded computation---such as partial
evaluation or a series of complete developments---followed by course
analysis of residual programs.  
The idea is to stage analysis into two phases: the first eliminates
the need for polyvariance in analysis by transforming the original
program into an equivalent, potentially larger, residual program.
The subsequent stage performs a course (monovariant) analysis of the
residual program.
By staging the analysis in this manner---first computing a precise but
bounded program evaluation, {\em then} an imprecise evaluation
approximation---the ``ever-worsening feedback loop''
\cite{might-shivers-icfp06} is avoided.
By using a sufficiently powerful notion of bounded evaluation, it
should be possible to construct flow analyses that form a true
hierarchy from a complexity perspective.
By using a sufficiently weak notion of bounded evaluation, it should
be possible to construct flow analyses that are arbitrarily
polyvariant, but computable in polynomial time.

\subsection{Truly Subcubic Inclusion-Based Flow Analysis}

This dissertation has focused on lower bounds, however recent upper
bound improvements have been made on the ``cubic bottleneck'' of
inclusion-based flow analyses such as 0CFA
\cite{midtgaard-vanhorn-tr09}.  These results have shown known set
compression techniques can be applied to obtain direct 0CFA algorithms
that run in $O(n^3 / \log n)$ time on a unit cost random-access memory
model machine.  While these results do provide a logarithmic
improvement, it is natural to wonder if there is a $O(n^c)$ algorithm
for 0CFA and related analyses, where $c < 3$.

At the same time, there have been recent algorithmic breakthroughs on
the all-pairs shortest path problem resulting in truly subcubic
algorithms.  Perhaps the graphical formulation of flow analysis from
\autoref{chap:linear-logic} can be adapted to exploit these
breakthroughs.

\subsection{Toward a Fundamental Theorem of Static Analysis}

A theorem due to \citet{statman79} says this: let {\bf P} be a
property of simply-typed $\lambda$-terms that we would like to detect
by static analysis, where {\bf P} is invariant under reduction
(normalization), and is computable in elementary time (polynomial, or
exponential, or doubly-exponential, or\dots).  Then {\bf P} is a {\em
  trivial} property: for any type $\tau$, {\bf P} is satisfied by {\em
  all} or {\em none} of the programs of type $\tau$.
\citet{henglein-mairson-popl91} have complemented these results,
showing that if a property is invariant under $\beta$-reduction for a
class of programs that can encode all Turing Machines solving problems
of complexity class {\sc{f}} using reductions from complexity class
{\sc{g}}, then any superset is either {\sc{f}}-complete or trivial.
Simple typability has this property for linear and linear affine
$\lambda$-terms \cite{henglein-mairson-popl91,mairson-jfp04}, and
these terms are sufficient to code all polynomial-time Turing
Machines.

We would like to prove some analogs of these theorems, with or without
the typing condition, but weakening the condition of ``invariant under
reduction'' to ``invariant under abstract interpretation.''

%\appendix
%\include{appendix-hacker}
%\include{appendix-orphans}

% Cites all sources -- for debugging.
% \nocite{*}

% http://www.tex.ac.uk/cgi-bin/texfaq2html?label=tocbibind
\cleardoublepage
\phantomsection
\addcontentsline{toc}{chapter}{Bibliography}
\bibliography{bibliography}

%\addcontentsline{toc}{chapter}{Postscript}
%\include{postscript}

\cleardoublepage
\phantomsection
\addcontentsline{toc}{chapter}{Colophon}
\chapter*{Colophon}
\label{colophon}
\thispagestyle{empty}

This dissertation was produced using Free Software on a digital
computer.  The document was composed in GNU Emacs and typeset using
\LaTeX\ and the {\em Brandeis dissertation} class by Peter M{\o}ller
Neergaard.  The Times font family is used for text and the Computer
Modern font family, designed by Donald E.~Knuth, is used for
mathematics.  Some figures were typeset using the \MP\ language by
John D.~Hobby and sequent calculus proofs were typeset using Didier
R\'emy's {\em Math Paragraph for Typesetting Inference Rules}.

This dissertation is Free Software.  It is released under the terms of
the Academic Free License version 3.0.  Source code is available:
\begin{center}
\url{http://svn.lambda-calcul.us/dissertation/}
\end{center}

Machine wash cold in commercial size, front loading machine, gentle
cycle, mild powder detergent.  Rinse thoroughly.  Tumble dry low in
large capacity dryer.  Do not iron.  Do not dry clean.  Do not bleach.
Secure all velcro closures.
\begin{center}
{\Large \ShortThirty\ \Tumbler\ \NoIroning\ \NoChemicalCleaning\
\NoBleech\ $\lambda^{\mbox{\small\Lineload}}$}
\end{center}

\end{document}